\documentclass{article}
\usepackage{amsmath, amsthm}
\usepackage{amssymb}
\usepackage{bbm}
\newtheorem{proposition}{Proposition}
\newtheorem{theorem}{Theorem}
\newtheorem{remark}{Remark}
\newtheorem{definition}{Definition}
\newtheorem{maintheorem}{Main Theorem}

\usepackage[margin=1in,dvips]{geometry}

\newcommand{\nabb}{{\nabla} \mkern-13mu /\,}
\newcommand{\lapp}{{\Delta} \mkern-13mu /\,}

\title{The Null Condition and Global Existence for Nonlinear Wave Equations on Slowly Rotating Kerr Spacetimes}
\author{Jonathan Luk}

\begin{document}

\maketitle

\begin{abstract}
We study a semilinear equation with derivatives satisfying a null condition on slowly rotating Kerr spacetimes. We prove that given sufficiently small initial data, the solution exists globally in time and decays with a quantitative rate to the trivial solution. The proof uses the robust vector field method. It makes use of the decay properties of the linear wave equation on Kerr spacetime, in particular the improved decay rates in the region $\{r\leq \frac{t}{4}\}$.
\end{abstract}

\section{Introduction}

In this paper, we consider the global existence for small data for a semilinear equation with null condition on a Kerr spacetime. Kerr spacetimes are stationary axisymmetric asymptotically flat black hole solutions to the vacuum Einstein equations
$$R_{\mu\nu}=0$$
in $3+1$ dimensions. They are parametrized by two parameters $\left(M, a\right)$, representing respectively the mass and the angular momentum of a black hole. We study semilinear equations on a Kerr spacetime with $a\ll M$ of the form
$$\Box_{g_K}\Phi=F(D \Phi),$$
where $\Box_{g_K}$ is the Laplace-Beltrami operator for the Kerr metric $g_K$, and $F$ denotes nonlinear terms that are at least quadratic and satisfy the null condition that we will define in Section \ref{theorem}.

The corresponding problem on Minkowski spacetime has been well studied. In 4+1 or higher dimension, the decay of the linear wave equation is sufficiently fast for one to prove global existence of small data for nonlinear wave equations with any quadratic nonlinearity \cite{Kl}. However, in 3+1 dimensions, which is also the dimension of physical relevance, the decay rate is only sufficient to prove the almost global existence of solutions \cite{JoK}. Indeed, a counter-example is known \cite{John} for the equation

$$\Box_m\Phi=(\partial_t\Phi)^2.$$
Nevertheless, if the quadratic nonlinearity satisfies the null condition defined by Klainerman, it has been proved independently by Christodoulou \cite{Chr} and Klainerman \cite{Knull} that any solutions to sufficiently small initial data are global in time. There have been an extensive literature on extensions and variations of the original results, including the cases of the multiple-speed system and the exterior domains (\cite{ST}, \cite{Sogge}, \cite{MNS}, \cite{MSo}).

The decay rate of the solutions to the linear wave equation on Kerr spacetimes with $a\ll M$ have been proved in \cite{DRL}, \cite{AB}, \cite{Ta} and \cite{LKerr}. The known decay outside the set $\{ct^*\leq r\leq Ct^*\}$ is sufficiently strong and the proof of them (in \cite{DRL}, \cite{AB} and \cite{LKerr}) is sufficiently robust that one expects the main obstacle from proving a small data global existence result (if it indeed holds) would come from quantities in the set $\{ct^*\leq r\leq Ct^*\}$. This set, however, approaches the same set in Minkowski spacetime as $t^*\to\infty$ due to the asymptotic flatness of Kerr spacetimes. Therefore, one expects that with a null condition similar to that on the Minkowski spacetime, a similar global existence result holds. Indeed, we have (see the precise version in Section \ref{theorem})

\begin{maintheorem}
Consider $\Box_{g_K}\Phi= F(D\Phi)$ where $F$ satisfies the null condition (see Section \ref{theorem}). Then for any initial data that are sufficiently small, the solution exists globally in time.
\end{maintheorem}

Our major motivation in studying the null condition on a Kerr spacetime is the problem of the stability of the Kerr spacetime. It is conjectured that Kerr spacetimes are stable. In the framework of the initial value problem, the stability of Kerr would mean that for any solution to the vacuum Einstein equations with initial data close to the initial data of a Kerr spacetime, its maximal Cauchy development has an exterior region that approaches a nearby, but possibly different, Kerr spacetime. In the case of the Minkowski spacetime, the null condition has served as a good model problem for the study of the stability of the Minkowski spacetime. We hope that this work will find relevance to the problem of the stability of the Kerr spacetime.

\subsection{Some Related Known Results}

We turn to some relevant work on linear and nonlinear scalar wave equations on Kerr spacetimes. The decay of solutions to the linear wave equation on Kerr spacetimes has received considerable attention. We mention some results on Kerr spacetimes with $a>0$ here and refer the readers to \cite{DRL}, \cite{LS} for references on the corresponding problem on Schwarzschild spacetimes. There has been a large literature on the mode stability and non-quantitative decay of azimuthal modes (See for example \cite{PT}, \cite{HW}, \cite{Wh}, \cite{FKSY}, \cite{FKSY2} and references in \cite{DRL}). The first global result for the Cauchy problem was obtained by Dafermos-Rodnianski in \cite{DRK}, in which they proved that for a class of small, axisymmetric, stationary perturbations of Schwarzschild spacetime, which include Kerr spacetimes that rotate sufficiently slowly, solutions to the wave equation are uniformly bounded. Similar results were obtained later using an integrated decay estimate on slowly rotating Kerr spacetimes by Tataru-Tohaneanu \cite{TT}. Using the integrated decay estimate, Tohaneanu also proved Strichartz estimates \cite{To}.

Decay for general solutions to the wave equation on sufficiently slowly rotating Kerr spacetimes was first proved by Dafermos-Rodnianski \cite{DRL} with a quantitative rate of $|\Phi|\leq C(t^*)^{-1+Ca}$. A similar result was later obtained by \cite{AB} using a physical space construction to obtain an integrated decay estimate. In all of \cite{TT}, \cite{DRL} and \cite{AB}, the integrated decay estimate is proved and plays an important role. All proofs of such estimates rely heavily on the separability of the wave equation, or equivalently, the existence of a Killing tensor on Kerr spacetime. In a recent work \cite{DRNPS}, Dafermos-Rodnianski prove the non-degenerate energy decay and the pointwise decay assuming the integrated local energy decay estimate and boundedness for the wave equation on an asymptotically flat spacetime. Their work shows a decay rate of $|\Phi|\leq Ct^{-1}$ and improves the rates in \cite{DRL} and \cite{AB}. In a similar framework, but assuming in addition exact stationarity, Tataru \cite{Ta} proved a local decay rate of $(t^*)^{-3}$ using Fourier-analytic methods. This applies in particular to sufficiently slowly rotating Kerr spacetimes. Dafermos and Rodnianski have recently announced a proof for the decay of solutions to the wave equation on the full range of sub-extremal Kerr spacetimes $a< M$. 

For nonlinear equations, global existence for the equation with power nonlinearity $\Box_{g_k}\Phi=\pm |\Phi|^p \Phi$ was initiated in \cite{NiS} and \cite{NiK}, in which the large data subcritical defocusing case of $p=2$ is studied. Later, there have been much work on the small data problem in which the sign of the nonlinearity is not important, and that the dispersive properties of the linear equation plays a crucial role. Global existence was proved for small radial data for $p>3$ on Reissner-Nordstrom spacetime \cite{DRNL} and for general small data vanishing on the bifurcate sphere for $p>2$ \cite{BSt} on Schwarzschild spacetime. Global existence was also proved for $p=4$ on Schwarzschild spacetime with general data that has small non-degenerate energy \cite{MMTT}. This was extended to the case of sufficiently slowly rotating Kerr spacetime in \cite{To}. Counterexample is known for the case $0<p<\sqrt{2}$ \cite{CG}. To our knowledge, the present work is the first work on semilinear equations with derivatives on black hole spacetimes.

\subsection{The Statement of the Main Theorem}\label{theorem}

Before introducing the null condition and stating the precise version of the Main Theorem, we briefly introduce the necessary concepts and notations on Kerr geometry and the vector field method. See \cite{LKerr} for more details.

The Kerr metric in the Boyer-Lindquist coordinates takes the following form:
\begin{equation}\label{kerrmetric}
\begin{split}
g_K=&-\left(1-\frac{2M}{r\left(1+\frac{a^2 \cos ^2{\theta}}{r^2}\right)}\right)dt^2+\frac{1+\frac{a^2\cos ^2 \theta}{r^2}}{1-\frac{2M}{r}+\frac{a^2}{r^2}}dr^2+r^2\left(1+\frac{a^2\cos ^2 \theta}{r^2}\right)d\theta^2 \\
&+r^2\left(1+\frac{a^2}{r^2}+\left(\frac{2M}{r}\right)\frac{a^2\sin ^2\theta}{r^2\left(1+\frac{a^2\cos ^2\theta}{r^2}\right)}\right)\sin ^2\theta d\phi^2
-4M \frac{a \sin ^2\theta}{r\left(1+\frac{a^2\cos ^2\theta}{r^2}\right)}dtd\phi.
\end{split}
\end{equation}
Let $r_+$ be the larger root of $\Delta=r^2-2Mr+a^2$. $r=r_+$ is the event horizon. In this paper, we will use the coordinate system $(t^*,r,\theta,\phi^*)$ defined by 
$$t^*=t+\chi(r)h(r),\quad \mbox{where  }\frac{dh(r)}{dr}=\frac{2Mr}{r^2-2Mr+a^2},$$
$$\phi^*=\phi+\chi(r)P(r),\quad \mbox{where  }\frac{dP(r)}{dr}=\frac{a}{r^2-2Mr+a^2},$$
where
$$\chi(r)=\left\{\begin{array}{clcr}1&r\le r^-_Y-\frac{r^-_Y-r_+}{2}\\0&r\ge r^-_Y-\frac{r^-_Y-r_+}{4}\end{array}\right.,$$
where $r_+$, as above, is the larger root of $\Delta=r^2-2Mr+a^2$ and $r^-_Y>r_+$ is a fixed constant very close to $r_+$, the value of which can be determined from the proof of the energy estimates in \cite{LKerr}. Following the notation in \cite{LKerr}, we will use $t^*=\tau$ to denote the $t^*$ slice on which we want to prove estimates and $t^*=\tau_0$ to denote the $t^*$ slice on which the initial data is posed.

In \cite{LKerr}, following \cite{DRK}, various quantities are defined via an explicit identification of the Kerr spacetime with the corresponding Schwarzschild spacetime with the same mass. We recall the identification:
$$r_S^2-2Mr_S=r^2-2Mr+a^2,$$
$$t_S+\chi(r_S)2M\log\left(r_S-2M\right)=t^*,$$
$$\theta_S=\theta,$$
$$\phi_S=\phi^*,$$
where $\chi$ is as above. 

Define 
$$r^*_S=r_S+2M\log(r_S-2M)-3M-2M\log M,$$
$$\mu=\frac{2M}{r_S},$$
$$u=\frac{1}{2}(t_S-r^*_S),$$
$$v=\frac{1}{2}(t_S+r^*_S).$$
We note that the variable $u$ will also be used to quantify decay.

We define in coordinates
$$\underline{L}=\partial_u\mbox{, in the $(u,v,\theta_S,\phi_S)$ coordinates},$$
$$L=2\partial_{t^*}+\chi(r)\frac{a}{Mr_+}\partial_{\phi^*}-\underline{L}$$

We can now define the ``good'' and ``bad'' derivatives. Define
$$\nabb\in\{\frac{1}{r}\partial_{\theta},\frac{1}{r}\partial_\phi\},$$
$$\overline{D}\in\{L,\frac{1}{r}\partial_{\theta},\frac{1}{r}\partial_\phi\},$$
$${D}\in\{\frac{1}{1-\mu}\underline{L}, L,\frac{1}{r}\partial_{\theta},\frac{1}{r}\partial_\phi\}.$$
Notice that $D$ spans the whole tangent space and that we always have $[D,\partial_{t^*}]=0$.

We now define the null condition. On Minkowski spacetime, the classical null condition can be defined geometrically by requiring the nonlinearity to have the form
$$A^{\mu\nu}\partial_{\mu}\Phi\partial_{\nu}\Phi,$$
where $A$ satisfies $A^{\mu\nu}\xi_\mu\xi_\nu=0$ whenever $\xi$ is null. On Kerr spacetime, we would like to define a notion of the null condition that includes this geometric notion. This is also because many physically relevant semilinear equations satisfy this condition. On the other hand, in order to prove the global existence result, we need to use the vector fields that capture the good derivative. We would therefore like to define the null condition using the vector fields defined in \cite{DRL}, \cite{LKerr}, i.e., using $D$ and $\overline{D}$. In particular, we want the nonlinearity to have at least one good, i.e., $\overline{D}$, derivative. This on its own is however inconsistent with the geometric null condition. We therefore allow a term in the quadratic nonlinearity that does not have a good derivative but decays in $r$.

\begin{definition}
Consider the nonlinearity $F(\Phi,D\Phi,t^*,r,\theta,\phi^*).$ We say that $F$ satisfies the null condition if
$$F=\Lambda_0(\Phi,t^*,r,\theta,\phi^*)D\Phi\overline{D}\Phi+\Lambda_1(\Phi,t^*,r,\theta,\phi^*)D\Phi{D}\Phi+\mathcal C(\Phi,D\Phi,t^*,r,\theta,\phi^*),$$
where $$ |D_\Phi^{i_1}\partial_{t^*}^{i_2}\partial_r^{i_3}\partial_\theta^{i_4}\partial_{\phi^*}^{i_5}\Lambda_j|\leq C(t^*)^{-i_2}r^{-i_3} \quad\mbox{for } i_1+i_2+i_3+i_4+i_5\leq 16, j=0,1$$
and
$$|D_\Phi^{i_1}\partial_{t^*}^{i_2}\partial_r^{i_3}\partial_\theta^{i_4}\partial_{\phi^*}^{i_5}\Lambda_1|\leq C(t^*)^{-i_2}r^{-1-i_3} \quad\mbox{for } i_1+i_2+i_3+i_4+i_5\leq 16\mbox{ and }r\geq\frac{9t^*}{10}$$
 and $\mathcal C$ denotes a polynomial that is at least cubic in $D\Phi$ (with coefficients in $\Phi,t^*,r,\theta,\phi^*$) satisfying
$$|D_\Phi^{i_1}\partial_{t^*}^{i_2}\partial_r^{i_3}\partial_\theta^{i_4}\partial_{\phi^*}^{i_5}\mathcal C|\leq C(t^*)^{-i_2}r^{-i_3}\sum_{s= 3}^S|D\Phi|^s \quad\mbox{for } i_1+i_2+i_3+i_4+i_5\leq 16$$
\end{definition}

\begin{remark}
The null condition is a special structure for the quadratic nonlinearity. We note that in our case, the restriction is necessary only for $r\geq \frac{9t^*}{10}$. Moreover, higher order terms should give better estimates and do not need any special structure.
\end{remark}
Under this definition of the null condition, global existence holds for small data. Moreover, the solution $\Phi$ satisfies pointwise decay estimates. In order to appropriately describe smallness, we introduce the language of compatible currents. Define the energy-momentum tensor
$$T_{\mu\nu}=\partial_\mu\Phi\partial_\nu\Phi-\frac{1}{2}g_{\mu\nu}\partial^\alpha\Phi\partial_\alpha\Phi.$$
By virtue of the wave equation, $T$ is divergence-free,
$$\nabla^\mu T_{\mu\nu}=0.$$
For a vector field $V$, define the compatible currents
$$J^V_{\mu}\left(\Phi\right)=V^{\nu}T_{\mu\nu}\left(\Phi\right),$$
$$K^V\left(\Phi\right)=\pi^V_{\mu\nu}T^{\mu\nu}\left(\Phi\right),$$
where $\pi^V_{\mu\nu}$ is the deformation tensor defined by
$$\pi^V_{\mu\nu}=\frac{1}{2}\left(\nabla_{\mu}V_{\nu}+\nabla_{\nu}V_{\mu}\right).$$
In particular, $K^V\left(\Phi\right)=\pi^V_{\mu\nu}=0$ if $V$ is Killing.
Since the energy-momentum tensor is divergence-free,
$$\nabla^{\mu}J^{V}_{\mu}\left(\Phi\right)=K^V\left(\Phi\right).$$
We also define the modified currents $$J^{V,w}_{\mu}\left(\Phi\right)=J^{V}_{\mu}\left(\Phi\right)+\frac{1}{8}\left(w\partial_{\mu}\Phi^2-\partial_{\mu}w\Phi^2\right),$$
$$K^{V,w}\left(\Phi\right)=K^V\left(\Phi\right)+\frac{1}{4}w\partial^{\nu}\Phi\partial_{\nu}\Phi-\frac{1}{8}\Box_g w\Phi^2.$$
Then $$\nabla^{\mu}J^{V,w}_{\mu}\left(\Phi\right)=K^{V,w}\left(\Phi\right).$$
In \cite{LKerr}, we have used the currents corresponding to $N$ and $(Z,w^Z)$ defined by
$$N=\partial_{t^*}+e\left(y_1\left(r\right)\hat{Y}+y_2\left(r\right)\hat{V}\right)$$
$$Z=u^2\underline{L}+v^2L,$$
$$w^Z=\frac{8tr^*_S\left(1-\frac{2M}{r_S}\right)}{r},$$
where
$$y_1\left(r\right)=1+\frac{1}{(\log(r-r_+))^3},$$
$$y_2\left(r\right)=\frac{1}{(\log(r-r_+))^3},$$
$$\Delta=r^2-2Mr+a^2,$$
and $\hat{Y}$ and $\hat{V}$ are compactly supported vector fields in a neighborhood of $\{r_+\leq r\leq r^-_Y\}$ and are null in $\{r_+\leq r\leq r^-_Y\}$ and $e$ is an appropriately small constant depending only on $a$ (See \cite{LKerr}). 
Since $N$ is future-directed, we have the pointwise inequality $$J^N_\mu(\Phi) n^\mu_{\Sigma_{t^*}}\geq 0.$$ In \cite{LKerr} we have shown that there exists a constant $C$ such that 
$$\int_{\Sigma_{t^*}}J^{Z,w^Z}_\mu(\Phi) n^\mu_{\Sigma_{t^*}}+C(t^*)^2\int_{\Sigma_{t^*}\cap\{r\leq r^-_Y\}}J^{N}_\mu(\Phi) n^\mu_{\Sigma_{t^*}}\geq 0.$$
These energy quantities will be used for $\Phi$ as well as derivatives of $\Phi$. We now define the commutators that we will used. $\partial_{t^*}$ is a Killing vector field that is defined as the coordinate vector field with respect to the $(t^*,r,\theta,\phi^*)$ coordinate system. Near the event horizon, we use the commutator $\hat{Y}$ which is  compactly supported in $\{ r\leq r^+_Y\}$ (where $r^+_Y>r^-_Y$ is an explicit constant in \cite{LKerr}), null in $\{r_+\leq r\leq r^-_Y\}$ and is transverse to the event horizon (See \cite{LKerr}). $\hat{Y}$ has good positivity property that reflects the celebrated red-shift effect. In the region of large $r$, we use the commutators $\tilde{\Omega}$. Let $\Omega_i$ be a basis of vector fields of rotations in Schwarzschild spacetimes. An explicit realization can be $\Omega=\partial_\phi, \sin\phi\partial_\theta\pm \frac{\cos\phi\cos\theta}{\sin\theta}\partial_\phi$. Define $\tilde{\Omega}_i=\chi(r)\Omega_i$ to be cutoff so that it is supported in $\{r>R_\Omega\}$ and equals $\Omega_i$ for $r>R_\Omega+1$ for some large $R$. We also use the commutator $S$ that would provide an improved decay rate of the solution. It is defined as
$$S=t^*\partial_{t^*}+h(r_S)\partial_{r},$$ where $h(r_S)= \left\{\begin{array}{clcr}r_S-2M &r_S\sim 2M\\r_S^*(1-\mu )&r\ge R\end{array}\right.$, for some large $R$ and is interpolated so that it is smooth and non-negative. For the commutators, we also use the notation that 
$$\Gamma\in\{\partial_{t^*},\tilde{\Omega}\}.$$
We are now in a position to state our Main Theorem precisely.

\begin{theorem}
Consider the equation
\begin{equation}\label{equation}
 \Box_g\Phi=F(\Phi,D\Phi,t^*,r,\theta,\phi^*),
\end{equation} 
where $F$ satisfies the null condition.
There exists an $\epsilon$ such that if the initial data of $\Phi$ satisfies 
\begin{equation*}
\begin{split}
 \sum_{i+j+k=16}\int_{\Sigma_{\tau_0}} J^{Z+CN,w^Z}_\mu\left(\hat{Y}^k\partial_{t^*}^i\tilde{\Omega}^j\Phi\right) n^{\mu}_{\Sigma_{\tau_0}} +\sum_{i+j+k=16}\int_{\Sigma_{\tau_0}} J^{Z+CN,w^Z}_\mu\left(\hat{Y}^k S\partial_{t^*}^i\tilde{\Omega}^j\Phi\right) n^{\mu}_{\Sigma_{\tau_0}}\leq \epsilon
\end{split}
\end{equation*}
and 
\begin{equation*}
\begin{split}
\sum_{\ell=0}^{13} r|D^\ell\Phi(\tau_0)|+ r|D^\ell S\Phi(\tau_0)|\leq \epsilon.
\end{split}
\end{equation*}
Then $\Phi$ exists globally in time.
Moreover, for all $\eta>0$, we can take $a$ sufficiently small such that the solution $\Phi$ obeys the decay estimate
$$|\Phi|\leq C\epsilon r^{-1}u^{-\frac{1}{2}}(t^*)^\eta, |D\Phi|\leq C\epsilon r^{-1}u^{-1}(t^*)^{\eta}, |\overline{D}\Phi|\leq C\epsilon r^{-1}(t^*)^{-1+\eta}\mbox{ for }r\geq R,\mbox{ and}$$ 
$$|\Phi|\leq C_\delta\epsilon (t^*)^{-\frac{3}{2}+\eta}r^{\delta}, |D\Phi|\leq C_\delta\epsilon (t^*)^{-\frac{3}{2}+\eta}r^{-\frac{1}{2}+\delta}\mbox{ for }r\leq \frac{t^*}{4}.$$
\end{theorem}

We specialize to a particular case which resembles better the classical null condition \cite{Knull}.
\begin{theorem}
Consider the equation
\begin{equation}\label{geometric null}
 \Box_g\Phi=\Gamma(\Phi)A^{\mu\nu}\partial_{\mu}\Phi\partial_{\nu}\Phi,
\end{equation} 
where $A$ satisfies $A^{\mu\nu}\xi_\mu\xi_\nu=0$ whenever $\xi\in TK$ is null. If the initial data of $\Phi$ satisfies \begin{equation*}
\begin{split}
 \sum_{i+j+k=16}\int_{\Sigma_{\tau_0}} J^{Z+CN,w^Z}_\mu\left(\hat{Y}^k\partial_{t^*}^i\tilde{\Omega}^j\Phi\right) n^{\mu}_{\Sigma_{\tau_0}} +\sum_{i+j+k=16}\int_{\Sigma_{\tau_0}} J^{Z+CN,w^Z}_\mu\left(\hat{Y}^k S\partial_{t^*}^i\tilde{\Omega}^j\Phi\right) n^{\mu}_{\Sigma_{\tau_0}}\leq \epsilon
\end{split}
\end{equation*}
and 
\begin{equation*}
\begin{split}
\sum_{\ell=0}^{13} r|D^\ell\Phi(\tau_0)|+ r|D^\ell S\Phi(\tau_0)|\leq \epsilon.
\end{split}
\end{equation*}
Then $\Phi$ exists globally in time.
Moreover, for all $\eta>0$, we can take $a$ sufficiently small such that the solution $\Phi$ obeys the decay estimate
$$|\Phi|\leq C\epsilon r^{-1}u^{-\frac{1}{2}}(t^*)^\eta, |D\Phi|\leq C\epsilon r^{-1}u^{-1}(t^*)^{\eta}, |\overline{D}\Phi|\leq C\epsilon r^{-1}(t^*)^{-1+\eta}\mbox{ for }r\geq R,\mbox{ and}$$ 
$$|\Phi|\leq C_\delta\epsilon (t^*)^{-\frac{3}{2}+\eta}r^{\delta}, |D\Phi|\leq C_\delta\epsilon (t^*)^{-\frac{3}{2}+\eta}r^{-\frac{1}{2}+\delta}\mbox{ for }r\leq \frac{t^*}{4}.$$
\end{theorem}
The above formulation is geometric and independent of the choice of coordinates. We note that this condition is obviously satisfied by the wave map equation in the intrinsic formulation.

\subsection{The Case of Minkowski Spacetime}

We now turn to the outline of the proof of the main theorem. In the original proof in \cite{Knull}, many symmetries of Minkowski spacetime are captured and exploited using the vector field method. Kerr spacetime, on the other hand, lacks symmetries and this limits the set of vector fields that is at our disposal. In view of this, we would like to re-examine the proof of the small data global existence result for the nonlinear wave equation with a null condition on Minkowski spacetime, using only the vector fields whose analogues in Kerr spacetimes have been established in previous works. In particular, we would have to avoid using the Lorentz boost. 

We first study the decay properties of the solutions to the linear wave equation on Minkowski spacetime. Since the vector field $T=\partial_t$ is Killing and $Z=u^2\partial_u+v^2\partial_v$ is conformally Killing, we have for $w=8t$ that
$$\int_{\Sigma_t}J^T_\mu(\Phi)n^\mu_{\Sigma_t}, \int_{\Sigma_t}J^{Z,w^Z}_\mu(\Phi)n^\mu_{\Sigma_t}$$
are conserved in time.

Decay can be proved using the above conserved quantities for $V\Phi$ for appropriate vector fields $V$. It is proved separately for $r\geq\frac{t}{2}$ and $r\leq\frac{t}{2}$. In the former case, we use the fact that $\Omega_{ij}=x_i\partial_{x_j}+x_j\partial_{x_i}$ is Killing on Minkowski spacetime and hence $\Box_m(\Omega^k\Phi)=0$. Since $\Omega$ has a weight in $r$, it can be proved that
$$|D\Phi|^2\leq Cr^{-2}\sum_{k=0}^2\int_{\Sigma_t} J^T_\mu(\Omega^k\Phi)n^\mu_{\Sigma_t}.$$
Notice that in this region $r^{-2}\leq Ct^{-2}$. It is known, for example by the representation formula, that this decay rate cannot be improved. In the region $r\leq\frac{t}{2}$, however, the decay rate is better. One can consider the conformal energy
$$\int J^{Z,w^Z}_\mu(\Phi)n^\mu_{\Sigma_t}\geq \int_{\Sigma_t} u^2\left(\underline{L}\Phi\right)^2+v^2\left(L\Phi\right)^2+\left(u^2+v^2\right)|\nabb\Phi|^2+\left(\frac{u^2+v^2}{r^2}\right)\Phi^2,$$
where $u=\frac{1}{2}(t-r)$, $v=\frac{1}{2}(t+r)$. 
In particular, we have $$|D\Phi|^2\leq t^{-2}\int_{\Sigma_t\cap\{r\leq\frac{t}{2}\}} \tau^2(D\Phi)^2\leq t^{-2}\int J^{Z,w^Z}_\mu(\Phi)n^\mu_{\Sigma_t}.$$
To improve the decay rate in this region, we can consider the equation for $S\Phi=(t\partial_t+r\partial_r)\Phi$ and use the integrated decay estimates as in \cite{LS}, \cite{LKerr}. This approach allows us to avoid the use of Lorentz boost in \cite{Knull} and the global elliptic estimates in \cite{KS}, both of which have no clear analogue in Kerr spacetimes. On Minkowski spacetime, a local energy decay estimate can be proved using the vector field $\left(1-\frac{1}{(1+r^2)^{\frac{1+\delta}{2}}}\right)\partial_r$ for the linear wave equation \cite{StR} which together with the conformal energy yields:
$$\int_{t}^{(1.1)t}\int_{\Sigma_{t'}\cap\{r\leq\frac{t'}{2}\}} r^{-1-\delta} J^T_\mu(\Phi)n^\mu_{\Sigma_{t'}} dt'\leq C\int_{\Sigma_\tau\cap\{r\leq\frac{t}{2}\}} (D\Phi)^2\leq Ct^{-2}\int J^{Z,w^Z}_\mu(\Phi)n^\mu_{\Sigma_t}.$$
This would imply that there exists a "dyadic" sequence $t_i\sim (1.1)^i t_0$ on which there is better decay
$$\int_{\Sigma_{t_i}\cap\{r\leq\frac{t}{2}\}} r^{-1-\delta} J^T_\mu(\Phi)n^\mu_{\Sigma_{t_i}}\leq Ct_i^{-3}\int J^{Z,w^Z}_\mu(\Phi)n^\mu_{\Sigma_t}.$$
Since $S$ is Killing on Minkowski spacetime, $\Box_m\Phi=0$ implies $\Box_m(S\Phi)=0$. Then the above argument would give
$$\int_{t}^{(1.1)t}\int_{\Sigma_{t}\cap\{r\leq\frac{t}{2}\}} r^{-1-\delta} J^T_\mu(S\Phi)n^\mu_{\Sigma_t} dt\leq t^{-2}\int J^{Z,w^Z}_\mu(S\Phi)n^\mu_{\Sigma_t}.$$
Since $S=t\partial_t+r\partial_r$ has a weight in $t$, we can integrate along the integral curves of $S$ from the "good" $t_i$ slice and get
$$\int_{\Sigma_{t}\cap\{r\leq\frac{t}{2}\}} r^{-1-\delta} J^T_\mu(\Phi)n^\mu_{\Sigma_{t}}\leq Ct^{-3}\int J^{Z,w^Z}_\mu(\Phi)n^\mu_{\Sigma_{t_0}}.$$
Together with the use of $\Omega$, we have the pointwise estimate
$$|D\Phi|^2\leq Cr^{-1+\delta}\sum_{k=0}^2\int_{\Sigma_{t}\cap\{r\leq\frac{t}{2}\}} r^{-1-\delta} J^T_\mu(\Omega^k\Phi)n^\mu_{\Sigma_{t}}\leq Cr^{-1+\delta}t^{-3}\sum_{k=0}^2\int J^{Z,w^Z}_\mu(\Omega^k\Phi)n^\mu_{\Sigma_{t_0}}.$$
We now study how this decay rate can be used for the nonlinear problem. The main idea is to prove the above conservation and decay estimates in a bootstrap setting, showing that the decay to the linear wave equation is sufficiently strong that the nonlinear terms can be treated as error. In this framework, the decay of $t^{-1}$ is borderline and since the decay rate is better when $r\leq\frac{t}{2}$, the difficulty arises when dealing with terms in the region $r\geq\frac{t}{2}$. Furthermore, in order to achieve this decay of $t^{-1}$ it is imperative to show that $\int_{\Sigma_t} J^T_\mu(\Phi)n^\mu_{\Sigma_t}$ is uniformly bounded in time. We now show a heuristic argument. With the inhomogeneous term, the conservation law for the energy now has the error term:
$$\int_{\Sigma_t} J^T_\mu(\Phi)n^\mu_{\Sigma_t}\leq \int_{\Sigma_{t_0}} J^T_\mu(\Phi)n^\mu_{\Sigma_{t_0}}+\left(\int_{t_0}^t\left(\int_{\Sigma_{t}} \left(\Box_m\Phi\right)^2\right)^{\frac{1}{2}}dt\right)^2,$$
and that for the conformal energy has the error term:
$$\int_{\Sigma_t} J^{Z,w^Z}_\mu(\Phi)n^\mu_{\Sigma_t}\leq \int_{\Sigma_{t_0}} J^{Z,w^Z}_\mu(\Phi)n^\mu_{\Sigma_{t_0}}+\left(\int_{t_0}^t\left(\int_{\Sigma_{t}} (t^2+r^2)\left(\Box_m\Phi\right)^2\right)^{\frac{1}{2}}dt\right)^2$$
Since $\Box_m\Phi$ is quadratic in $D\Phi$, we can use Holder's inequality on the inside integral to control one term in $L^2$ and one in $L^\infty$. However, since on the linear level $D\Phi$ is bounded in $L^2$ and decays as $t^{-1}$ in $L^\infty$, the inhomogeneous term for the estimate for the energy is controlled by
$$\left(\int_{t_0}^t (t)^{-1}(\int_{\Sigma_{t}} J^{T}_\mu(\Phi)n^\mu_{\Sigma_{t}})^{\frac{1}{2}} dt\right)^2.$$ 
This is barely insufficient to show that the energy is bounded. We therefore need to make use of the null condition. The null condition would allow one to prove
\begin{equation}\label{null}
\int (D\Phi\overline{D}\Phi)^2\leq C t^{-2}\int_{\Sigma_{t}} J^{Z,w^Z}_\mu(\partial^k\Phi)n^\mu_{\Sigma_{t}}
\end{equation}
In order to prove this estimate, we observe that in the conformal energy, the good derivatives ($\partial_v$, $\nabb$) has better decay rates. In order to use this, we then need to control the conformal energy. Using again the null condition, the inhomogeneous term in the conservation law for the conformal energy can be bounded by
$$\left(\int_{t_0}^t t^{-1}(\int_{\Sigma_{t}} J^{Z,w^Z}_\mu(\Phi)n^\mu_{\Sigma_{t}})^{\frac{1}{2}} dt\right)^2.$$
This would not be sufficient to prove that the conformal energy is bounded, but is sufficient to prove that it grows no faster than $t^{\eta}$ for sufficiently small data. This in turn would be sufficient to prove the boundedness of the energy and obtain all the necessary decay rates. In practice, the argument is more complicated as we need to control the higher order energy and conformal energy in order to obtain the decay rates.
\subsection{The Case of Kerr Spacetime}

In \cite{DRL} and \cite{LKerr}, all the analogues of the above estimates have been proved in the linear setting in Kerr spacetimes. However, it is apparent from the linear case that several issues arise as we apply a similar strategy to the nonlinear problem on Kerr spacetime. 

Among other issues, two difficulties loom large. The first of these is the lack of symmetries in Kerr spacetimes. While Kerr spacetimes possess the Killing vector field $\partial_{t^*}$, it is spacelike in a neighborhood of the event horizon and thus does not give a non-negative conserved quantities. The works \cite{DRK}, \cite{DRL} suggest that we can instead use the vector fields $N$ and $Z$ on Kerr spacetime as substitutes for $T$ and $Z$ on Minkowski spacetime. $N$ in constructed as the Killing vector field $\partial_{t^*}$ added to a small amount of the red-shift vector field near the event horizon. The red-shift vector field, first introduced in \cite{DRS}, takes advantage of the geometry of the event horizon and has been used crucially to obtain decay rates in \cite{DRS}, \cite{DRK}, \cite{DRL}, \cite{LS}, and \cite{LKerr}. It is one of the few stable features of the Schwarzschild spacetime. The vector field $Z$ approaches the corresponding $Z$ on Minkowski spacetime at the asymptotically flat end and has the weights in $r$ and $t^*$ from which we can prove decay. These vector fields, however, do not correspond to any symmetries of Kerr spacetimes, and therefore, as already is apparent in the linear scenario, the energy estimates would contain error terms that need to be controlled. One consequence is that even in the linear setting, the conformal energy is not bounded. Similar issues arise for the vector field commutators $\Omega$ and $S$, which are crucial in obtaining pointwise decay estimates, whose corresponding error terms at the linear level have been studied in \cite{DRL}, \cite{LS}, \cite{LKerr}. A further issue that arises in the case of the Kerr spacetime is the lack of good vector field commutator that are useful to obtain control of higher order derivatives. This has been treated in the linear setting in \cite{DRK} and \cite{DRL} using $\partial_{t^*}$ and the red-shift vector field as commutators and retrieving all other derivatives via elliptic estimates. In the nonlinear setting, we again use elliptic estimates, noting however that the proof of the elliptic estimates now couples with that of the energy estimates in a bootstrap argument.

Secondly, Kerr spacetimes contain trapped null geodesics. As a consequence, any decay results at the linear level must involve a loss of derivatives. This is manifested in the degeneracy of the integrated decay estimate near $r=3M$. We note, however, that on the linear level the non-degenerate energy can be proved to be bounded without any loss of derivatives. We therefore prove energy bounds that is consistent with the linear scenario. We would try to prove on the highest level of derivatives only a boundedness result and begin to prove decay results on the level of fewer derivatives. However, as we will see, the nonlinear effect comes into play and it is not possible to prove even the boundedness of the non-degenerate energy at the highest level of derivatives. We can nevertheless show that it is bounded by $(t^*)^\eta$. On the level of one less derivative, we can prove that the conformal energy grows no faster than $\tau^{1+\eta}$. Using this fact as we prove the estimates for the non-degenerate energy, we can show that at this level of derivatives, the non-degenerate energy is bounded. This is crucial for obtaining the necessary borderline decay of $(t^*)^{-1}$ in $r\geq\frac{t^*}{2}$, thus allowing us to close the bootstrap argument. Trapping would also cause a loss in derivatives when controlling the error terms arising from the commutation with $S$. To tackle this problem, we would commute with $S$ only once. With this approach, we would not have an improved decay for $DS\Phi$ in $r\leq\frac{t^*}{2}$. Nevertheless, we can show that the bootstrap can be closed. Here we make use of the fact that as we close the assumptions for $S\Phi$, we are at a level of derivatives of $\Phi$ such that the local energy flux decays.

In the next section, we will introduce the energy quantities on Kerr spacetimes that can be thought of as analogues of the energy, conformal energy and the integrated local energy. In section \ref{estimates}, we will state the energy estimates that they satisfy. In section \ref{sectionelliptic}, we will state the elliptic estimates that will be used. Then in section \ref{sectionpointwise}, we prove the necessary $L^\infty$ estimates. With all this preparation, we then prove all the estimates using a bootstrap argument in section \ref{bootstrap}. This then easily implies the main theorem in \ref{pfmaintheorem}.

\section{The Energy Quantities}\label{estimates}
We use three kinds of energy quantities, following the notation in \cite{LKerr}. The represent the non-degenerate energy, the conformal energy and the energy norm for the integrated decay estimate. The nondegnerate energy controls all derivatives:
\begin{proposition}
\begin{equation*}
\begin{split}
&\int_{\Sigma_\tau} (D\Phi)^2
\leq C\int_{\Sigma_\tau}  J^{N}_\mu\left(\Phi\right)n^\mu_{\Sigma_\tau}.
\end{split}
\end{equation*}
\end{proposition}
The conformal energy gives different weights to different derivatives and this will be crucially used to capture the null condition:
\begin{proposition}\label{Zlowerbound}
\begin{equation*}
\begin{split}
&\int_{\Sigma_\tau\cap\{r\geq r^-_Y\}} u^2\left(\underline{L}\Phi\right)^2+v^2\left(L\Phi\right)^2+\left(u^2+v^2\right)|\nabb\Phi|^2+\left(\frac{u^2+v^2}{r^2}\right)\Phi^2\\
\leq &C\int_{\Sigma_\tau} J^{Z+N,w^Z}_\mu\left(\Phi\right)n^\mu_{\Sigma_\tau}+C^2 \tau^2\int_{\Sigma_\tau\cap\{r\leq r^-_Y\}} J^{N}_\mu\left(\Phi\right)n^\mu_{\Sigma_\tau}.
\end{split}
\end{equation*}
\end{proposition}
We use the following notations even though they do not correspond to any vector fields:
\begin{definition}
$$K^{X_0}\left(\Phi\right)=r^{-1-\delta}\mathbbm{1}_{\{|r- 3M|\geq \frac{M}{8}\}}J^N_\mu\left(\Phi\right)n^\mu_{\Sigma_\tau}+r^{-1-\delta}\left(\partial_r\Phi\right)^2+r^{-3-\delta}\Phi^2,\quad\mbox{and}$$
$$K^{X_1}\left(\Phi\right)=r^{-1-\delta}J^N_\mu\left(\Phi\right)n^\mu_{\Sigma_\tau}+r^{-3-\delta}\Phi^2.$$
\end{definition}

\section{The Energy Estimates}
We have proved in \cite{LKerr} the energy estimates for the energy quantities define in the last section for $\Box_{g_K}\Phi=G$. We have boundedness for the non-degenerate energy:

\begin{proposition}\label{bddcom}
Let $G=G_1+G_2$ be any way to decompose the function $G$. Then
\begin{equation*}
\begin{split}
&\int_{\Sigma_{\tau}} J^{N}_\mu\left(\Phi\right)n^\mu_{\Sigma_{\tau}} +\int_{\mathcal H(\tau',\tau)} J^{N}_\mu\left(\Phi\right)n^\mu_{\mathcal H^+} +\iint_{\mathcal R(\tau',\tau)\cap\{r\leq r^-_Y\}}K^{N}\left(\Phi\right)+\iint_{\mathcal R(\tau',\tau)}K^{X_0}\left(\Phi\right)\\
\leq &C\left(\int_{\Sigma_{\tau'}} J^{N}_\mu\left(\Phi\right)n^\mu_{\Sigma_{\tau'}}+\left(\int_{\tau'-1}^{\tau+1}\left(\int_{\Sigma_{t^*}} G_1^2\right)^{\frac{1}{2}}dt^*\right)^2+\iint_{\mathcal R(\tau'-1,\tau+1)}G_1^2\right.\\
&\left.+\sum_{m=0}^{1}\iint_{\mathcal R(\tau'-1,\tau+1)}r^{1+\delta}\left(\partial_{t^*}^{m}G_2\right)^2+\sup_{t^*\in [\tau'-1,\tau+1]}\int_{\Sigma_{t^*}\cap\{|r-3M|\leq\frac{M}{8}\}} G_2^2\right).
\end{split}
\end{equation*}
\end{proposition}
We need an extra derivative for the inhomogeneous term because of trapping. If we know a priori that $G$ is supported away from the trapped region, this loss in derivative is unnecessary.
\begin{proposition}\label{bddcom2}
 Let $G=G_1+G_2$ be any way to decompose the function $G$. Suppose $G_2$ is supported away from $\{r:|r-3M|\leq\frac{M}{8}\}$. Then
\begin{equation*}
\begin{split}
&\int_{\Sigma_{\tau}} J^{N}_\mu\left(\Phi\right)n^\mu_{\Sigma_{\tau}} +\int_{\mathcal H(\tau',\tau)} J^{N}_\mu\left(\Phi\right)n^\mu_{\mathcal H^+} +\iint_{\mathcal R(\tau',\tau)\cap\{r\leq r^-_Y\}}K^{N}\left(\Phi\right)+\iint_{\mathcal R(\tau',\tau)}K^{X_0}\left(\Phi\right)\\
\leq &C\left(\int_{\Sigma_{\tau'}} J^{N}_\mu\left(\Phi\right)n^\mu_{\Sigma_{\tau'}}+\left(\int_{\tau'-1}^{\tau+1}\left(\int_{\Sigma_{t^*}} G_1^2\right)^{\frac{1}{2}}dt^*\right)^2+\iint_{\mathcal R(\tau'-1,\tau+1)}G_1^2\right.\\
&\left.+\iint_{\mathcal R(\tau'-1,\tau+1)}r^{1+\delta}G_2^2\right).
\end{split}
\end{equation*}
\end{proposition}
The estimates for $K^{X_1}$ were also proved. It is be estimated in the same way as $K^{X_0}$ except with an extra derivative.
\begin{proposition}\label{bddcom1}
\begin{equation*}
\begin{split}
&\iint_{\mathcal R(\tau',\tau)}K^{X_1}\left(\Phi\right)\\
\leq& C\left(\sum_{m=0}^1\int_{\Sigma_{\tau'}} J^{N}_\mu\left(\partial_{t^*}^{m}\Phi\right)n^\mu_{\Sigma_{\tau'}}+\sum_{m=0}^1\left(\int_{\tau'-1}^{\tau+1}\left(\int_{\Sigma_{t^*}} \left(\partial_{t^*}^{m}G_1\right)^2\right)^{\frac{1}{2}}dt^*\right)^2+\sum_{m=0}^1\iint_{\mathcal R(\tau'-1,\tau+1)}\left(\partial_{t^*}^{m}G_1\right)^2\right.\\
&\left.+\sum_{m=0}^{2}\iint_{\mathcal R(\tau'-1,\tau+1)}r^{1+\delta}\left(\partial_{t^*}^{m}G_2\right)^2+\sup_{t^*\in [\tau'-1,\tau+1]}\sum_{m=0}^1\int_{\Sigma_{t^*}\cap\{|r-3M|\leq\frac{M}{8}\}} \left(\partial_{t^*}^{m}G_2\right)^2\right).
\end{split}
\end{equation*}
\end{proposition}
As before, if the inhomogeneous term is supported away from the trapped set, we can save a derivative:
\begin{proposition}\label{bddcom3}
 Let $G=G_1+G_2$ be any way to decompose the function $G$. Suppose $G_2$ is supported away from $\{r:|r-3M|\leq\frac{M}{8}\}$. Then
\begin{equation*}
\begin{split}
&\iint_{\mathcal R(\tau',\tau)}K^{X_1}\left(\Phi\right)\\
\leq& C\left(\sum_{m=0}^1\int_{\Sigma_{\tau'}} J^{N}_\mu\left(\partial_{t^*}^{m}\Phi\right)n^\mu_{\Sigma_{\tau'}}+\sum_{m=0}^1\left(\int_{\tau'-1}^{\tau+1}\left(\int_{\Sigma_{t^*}} \left(\partial_{t^*}^{m}G_1\right)^2\right)^{\frac{1}{2}}dt^*\right)^2+\sum_{m=0}^1\iint_{\mathcal R(\tau'-1,\tau+1)}\left(\partial_{t^*}^{m}G_1\right)^2\right.\\
&\left.+\sum_{m=0}^{1}\iint_{\mathcal R(\tau'-1,\tau+1)}r^{1+\delta}\left(\partial_{t^*}^{m}G_2\right)^2\right).
\end{split}
\end{equation*}
\end{proposition}
The conformal energy satisfies the following estimates:
\begin{proposition}\label{conformalenergy}
For $\delta, \delta'>0$ sufficiently small and $0\leq\gamma< 1$, there exist $c=c(\delta,\gamma)$ and $C=C(\delta,\gamma)$ such that the following estimate holds for any solution to $\Box_{g_K}\Phi=G$:
\begin{equation*}
\begin{split}
&c\int_{\Sigma_{\tau}} J^{Z,w^Z}_\mu\left(\Phi\right)n^\mu_{\Sigma_{{\tau}}}+\tau^2\int_{\Sigma_{\tau}\cap\{r\leq \gamma\tau\}} J^{N}_\mu\left(\Phi\right)n^\mu_{\Sigma_{\tau}}\\
\leq &C\int_{\Sigma_{\tau_0}} J^{Z+CN,w^Z}_\mu\left(\Phi\right)n^\mu_{\Sigma_{{\tau_0}}}+C\iint_{\mathcal R(\tau_0,\tau)} t^*r^{-1+\delta}K^{X_1}\left(\Phi\right)\\
& +C\delta'\iint_{\mathcal R(\tau_0,\tau)\cap\{r\leq \frac{t^*}{2}\}} (t^*)^2K^{X_0}\left(\Phi\right)+C\left(\delta'+a\right)\iint_{\mathcal R(\tau_0,\tau)\cap\{r\leq r^-_Y\}}(t^*)^2K^N\left(\Phi\right)\\
&+C(\delta')^{-1}\left(\int_{\tau_0}^{\tau}\left(\int_{\Sigma_{t^*}\cap\{r\geq \frac{t^*}{2}\}}r^2 G^2 \right)^{\frac{1}{2}}dt^*\right)^2+C(\delta')^{-1}\sum_{m=0}^1\iint_{\mathcal R(\tau_0,\tau)\cap\{r\leq\frac{9t^*}{10}\}} (t^*)^2r^{1+\delta}\left(\partial_{t^*}^m G\right)^2\\
&+C(\delta')^{-1}\sup_{t^*\in [\tau_0,\tau]}\int_{\Sigma_{t^*}\cap\{r^-_Y\leq r\leq \frac{25M}{8}\}} (t^*)^2 G^2.
\end{split}
\end{equation*}
\end{proposition}
\begin{remark}
As in Proposition \ref{bddcom2}, we can save a derivative if we know that the inhomogeneous term is supported away from the trapped region. More precisely, let $G=G_1+G_2$ be any way to decompose the function $G$. Suppose $G_2$ is supported away from $\{r:|r-3M|\leq\frac{M}{8}\}$. Then we can replace 
$$\sum_{m=0}^1\iint_{\mathcal R(\tau_0,\tau)\cap\{r\leq\frac{9t^*}{10}\}} (t^*)^2r^{1+\delta}\left(\partial_{t^*}^m G\right)^2+\sup_{t^*\in [\tau_0,\tau]}\int_{\Sigma_{t^*}\cap\{r^-_Y\leq r\leq \frac{25M}{8}\}} (t^*)^2 G^2$$
by $$\sum_{m=0}^1\iint_{\mathcal R(\tau_0,\tau)\cap\{r\leq\frac{9t^*}{10}\}} (t^*)^2r^{1+\delta}\left(\partial_{t^*}^m G_1\right)^2+\iint_{\mathcal R(\tau_0,\tau)\cap\{r\leq\frac{9t^*}{10}\}} (t^*)^2r^{1+\delta} G_2^2+\sup_{t^*\in [\tau_0,\tau]}\int_{\Sigma_{t^*}\cap\{r^-_Y\leq r\leq \frac{25M}{8}\}} (t^*)^2 G_1^2.$$
in Proposition \ref{conformalenergy}. This follows from a straight forward modification of the proof in \cite{LKerr}.
\end{remark}
The estimates for $K^{X_0}$ and $K^{X_1}$ can be localized to $r\leq\frac{t^*}{2}$ if we control it by the conformal energy:
\begin{proposition}\label{X0}
\begin{enumerate}
\item Localized Estimate for $X_0$
\begin{equation*}
\begin{split}
&\iint_{\mathcal R(\tau',\tau)\cap\{r\leq\frac{t^*}{2}\}}K^{X_0}\left(\Phi\right)\\
\leq& C\left(\tau^{-2}\int_{\Sigma_{\tau'}} J^{Z+N,w^Z}_\mu\left(\Phi\right)n^\mu_{\Sigma_{\tau'}}+ C\int_{\Sigma_{\tau'}\cap\{r\leq r^-_Y\}} J^{N}_\mu\left(\Phi\right)n^\mu_{\Sigma_{\tau'}}\right)\\
&+C\left(\sum_{m=0}^{1}\iint_{\mathcal R(\tau'-1,\tau+1)\cap\{r\leq\frac{9t^*}{10}\}}r^{1+\delta'}\left(\partial_{t^*}^{m}{G}\right)^2+\sup_{t^*\in [\tau'-1,\tau+1]}\int_{\Sigma_{t^*}\cap\{|r-3M|\leq\frac{M}{8}\}\cap\{r\leq\frac{9t^*}{10}\}} G^2\right).
\end{split}
\end{equation*}
\item Localized Estimate for $X_1$
\begin{equation*}
\begin{split}
&\iint_{\mathcal R(\tau',\tau)\cap\{r\leq\frac{t^*}{2}\}}K^{X_1}\left(\Phi\right)\\
\leq& C\left(\tau^{-2}\sum_{m=0}^{1}\int_{\Sigma_{\tau'}} J^{Z+N,w^Z}_\mu\left(\partial_{t^*}^m\Phi\right)n^\mu_{\Sigma_{\tau'}}+ C\sum_{m=0}^{1}\int_{\Sigma_{\tau'}\cap\{r\leq r^-_Y\}} J^{N}_\mu\left(\partial_{t^*}^m\Phi\right)n^\mu_{\Sigma_{\tau'}}\right)\\
&+C\left(\sum_{m=0}^{2}\iint_{\mathcal R(\tau'-1,\tau+1)\cap\{r\leq\frac{9t^*}{10}\}}r^{1+\delta}\left(\partial_{t^*}^{m}{G}\right)^2\right.\\
&\left.+\sup_{t^*\in [\tau'-1,\tau+1]}\sum_{m=0}^{1}\int_{\Sigma_{t^*}\cap\{|r-3M|\leq\frac{M}{8}\}\cap\{r\leq\frac{9t^*}{10}\}} \left(\partial_{t^*}^m G\right)^2\right).
\end{split}
\end{equation*}
\end{enumerate}
\end{proposition}
\begin{remark}
As before, if $G=G_1+G_2$ and $G_2$ is supported outside $\{|r-3M|\leq \frac{M}{8}\}$, we can replace, in Proposition \ref{X0}.1, 
$$\sum_{m=0}^{1}\iint_{\mathcal R(\tau'-1,\tau+1)\cap\{r\leq\frac{9t^*}{10}\}}r^{1+\delta}\left(\partial_{t^*}^{m}{G_2}\right)^2+\sup_{t^*\in [\tau'-1,\tau+1]}\int_{\Sigma_{t^*}\cap\{|r-3M|\leq\frac{M}{8}\}\cap\{r\leq\frac{9t^*}{10}\}} G^2$$
by $$\iint_{\mathcal R(\tau'-1,\tau+1)\cap\{r\leq\frac{9t^*}{10}\}}r^{1+\delta}G_2^2;$$
and replace in Proposition \ref{X0}.2,
$$\sum_{m=0}^{2}\iint_{\mathcal R(\tau'-1,\tau+1)\cap\{r\leq\frac{9t^*}{10}\}}r^{1+\delta}\left(\partial_{t^*}^{m}{G_2}\right)^2+\sum_{m=0}^1\sup_{t^*\in [\tau'-1,\tau+1]}\int_{\Sigma_{t^*}\cap\{|r-3M|\leq\frac{M}{8}\}\cap\{r\leq\frac{9t^*}{10}\}} \left(\partial_{t^*}^m G_2\right)^2$$
by $$\sum_{m=0}^{1}\iint_{\mathcal R(\tau'-1,\tau+1)\cap\{r\leq\frac{9t^*}{10}\}}r^{1+\delta}\left(\partial_{t^*}^{m}{G_2}\right)^2.$$
\end{remark}

\section{The Elliptic Estimates and Hardy Inequality}\label{sectionelliptic}
We have also proved in \cite{LKerr} the following elliptic estimates:
\begin{proposition}\label{elliptic}
Suppose $\Box_{g_K}\Phi=G$. For $m\geq 1$ and for any $\alpha$,
\begin{enumerate}
\item Boundedness of Weighted Energy
\begin{equation*}
\begin{split}
\int_{\Sigma_{\tau}\cap\{r\geq r^-_Y\}} r^\alpha\left(D^m\Phi\right)^2 \leq C_\alpha\left(\sum_{j=0}^{m-1}\int_{\Sigma_{\tau}} r^\alpha J^N_\mu\left(\partial_{t^*}^j\Phi\right)n^\mu_{\Sigma_\tau}+\sum_{j=0}^{m-2}\int_{\Sigma_{\tau}} r^\alpha\left(D^jG\right)^2\right).
\end{split}
\end{equation*}
\item Boundedness of Local Energy\\
For any $0< \gamma<\gamma'$,
$$\int_{\Sigma_{\tau}\cap\{r^-_Y\leq r\leq \gamma t^*\}} r^\alpha\left(D^m\Phi\right)^2 \leq C_\alpha\left(\sum_{j=0}^{m-1}\int_{\Sigma_{\tau}\cap\{r\leq \gamma' t^*\}} r^\alpha J^N_\mu\left(\partial_{t^*}^j\Phi\right)n^\mu_{\Sigma_\tau}+\sum_{j=0}^{m-2}\int_{\Sigma_{\tau}} r^\alpha\left(D^jG\right)^2\right).$$
\end{enumerate}
\end{proposition}

We need a Hardy-type inequality that is improves the analogous one in \cite{LKerr}:
\begin{proposition}\label{Hardy}
For $R>R'$,
$$\int_{\Sigma_\tau\cap\{r\geq R\}} r^{\alpha-2}\Phi^2 \leq C\int_{\Sigma_\tau\cap\{r\geq R'\}} r^{\alpha}J^N_\mu\left(\Phi\right)n^\mu_{\Sigma_\tau}.$$
\end{proposition}
\begin{proof}
Let $k(r)$ be defined by solving $$k'(r,\theta,\phi)=r^{\alpha-2}vol,$$ in the region $r\geq R'$, where $vol=vol\left(r,\theta,\phi\right)$ is the volume density on $\Sigma_\tau$ with $r, \theta,\phi$ coordinates, with boundary condition $k(R',\theta,\phi)=0$.
Now
\begin{equation*}
\begin{split}
\int_{\Sigma_\tau\cap\{r\geq R\}} r^{\alpha-2}\Phi^2=& \iiint_{R'}^{\infty} k'(r)\Phi^2 dr d\theta d\phi\\
\leq& -2\iiint k(r)\Phi\partial_r\Phi dr d\theta d\phi\\
\leq &2\left(\iiint_{R'}^{\infty} \frac{1+k(r)^2}{1+k'(r)}\left(\partial_r \Phi\right)^2 dr d\theta d\phi\right)^{\frac{1}{2}}\left(\iiint_{R'}^{\infty} (1+k'(r))\Phi^2 dr d\theta d\phi\right)^{\frac{1}{2}}
\end{split}
\end{equation*}
Notice that $vol \sim r^2$, $k(r)\sim r^{\alpha+1}$ and $1+k'(r)\sim r^{\alpha}$. Hence $\frac{1+k(r)^2}{1+k'(r)}\sim r^{\alpha}vol$. The lemma follows.
\end{proof}

With the help of this Hardy inequality, we are able to ``localize'' the elliptic estimates for $r\geq R$. 
\begin{proposition}\label{ellipticoutside}
Suppose $\Box_{g_K}\Phi=G$. For $m\geq 1$ and for any $\alpha$,
For any $R>R'$,
$$\int_{\Sigma_{\tau}\cap\{r\geq R\}} r^\alpha\left(D^m\Phi\right)^2 \leq C_{\alpha,R,R'}\left(\sum_{j=0}^{m-1}\int_{\Sigma_{\tau}\cap\{r\geq R'\}} r^\alpha J^N_\mu\left(\partial_{t^*}^j\Phi\right)n^\mu_{\Sigma_\tau}+\sum_{j=0}^{m-2}\int_{\Sigma_{\tau}} r^\alpha\left(D^jG\right)^2\right).$$
\end{proposition}
Near the event horizon, elliptic estimates have been proved to control all the derivatives if we have control the $\partial_{t^*}$ and the $\hat{Y}$ derivatives \cite{DRK}, \cite{DRL}, \cite{LKerr}:
\begin{proposition}\label{elliptichorizon}
Suppose $\Box_{g_K}\Phi=G$. For every $m\geq 1$,
\begin{equation*}
\begin{split}
\int_{\Sigma_{\tau}\cap\{r\leq r^-_Y\}} \left(D^m\Phi\right)^2 \leq C\left(\sum_{j+k\leq m-1}\int_{\Sigma_{\tau}\cap\{r\leq r^-_Y\}} J^N_\mu\left(\partial_{t^*}^j\hat{Y}^k\Phi\right)n^\mu_{\Sigma_\tau}+\sum_{j=0}^{m-2}\int_{\Sigma_{\tau}\cap\{r\leq r^-_Y\}}\left(D^jG\right)^2\right).
\end{split}
\end{equation*}
\end{proposition}
This is useful together with the follow control for the equation commuted with $\hat{Y}$:
\begin{proposition}\label{commYcontrol}
Suppose $\Box_{g_K}\Phi=G$. For every $k\geq 0$,
\begin{equation*}
\begin{split}
&\int_{\Sigma_\tau\cap\{r\leq r^+_Y\}} J^{N}_\mu\left(\hat{Y}^{k}\Phi\right)n^\mu_{\Sigma_\tau} +\int_{\mathcal H(\tau',\tau)} J^{N}_\mu\left(\hat{Y}^{k}\Phi\right)n^\mu_{\Sigma_\tau}+ \iint_{\mathcal R(\tau',\tau)\cap\{r\leq r^-_Y\}}J_\mu^{N}\left(\hat{Y}^{k}\Phi\right)n^\mu_{\Sigma_{t^*}}\\
\leq &C\left(\sum_{j+m\leq k}\int_{\Sigma_{\tau'}\cap\{r\leq r^+_Y\}} J^N_\mu\left(\partial_{t^*}^j\hat{Y}^m\Phi\right)n^\mu_{\Sigma_{\tau'}}+\sum_{j=0}^k\int_{\Sigma_{\tau}\cap\{r\leq r^+_Y\}} J^N_\mu\left(\partial_{t^*}^j\Phi\right)n^\mu_{\Sigma_{\tau}}\right.\\
&\left.+\sum_{j=0}^k\iint_{\mathcal R(\tau',\tau)\cap\{ r\leq \frac{23M}{8}\}} J^N_\mu\left(\partial_{t^*}^j\Phi\right)n^\mu_{\Sigma_{t^*}}+\sum_{j=0}^{k}\iint_{\mathcal R(\tau',\tau)\cap\{r\leq\frac{23M}{8}\}}\left(D^jG\right)^2\right).
\end{split}
\end{equation*}
\end{proposition}

\section{Pointwise Estimates}\label{sectionpointwise}
We prove pointwise estimates using Sobolev embedding. We will have different estimates in the region $\{r\geq\frac{t^*}{4}\}$ and $\{r\leq\frac{t^*}{4}\}$. We first consider $\{r\geq\frac{t^*}{4}\}$. For this region, we will prove five different pointwise estimates. First, we prove a boundedness result for $D\Phi$ (Proposition \ref{SE}) using only standard Sobolev Embedding and the elliptic estimates in Proposition \ref{elliptic}. Then we prove decay estimates of $r^{-1}$ for $D^\ell\Phi$ using the $r$ weight in the vector field commutator $\tilde{\Omega}$ and the non-degenerate energy (Proposition \ref{r}). It is crucial that this depends only on the non-degenerate energy but not the conformal energy because we will not be able to prove boundedness of the conformal energy (which already is the case in the \emph{linear} situation, see \cite{DRL}, \cite{AB}, \cite{LKerr}). Notice that Proposition \ref{SE} does not follow from Proposition \ref{r} because the latter requires an extra derivative. This save in derivatives is strictly speaking not necessary for the bootstrap if we have instead assumed an extra derivative of regularity in the initial data. Thirdly, using similar ideas, we will prove the decay of $r^{-1}$ for $\Phi$ using $\tilde{\Omega}$ and the conformal energy (Proposition \ref{rnoderivatives}). Then we prove an extra decay rate of $D\Phi$ using the conformal energy. For any derivatives, we will have an extra decay in the $u$ variable, which degenerates in the wave zone (Proposition \ref{ru}). For the good derivatives, we will have an extra decay in the $v$ variable (Proposition \ref{rv}). This decay rate will be crucial in capturing the good derivative in the null condition. 
\begin{proposition}\label{SE}
For $r\geq\frac{t^*}{4}$\, we have 
\begin{equation*}
\begin{split}
|D\Phi|^2
\leq& C\left(\sum_{k=0}^2\int_{\Sigma_\tau} J^{N}_\mu\left(\partial_{t^*}^k\Phi\right) n^\mu_{\Sigma_\tau}+\sum_{k=0}^1\int_{\Sigma_\tau} \left(D^k\Box_{g_K}\Phi\right)^2\right).
\end{split}
\end{equation*}
\end{proposition}
\begin{proof}
By standard Sobolev Embedding in three dimensions and Proposition \ref{elliptic},
\begin{equation*}
\begin{split}
|D\Phi|^2
\leq& C\sum_{k=1}^3\int_{\Sigma_\tau\cap\{r\geq r^-_Y\}} (D^k\Phi)^2\\
\leq& C\left(\sum_{k=0}^2\int_{\Sigma_\tau} J^{N}_\mu\left(\partial_{t^*}^k\Phi\right) n^\mu_{\Sigma_\tau}+\sum_{k=0}^1\int_{\Sigma_\tau} \left(D^k\Box_{g_K}\Phi\right)^2\right).
\end{split}
\end{equation*}
\end{proof}
We then prove the decay rate of $r^{-1}$ for $D^\ell\Phi$. The idea here is standard: Making use of the commutator $\tilde{\Omega}$, we use the Sobolev Embedding on the 2-sphere and then integrate along the $r$ direction. 
\begin{proposition}\label{r}
For $r\geq\frac{t^*}{4}$ and $\ell\geq 1$, we have 
\begin{equation*}
\begin{split}
&|D^\ell\Phi|^2\\
\leq& Cr^{-2}\left(\sum_{m=0}^{\ell}\sum_{k=0}^2\int_{\Sigma_\tau} J^{N}_\mu\left(\partial_{t^*}^m\tilde{\Omega}^k\Phi\right) n^\mu_{\Sigma_\tau}+\sum_{j=0}^{\ell-1}\sum_{k=0}^2\int_{\Sigma_\tau\cap\{u'\sim u\}\cap\{r\geq\frac{\tau}{2}\}} \left(D^j\Box_{g_K}\left(\tilde{\Omega}^k\Phi\right)\right)^2\right).
\end{split}
\end{equation*}
\begin{proof}
\begin{equation*}
\begin{split}
&r^{2}|D^\ell\Phi|^2\\\leq &C\int_{\mathbb S^2}\left(\left(D^\ell\Phi\right)^2+\left(\tilde{\Omega}D^\ell\Phi\right)^2+\left(\tilde{\Omega}^2 D^\ell\Phi\right)^2\right)r^{2} dA\\
\leq&C\left(\sum_{k=0}^2\int_{\mathbb S^2(\tilde{r})}\left(\tilde{\Omega}^kD^\ell\Phi\right)^2\tilde{r}^{2} dA+\int^{\tilde{r}}_r\int_{\mathbb S^2(r')}|\partial_r\tilde{\Omega}^kD^\ell\Phi\tilde{\Omega}^kD^\ell\Phi|(r')^{2}+\left(\tilde{\Omega}^kD^\ell\Phi\right)^2r' dAdr'\right).
\end{split}
\end{equation*}
Noticing that $|[D,\tilde{\Omega}]\Phi|\leq C|D\Phi|$, we have
\begin{equation*}
\begin{split}
&r^{2}|D^\ell\Phi|^2\\
\leq&C\left(\sum_{k=0}^2\int_{\mathbb S^2(\tilde{r})}\left(\tilde{\Omega}^kD^\ell\Phi\right)^2\tilde{r}^{2} dA+\int^{\tilde{r}}_r\int_{\mathbb S^2(r')}|\partial_r D^\ell\tilde{\Omega}^k\Phi D^\ell\tilde{\Omega}^k\Phi|(r')^{2}+\left(\tilde{\Omega}^kD^\ell\Phi\right)^2r' dAdr'\right)\\
\leq&C\left(\sum_{k=0}^2\int_{\mathbb S^2(\tilde{r})}\left(D^\ell\tilde{\Omega}^k\Phi\right)^2\tilde{r}^{2} dA+\int^{\tilde{r}}_r\int_{\mathbb S^2(r')}|D^{\ell+1}\tilde{\Omega}^k\Phi D^\ell\tilde{\Omega}^k\Phi|(r')^{2}+\left(D^\ell\tilde{\Omega}^k\Phi\right)^2r' dAdr'\right)\\
\leq&C\left(\sum_{k=0}^2\int_{\mathbb S^2(\tilde{r})}\left(D^\ell\tilde{\Omega}^k\Phi\right)^2\tilde{r}^{2} dA+\int^{\tilde{r}}_r\int_{\mathbb S^2(r')}\left(D^{\ell+1}\tilde{\Omega}^k\Phi \right)^2(r')^{2}+\left(D^\ell\tilde{\Omega}^k\Phi\right)^2(r')^2 dAdr'\right).\\
\end{split}
\end{equation*}
Take $r\leq\tilde{r}\leq r+1$. By Proposition \ref{elliptic},
\begin{equation*}
\begin{split}
&\sum_{k=0}^2\int^{r+1}_r\int_{\mathbb S^2(r')}\left(D^{\ell+1}\tilde{\Omega}^k\Phi \right)^2(r')^{2}+\left(D^\ell\tilde{\Omega}^k\Phi\right)^2(r')^2 dAdr'\\
\leq&C\left(\sum_{m=0}^{\ell}\sum_{k=0}^2\int_{\Sigma_\tau}J^N_\mu\left(\partial_{t^*}^m\Omega^k\Phi\right)n^\mu_{\Sigma_\tau}+\sum_{j=0}^{\ell-1}\sum_{k=0}^2\int_{\Sigma_\tau\cap\{u'\sim u\}\cap\{r\geq\frac{\tau}{2}\}} \left(D^j\Box_{g_K}\left(\tilde{\Omega}^k\Phi\right)\right)^2\right).
\end{split}
\end{equation*}
By pigeonholing on this we also get that for some $\tilde{r}$,
\begin{equation*}
\begin{split}
&\sum_{k=0}^2\int_{\mathbb S^2(\tilde{r})}\left(D^\ell\Omega^k\Phi\right)^2\tilde{r}^{2} dA\\
\leq&C\left(\sum_{m=0}^{\ell}\sum_{k=0}^2\int_{\Sigma_\tau}J^N_\mu\left(\partial_{t^*}^m\Omega^k\Phi\right)n^\mu_{\Sigma_\tau}+\sum_{j=0}^{\ell-1}\sum_{k=0}^2\int_{\Sigma_\tau\cap\{u'\sim u\}\cap\{r\geq\frac{\tau}{2}\}} \left(D^j\Box_{g_K}\left(\tilde{\Omega}^k\Phi\right)\right)^2\right).
\end{split}
\end{equation*}

\end{proof}

\end{proposition}
We would also like to prove the pointwise decay in $r$ for $\Phi$. However, we need to use the conformal energy as well as the non-degenerate energy. We note that only the decay in $r$ will be used in the bootstrap argument, the decay in $u$ is proved to achieved the decay rate asserted in Theorem 1.
\begin{proposition}\label{rnoderivatives}
Consider $\Box_{g_K}\Phi=G$. For $r\geq\frac{t^*}{4}$, we have 
\begin{equation*}
\begin{split}
|\Phi|^2
\leq&Cr^{-2}(1+|u|)^{-1}\left(\sum_{k=0}^2\int_{\Sigma_\tau}J^{Z+N,w^Z}_\mu\left(\Omega^k\Phi\right)n^\mu_{\Sigma_\tau}+C\tau^2\sum_{k=0}^2\int_{\Sigma_\tau\cap\{r\leq r^-_Y\}}J^N_\mu\left(\Omega^k\Phi\right)n^\mu_{\Sigma_\tau}\right).
\end{split}
\end{equation*}
\end{proposition}
\begin{proof}
Following the proof of Proposition \ref{r}, we have
\begin{equation*}
\begin{split}
r^2|\Phi|^2
\leq&C\left(\sum_{k=0}^2\int_{\mathbb S^2(\tilde{r})}\left(\tilde{\Omega}^k\Phi\right)^2\tilde{r}^{2} dA+|\int^{\tilde{r}}_r\int_{\mathbb S^2(r')}|\tilde{\Omega}^k\Phi D\tilde{\Omega}^k\Phi| (r')^{2}+\left(\tilde{\Omega}^k\Phi\right)^2r' dAdr'|\right).\\
\end{split}
\end{equation*}
We will treat separately the cases $|u|\leq 1$, $u\geq 1 $, $u\leq 1$. For $|u|\leq 1$,
take $r\leq\tilde{r}\leq r+1$. By Proposition \ref{Zlowerbound},
\begin{equation*}
\begin{split}
&\sum_{k=0}^2\int^{r+1}_r\int_{\mathbb S^2(r')}\left(D\tilde{\Omega}^k\Phi \right)^2(r')^{2}+\left(\tilde{\Omega}^k\Phi\right)^2(r')^2 dAdr'\\
\leq&C\left(\sum_{k=0}^2\int_{\Sigma_\tau}J^{Z+N,w^Z}_\mu\left(\Omega^k\Phi\right)n^\mu_{\Sigma_\tau}+C\tau^2\sum_{k=0}^2\int_{\Sigma_\tau\cap\{r\leq r^-_Y\}}J^N_\mu\left(\Omega^k\Phi\right)n^\mu_{\Sigma_\tau}\right).
\end{split}
\end{equation*}
By pigeonholing on this we also get that for some $\tilde{r}$,
\begin{equation*}
\begin{split}
&\sum_{k=0}^2\int_{\mathbb S^2(\tilde{r})}\left(\Omega^k\Phi\right)^2\tilde{r}^{2} dA\\
\leq&C\left(\sum_{k=0}^2\int_{\Sigma_\tau}J^{Z+N,w^Z}_\mu\left(\Omega^k\Phi\right)n^\mu_{\Sigma_\tau}+C\tau^2\sum_{k=0}^2\int_{\Sigma_\tau\cap\{r\leq r^-_Y\}}J^N_\mu\left(\Omega^k\Phi\right)n^\mu_{\Sigma_\tau}\right).
\end{split}
\end{equation*}
For $u\geq 1$, pick a fixed $R$ and let $\tilde{r}\in [R,R+1]$. Then by a pigeonhole argument, there is some $\tilde{r}$ such that 
$$\sum_{k=0}^2\int_{\mathbb S^2(\tilde{r})}\left(\tilde{\Omega}^k\Phi\right)^2\tilde{r}^{2} dA\leq Cu^{-2}\left(\sum_{k=0}^2\int_{\Sigma_\tau}J^{Z+N,w^Z}_\mu\left(\Omega^k\Phi\right)n^\mu_{\Sigma_\tau}+C\tau^2\sum_{k=0}^2\int_{\Sigma_\tau\cap\{r\leq r^-_Y\}}J^N_\mu\left(\Omega^k\Phi\right)n^\mu_{\Sigma_\tau}\right).$$
By Proposition \ref{Zlowerbound},
\begin{equation*}
\begin{split}
&\sum_{k=0}^2\int^{r}_R\int_{\mathbb S^2(r')}|\tilde{\Omega}^k\Phi D\tilde{\Omega}^k\Phi|(r')^2+\left(\tilde{\Omega}^k\Phi\right)^2r' dAdr'\\
\leq&C\left(rt^*u^{-2}+ru^{-2}\right)\left(\sum_{k=0}^2\int_{\Sigma_\tau}J^{Z+N,w^Z}_\mu\left(\Omega^k\Phi\right)n^\mu_{\Sigma_\tau}+C\tau^2\sum_{k=0}^2\int_{\Sigma_\tau\cap\{r\leq r^-_Y\}}J^N_\mu\left(\Omega^k\Phi\right)n^\mu_{\Sigma_\tau}\right).
\end{split}
\end{equation*}
Using the fact that $t^*\leq Cu$ in this region, we have the desired bound in this region.

Finally, for $u\leq 1$, pick $\tilde{r}\in [-2u,-3u]$. Then by a pigeonhole argument, there is an $\tilde{r}$ such that 
$$\sum_{k=0}^2\int_{\mathbb S^2(\tilde{r})}\left(\tilde{\Omega}^k\Phi\right)^2\tilde{r}^{2} dA\leq Cu^{-2}\left(\sum_{k=0}^2\int_{\Sigma_\tau}J^{Z+N,w^Z}_\mu\left(\Omega^k\Phi\right)n^\mu_{\Sigma_\tau}+C\tau^2\sum_{k=0}^2\int_{\Sigma_\tau\cap\{r\leq r^-_Y\}}J^N_\mu\left(\Omega^k\Phi\right)n^\mu_{\Sigma_\tau}\right).$$
By Proposition \ref{Zlowerbound},
\begin{equation*}
\begin{split}
&\sum_{k=0}^2\int^{\infty}_r\int_{\mathbb S^2(r')}|\tilde{\Omega}^k\Phi D\tilde{\Omega}^k\Phi|(r')^2+\left(\tilde{\Omega}^k\Phi\right)^2r' dAdr'\\
\leq&C|u|^{-1}\left(\sum_{k=0}^2\int_{\Sigma_\tau}J^{Z+N,w^Z}_\mu\left(\Omega^k\Phi\right)n^\mu_{\Sigma_\tau}+C\tau^2\sum_{k=0}^2\int_{\Sigma_\tau\cap\{r\leq r^-_Y\}}J^N_\mu\left(\Omega^k\Phi\right)n^\mu_{\Sigma_\tau}\right),
\end{split}
\end{equation*}
which gives the desired bound.
\end{proof}

We would like to use the conformal energy and elliptic estimates to prove decay in the $u$ variable. However, we need to be careful when applying the localized version of the elliptic estimates. In particular, we need to perform a dyadic decomposition in the variable $u$. We remark that we can prove this for any number of derivatives by iterating the cutoff procedure in the proof of the following Proposition. However, as this will not be necessary in the sequel, we will be content with the following Proposition:

\begin{proposition}\label{energyu}
Suppose $\Box_{g_K}\Phi=G$. Let $r\geq\frac{t^*}{4}$, $\ell=1$ or $2$ and $u_0$ be the $u$- coordinate corresponding to the two sphere $(\tau,r_0)$.
\begin{equation*}
\begin{split}
&\int_{r_0}^{r_0+1}\int_{\mathbb S^2(r')} \left(D^\ell\Phi\right)^2 (r')^2 dA dr'\\
 \leq& C\left(1+|u_0|\right)^{-2}\sum_{j=0}^{\ell-1}\left(\int_{\Sigma_\tau}J^{Z+CN}_\mu\left(\partial^j_{t^*}\Phi\right) n^\mu_{\Sigma_\tau}+C\tau^2\int_{\Sigma_\tau\cap\{r\leq r^-_Y\}}J^{N}_\mu\left(\partial^j_{t^*}\Phi\right) n^\mu_{\Sigma_\tau}\right)\\
&+C\sum_{j=0}^{\ell-2}\int_{\Sigma_\tau\cap\{u\sim u_0\}\cap\{r\geq\frac{\tau}{2}\}} \left(D^j G\right)^2.
\end{split}
\end{equation*}
\end{proposition}
\begin{proof}
The $\ell=1$ case is trivial. For $\ell=2$, we consider separately the cases: Case $0$: $|u_0|\leq C_*$, Case $1_k$: $2^k\leq u_0\leq 2^{k+1}$, Case $2_k$: $-2^{k+1}\leq u_0\leq -2^k$, $k \geq \frac{\log C_*}{\log 2}$ for some sufficiently large but fixed $C_*$. In Case $0$, we have $|u|\leq C$ for the range $[r_0,r_0+1]$ and hence the Proposition is obvious as we have $1\leq C\left(1+|u|\right)^{-2}$.

For the other cases, we consider a cutoff function $\chi:\mathbb R \to \mathbb R_{\geq 0}$ which is compactly supported in $[-2, 2]$ and identically $1$ in $[-1,1]$. In case $1_k$ (resp. $2_k$), we consider $\tilde{\Phi}$ to be defined by $\tilde{\Phi}(\tau,r,\theta,\phi)=\chi\left(2^{-k+3}\left(r-r_0\right)\right)\Phi\left(\tau,r,\theta,\phi\right)$. Then $\tilde{\Phi}$ is supported in $[r_0-2^{k-2}, r_0+2^{k-2}]$ and equals $\Phi$ in $[r_0-2^{k-3}, r_0+2^{k-3}]$. On the support of $\tilde{\Phi}$, $|\Box_{g_K}\tilde{\Phi}-G|\leq C\displaystyle\sum_{j=0}^1 2^{-(2-j)k}|D^j\Phi|$. We also have that on the support of $\tilde{\Phi}$, $|u-u_0|\leq \frac{1}{2}|r^*_S-(r^*_0)_S|\leq \frac{1}{2}|r-r_0|+\frac{M}{2}|\log \frac{r-2M}{r_0-2M}|\leq |r-r_0|\leq 2^{k-1}$ for $r_0$ sufficiently large (which we can assume for otherwise $\tau$ and $r$ must both be bounded, in which case we must be in Case $0$ for appropriately chosen $C_*$). Hence $u\sim 2^k$ (resp. $u\sim -2^k$). 

Therefore, by Proposition \ref{elliptic}.1 applied twice, first to $\tilde{\Phi}$ then to $\Phi$, we have,
\begin{equation*}
\begin{split}
&\int_{r_0}^{r_0+1}\int_{\mathbb S^2(r')} \left(D^2\Phi\right)^2 (r')^2 dA dr'\leq \int_{\Sigma_\tau\cap\{r_0\leq r\leq r_0+1\}} \left(D^2\tilde{\Phi}\right)^2 \\
 \leq& C\sum_{j=0}^{1}\int_{\Sigma_\tau\cap\{r_0-2^{k-3}\leq r\leq r_0+2^{k-3}\}}J^{N}_\mu\left(\partial_{t^*}^j\tilde{\Phi}\right) n^\mu_{\Sigma_\tau}+C\sum_{j=0}^1\int_{\Sigma_\tau\cap\{r_0-2^{k-2}\leq r\leq r_0+2^{k-2}\}}2^{-(2-j)2k}\left(D^j\Phi\right)^2\\
 &+C\int_{\Sigma_\tau\cap\{r_0-2^{k-2}\leq r\leq r_0+2^{k-2}\}} G^2\\
\leq& C\sum_{j=0}^{1}\int_{\Sigma_\tau\cap\{r_0-2^{k-2}\leq r\leq r_0+2^{k-2}\}}\left(2^{-2k}\Phi^2+J^{N}_\mu\left(\partial_{t^*}^j\Phi\right) n^\mu_{\Sigma_\tau}\right)+C\int_{\Sigma_\tau\cap\{r_0-2^{k-2}\leq r\leq r_0+2^{k-2}\}}  G^2\\
 \leq& C\left(1+|u_0|\right)^{-2}\sum_{j=0}^{1}\left(\int_{\Sigma_\tau}J^{Z+CN}_\mu\left(\partial^j_{t^*}\Phi\right) n^\mu_{\Sigma_\tau}+C\tau^2\int_{\Sigma_\tau\cap\{r\leq r^-_Y\}}J^{N}_\mu\left(\partial^j_{t^*}\Phi\right) n^\mu_{\Sigma_\tau}\right)+C\int_{\Sigma_\tau\cap\{u\sim u_0\}} G^2.
\end{split}
\end{equation*}
\end{proof}

Using this we can prove more decay in the $u$ variable:
\begin{proposition}\label{ru}
Suppose $\Box_{g_K}\Phi=G$. For $r\geq\frac{t^*}{4}$ and $\ell\geq 1$, we have 
\begin{equation*}
\begin{split}
&|D\Phi|^2\\
\leq& Cr^{-2}\left(1+|u|\right)^{-2}\sum_{j=0}^{1}\sum_{k=0}^2\left(\int_{\Sigma_\tau}J^{Z+CN}_\mu\left(\partial^j_{t^*}\tilde{\Omega}^k\Phi\right) n^\mu_{\Sigma_\tau}+C\tau^2\int_{\Sigma_\tau\cap\{r\leq r^-_Y\}}J^{N}_\mu\left(\partial^j_{t^*}\tilde{\Omega}^k\Phi\right) n^\mu_{\Sigma_\tau}\right)\\
&+Cr^{-2}\sum_{k=0}^2\int_{\Sigma_\tau\cap\{u'\sim u\}\cap\{r\geq\frac{\tau}{2}\}} \left(\Box_{g_K}\left(\tilde{\Omega}^k\Phi\right)\right)^2.
\end{split}
\end{equation*}
\end{proposition}
\begin{proof}
Following the proof of Proposition \ref{r}, we have
\begin{equation*}
\begin{split}
&r^{2}|D\Phi|^2\\
\leq&C\left(\sum_{k=0}^2\int_{\mathbb S^2(\tilde{r})}\left(D\Omega^k\Phi\right)^2\tilde{r}^{2} dA+\int^{\tilde{r}}_r\int_{\mathbb S^2(r')}\left(D^{2}\Omega^k\Phi\right)\left(D\Omega^k\Phi\right)(r')^{2}+\left(D\Omega^k\Phi\right)^2r' dAdr'\right)\\
\end{split}
\end{equation*}
Take $r\leq\tilde{r}\leq r+1$. Then by Proposition \ref{energyu},
\begin{equation*}
\begin{split}
&\sum_{k=0}^2\int^{r+1}_r\int_{\mathbb S^2(r')}\left(D^{2}\Omega^k\Phi\right)\left(D\Omega^k\Phi\right)(r')^{2}+\left(D\Omega^k\Phi\right)^2r' dAdr'\\
 \leq& C\left(1+|u|\right)^{-2}\sum_{j=0}^{1}\sum_{k=0}^2\left(\int_{\Sigma_\tau}J^{Z+CN}_\mu\left(\partial^j_{t^*}\tilde{\Omega}^k\Phi\right) n^\mu_{\Sigma_\tau}+C\tau^2\int_{\Sigma_\tau\cap\{r\leq r^-_Y\}}J^{N}_\mu\left(\partial^j_{t^*}\tilde{\Omega}^k\Phi\right) n^\mu_{\Sigma_\tau}\right)\\
&+C\sum_{k=0}^2\int_{\Sigma_\tau\cap\{u'\sim u\}\cap\{r\geq\frac{\tau}{2}\}}\left(\Box_{g_K}\left(\tilde{\Omega}^k\Phi\right)\right)^2.
\end{split}
\end{equation*}
By pigeonholing on this we also get that for some $\tilde{r}$,
\begin{equation*}
\begin{split}
&\sum_{k=0}^2\int_{\mathbb S^2(\tilde{r})}\left(D^\ell\Omega^k\Phi\right)^2\tilde{r}^{2} dA\\
\leq& C\left(1+|u|\right)^{-2}\sum_{j=0}^{1}\sum_{k=0}^2\left(\int_{\Sigma_\tau}J^{Z+CN}_\mu\left(\partial^j_{t^*}\tilde{\Omega}^k\Phi\right) n^\mu_{\Sigma_\tau}+C\tau^2\int_{\Sigma_\tau\cap\{r\leq r^-_Y\}}J^{N}_\mu\left(\partial^j_{t^*}\tilde{\Omega}^k\Phi\right) n^\mu_{\Sigma_\tau}\right)\\
&+C\sum_{k=0}^2\int_{\Sigma_\tau\cap\{u'\sim u\}\cap\{r\geq\frac{\tau}{2}\}} \left(\Box_{g_K}\left(\tilde{\Omega}^k\Phi\right)\right)^2.
\end{split}
\end{equation*}
\end{proof}

We have a better pointwise decay for a ``good'' derivative:
\begin{proposition}\label{rv}
For $r\geq\frac{t^*}{4}$, we have 
\begin{equation*}
\begin{split}
|\bar{D}\Phi|^2\leq&C r^{-4}\sum_{k=0}^2\sum_{i+j\leq 1}\left(\int_{\Sigma_{\tau}} J^N_\mu\left(S^i\partial_{t^*}^j\Phi\right)n^\mu_{\Sigma_\tau}+\int_{\Sigma_\tau}J^{Z+CN}_\mu\left(\partial^j_{t^*}\tilde{\Omega}^k\Phi\right) n^\mu_{\Sigma_\tau}\right.\\
&\left.\quad\quad\quad\quad\quad\quad+C\tau^2\int_{\Sigma_\tau\cap\{r\leq r^-_Y\}}J^{N}_\mu\left(\partial^j_{t^*}\tilde{\Omega}^k\Phi\right) n^\mu_{\Sigma_\tau}+\int_{\Sigma_{\tau}} \left(\Box_{g_K}\left(\tilde{\Omega}^k\Phi\right)\right)^2\right)\\
&+Cr^{-2}\sum_{k=0}^2\int_{\Sigma_\tau\cap\{r\geq\frac{\tau}{2}\}} \left(\Box_{g_K}\left(\tilde{\Omega}^k\Phi\right)\right)^2.
\end{split}
\end{equation*}
\begin{proof}
\begin{equation*}
\begin{split}
&r^{2}|\bar{D}\Phi|^2\\\leq &C\int_{\mathbb S^2}\left(\left(\bar{D}\Phi\right)^2+\left(\tilde{\Omega}\bar{D}\Phi\right)^2+\left(\tilde{\Omega}^2 \bar{D}\Phi\right)^2\right)r^{2} dA\\
\leq&C\left(\sum_{k=0}^2\int_{\mathbb S^2(\tilde{r})}\left(\tilde{\Omega}^k\bar{D}\Phi\right)^2\tilde{r}^{2} dA+\int^{\tilde{r}}_r\int_{\mathbb S^2(r')}|\partial_r\tilde{\Omega}^k\bar{D}\Phi\tilde{\Omega}^k\bar{D}\Phi|(r')^{2}+\left(\tilde{\Omega}^k\bar{D}\Phi\right)^2r' dAdr'\right).
\end{split}
\end{equation*}
Noticing that $|[D,\tilde{\Omega}]\Phi|\leq C|D\Phi|$, $|[\bar{D},\tilde{\Omega}]\Phi|\leq C\left(|\bar{D}\Phi|+r^{-1}|D\Phi|\right)$ and $|\bar{D},\partial_r\Phi|\leq Cr^{-1}|D\Phi|$, we have
\begin{equation}\label{rv1}
\begin{split}
&r^{2}|\bar{D}\Phi|^2\\
\leq&C\sum_{k=0}^2\left(\int_{\mathbb S^2(\tilde{r})}\left(\tilde{\Omega}^k\bar{D}\Phi\right)^2\tilde{r}^{2} dA\right.\\
&\left.+\int^{\tilde{r}}_r\int_{\mathbb S^2(r')}|\partial_r\bar{D}\tilde{\Omega}^k\Phi\bar{D}\tilde{\Omega}^k\Phi|(r')^{2}+\left(\tilde{\Omega}^k\bar{D}\Phi\right)^2r' dAdr'\right)\\
\leq&C\sum_{k=0}^2\left(\int_{\mathbb S^2(\tilde{r})}\left(\left(\bar{D}\tilde{\Omega}^k\Phi\right)^2+\tilde{r}^{-2}\left(D\tilde{\Omega}^k\Phi\right)^2\right)\tilde{r}^{2} dA\right.\\
&\left.+\int^{\tilde{r}}_r\int_{\mathbb S^2(r')}\left(\left(\bar{D}D\tilde{\Omega}^k\Phi\right)^2+\left(\bar{D}\tilde{\Omega}^k\Phi\right)^2+(r')^{-2}\left(D\tilde{\Omega}^k\Phi\right)^2\right)(r')^2 dAdr'\right).
\end{split}
\end{equation}
The last term already exhibits better decay rate:
\begin{equation*}
\begin{split}
\int^{\tilde{r}}_r\int_{\mathbb S^2(r')}(r')^{-2}\left(D\tilde{\Omega}^k\Phi\right)^2(r')^2 dAdr'\leq &Cr^{-2}\int_{\Sigma_\tau}J^{N}_\mu\left(\tilde{\Omega}^k\Phi\right) n^\mu_{\Sigma_\tau}
\end{split}
\end{equation*}
We will now show that the energy quantities involving $\bar{D}$ obey better decay rates. This is immediate for the term $\int^{\tilde{r}}_r\int_{\mathbb S^2(r')}\left(\bar{D}\tilde{\Omega}^k\Phi\right)^2(r')^2 dAdr'$ using the conformal energy:
\begin{equation*}
\begin{split}
&\int^{\tilde{r}}_r\int_{\mathbb S^2(r')}\left(\bar{D}\tilde{\Omega}^k\Phi\right)^2(r')^2 dAdr'\\
\leq&Cv^{-2}\left(\int_{\Sigma_\tau}J^{Z+CN}_\mu\left(\tilde{\Omega}^k\Phi\right) n^\mu_{\Sigma_\tau}+C\tau^2\int_{\Sigma_\tau\cap\{r\leq r^-_Y\}}J^{N}_\mu\left(\tilde{\Omega}^k\Phi\right) n^\mu_{\Sigma_\tau}\right).
\end{split}
\end{equation*}
However, we note that this cannot be shown directly for the term $\int^{\tilde{r}}_r\int_{\mathbb S^2(r')}\left(\bar{D}D\tilde{\Omega}^k\Phi\right)^2(r')^2 dAdr'$ with the conformal energy because we cannot commute $\Box_{g_K}$ with derivatives of every direction. In order to remedy this, we use the non-degenerate energy for $S\Phi$. In particular, we use the fact that for $r\geq r^-_Y$, $|\bar{D}\Phi|\leq Cv^{-1}\left(|S\Phi|+u|D\Phi|+vr^{-1}|D\Phi|\right)$.
\begin{equation*}
\begin{split}
&\int^{\tilde{r}}_r\int_{\mathbb S^2(r')}\left(\bar{D}D\tilde{\Omega}^k\Phi\right)^2(r')^2dAdr'\\
\leq&C\int^{\tilde{r}}_r\int_{\mathbb S^2(r')}\left((v')^{-2}\left(SD\tilde{\Omega}^k\Phi\right)^2+(u')^2(v')^{-2}\left(D^{2}\tilde{\Omega}^k\Phi\right)+(r')^{-2}\left(D^{2}\tilde{\Omega}^k\Phi\right)^2 \right)(r')^2dAdr'\\
\leq&C\int^{\tilde{r}}_r\int_{\mathbb S^2(r')}\left((v')^{-2}\left(DS\tilde{\Omega}^k\Phi\right)^2+(v')^{-2}\left(D\tilde{\Omega}^k\Phi\right)^2\right.\\
&\left.\quad\quad\quad\quad\quad\quad\quad+(u')^2(v')^{-2}\left(D^{2}\tilde{\Omega}^k\Phi\right)+(r')^{-2}\left(D^{2}\tilde{\Omega}^k\Phi\right)^2 \right)(r')^2dAdr'\\
\end{split}
\end{equation*}
Take $r\leq\tilde{r}\leq r+1$. We have, for the first two terms,
\begin{equation*}
\begin{split}
&\int^{r+1}_r\int_{\mathbb S^2(r')}(v')^{-2}\left(\left(DS\tilde{\Omega}^k\Phi\right)^2+\left(D\tilde{\Omega}^k\Phi\right)^2\right)(r')^2dAdr'\leq Cv^{-2}\sum_{j=0}^1\int_{\Sigma_\tau}J^{N}_\mu\left(S^j\tilde{\Omega}^k\Phi\right) n^\mu_{\Sigma_\tau}.
\end{split}
\end{equation*}
The third term can be estimated by Proposition \ref{energyu},
\begin{equation*}
\begin{split}
&\int^{r+1}_r\int_{\mathbb S^2(r')}(u')^2(v')^{-2}\left(D^{2}\tilde{\Omega}^k\Phi\right)(r')^2dAdr'\\
\leq& Cv^{-2}\sum_{j=0}^{1}\left(\int_{\Sigma_\tau}J^{Z+CN}_\mu\left(\partial^j_{t^*}\tilde{\Omega}^k\Phi\right) n^\mu_{\Sigma_\tau}+C\tau^2\int_{\Sigma_\tau\cap\{r\leq r^-_Y\}}J^{N}_\mu\left(\partial^j_{t^*}\tilde{\Omega}^k\Phi\right) n^\mu_{\Sigma_\tau}\right)\\
&+C\int_{\Sigma_\tau\cap\{u'\sim u\}\cap\{r\geq\frac{\tau}{2}\}} \left(\Box_{g_K}\left(\tilde{\Omega}^k\Phi\right)\right)^2.
\end{split}
\end{equation*}
The fourth term can be estimated elliptically by Proposition \ref{elliptic}:
\begin{equation*}
\begin{split}
&\int^{r+1}_r\int_{\mathbb S^2(r')}(r')^{-2}\left(D^{2}\tilde{\Omega}^k\Phi\right)^2(r')^2dAdr'\\
\leq& C r^{-2}\left(\sum_{j=0}^{1}\int_{\Sigma_{\tau}} J^N_\mu\left(\partial_{t^*}^j\tilde{\Omega}^k\Phi\right)n^\mu_{\Sigma_\tau}+\int_{\Sigma_{\tau}} \left(\Box_{g_K}\left(\tilde{\Omega}^k\Phi\right)\right)^2\right).
\end{split}
\end{equation*}
Collecting all the above estimates and noting that $r\geq\frac{\tau}{2}$, we get
\begin{equation}\label{rv2}
\begin{split}
&\sum_{k=0}^2\int^{r+1}_r\int_{\mathbb S^2(r')}\left(\left(\bar{D}D\tilde{\Omega}^k\Phi\right)^2+\left(\bar{D}\tilde{\Omega}^k\Phi\right)^2+(r')^{-2}\left(D\tilde{\Omega}^k\Phi\right)^2\right)(r')^2 dAdr'\\
\leq&C v^{-2}\sum_{k=0}^2\sum_{i+j\leq 1}\left(\int_{\Sigma_{\tau}} J^N_\mu\left(S^i\partial_{t^*}^j\Phi\right)n^\mu_{\Sigma_\tau}+\int_{\Sigma_\tau}J^{Z+CN}_\mu\left(\partial^j_{t^*}\tilde{\Omega}^k\Phi\right) n^\mu_{\Sigma_\tau}\right.\\
&\left.\quad\quad\quad\quad\quad\quad+C\tau^2\int_{\Sigma_\tau\cap\{r\leq r^-_Y\}}J^{N}_\mu\left(\partial^j_{t^*}\tilde{\Omega}^k\Phi\right) n^\mu_{\Sigma_\tau}+\int_{\Sigma_{\tau}} \left(\Box_{g_K}\left(\tilde{\Omega}^k\Phi\right)\right)^2\right)\\
&+C\sum_{k=0}^2\int_{\Sigma_\tau\cap\{r\geq\frac{\tau}{2}\}} \left(\Box_{g_K}\left(\tilde{\Omega}^k\Phi\right)\right)^2.
\end{split}
\end{equation}
By pigeonholing on this we also get that for some $\tilde{r}$,
\begin{equation}\label{rv3}
\begin{split}
&\sum_{k=0}^2\left(\int_{\mathbb S^2(\tilde{r})}\left(\bar{D}\tilde{\Omega}^k\Phi\right)^2+\tilde{r}^{-2}\left(D\tilde{\Omega}^k\Phi\right)^2\right)\tilde{r}^{2} dA\\
\leq&C v^{-2}\sum_{k=0}^2\sum_{i+j\leq 1}\left(\int_{\Sigma_{\tau}} J^N_\mu\left(S^i\partial_{t^*}^j\Phi\right)n^\mu_{\Sigma_\tau}+\int_{\Sigma_\tau}J^{Z+CN}_\mu\left(\partial^j_{t^*}\tilde{\Omega}^k\Phi\right) n^\mu_{\Sigma_\tau}\right.\\
&\left.\quad\quad\quad\quad\quad\quad+C\tau^2\int_{\Sigma_\tau\cap\{r\leq r^-_Y\}}J^{N}_\mu\left(\partial^j_{t^*}\tilde{\Omega}^k\Phi\right) n^\mu_{\Sigma_\tau}+\int_{\Sigma_{\tau}} \left(\Box_{g_K}\left(\tilde{\Omega}^k\Phi\right)\right)^2\right)\\
&+C\sum_{k=0}^2\int_{\Sigma_\tau\cap\{r\geq\frac{\tau}{2}\}} \left(\Box_{g_K}\left(\tilde{\Omega}^k\Phi\right)\right)^2.
\end{split}
\end{equation}
(\ref{rv1}), (\ref{rv2}) and (\ref{rv3}) together imply the Proposition.
\end{proof}
\end{proposition}
We now turn to the region $r\leq\frac{t^*}{4}$. We first show a simple Sobolev embedding result.
\begin{proposition}\label{sSobolev}
Suppose $\Box_{g_K}\Phi=G$. For $\ell\geq 1$ and $r\leq\frac{t^*}{4}$, 
\begin{equation*}
|D^\ell\Phi|^2\leq C\left(\sum_{j+m\leq \ell+1}\int_{\Sigma_\tau\cap\{r\leq\frac{t^*}{2}\}}J^{N}_\mu\left(\partial_{t^*}^m\hat{Y}^j\Phi\right)n^\mu_{\Sigma_\tau}+\sum_{j=0}^{\ell}\int_{\Sigma_\tau}(D^jG)^2\right)
\end{equation*}
\end{proposition}

We can capture better estimates in $r$ if we use an extra derivative.
\begin{proposition}\label{SEDinside}
For $\ell\geq 1$ and $r\leq\frac{t^*}{4}$, 
\begin{equation*}
|D^\ell\Phi|^2\leq Cr^{-2}\left(\sum_{j+m+k\leq \ell+2}\int_{\Sigma_\tau\cap\{r\leq\frac{t^*}{2}\}}J^{N}_\mu\left(\partial_{t^*}^m\hat{Y}^j\tilde{\Omega}^k\Phi\right)n^\mu_{\Sigma_\tau}+\sum_{j=0}^{\ell+1-k}\sum_{k=0}^2\int_{\Sigma_\tau}\left(D^j\Box_{g_K}\left(\tilde{\Omega}^k\Phi\right)\right)^2\right)
\end{equation*}
\end{proposition}
\begin{proof}
We only need to consider the situation when $r\geq R_\Omega +C$. For otherwise, this Proposition is implied by Proposition \ref{sSobolev} since $r$ is finite. We assume from now on that $r\geq R_\Omega +C$.
Following the proof of Proposition \ref{r}, we have
\begin{equation*}
\begin{split}
&r^2|D^\ell\Phi|^2\\
\leq&C\left(\sum_{k=0}^2\int_{\mathbb S^2(\tilde{r})}\left(D^\ell\tilde{\Omega}^k\Phi\right)^2\tilde{r}^{2} dA+\int_{\tilde{r}}^r\int_{\mathbb S^2(r')}\left(D^{\ell+1}\tilde{\Omega}^k\Phi \right)^2(r')^{2}+\left(D^\ell\tilde{\Omega}^k\Phi\right)^2(r')^2 dAdr'\right)\\
\end{split}
\end{equation*}
Take $r-1\leq\tilde{r}\leq r$. By Proposition \ref{elliptic}.2 and \ref{elliptichorizon},
\begin{equation*}
\begin{split}
&\sum_{k=0}^2\int^{r}_{r-1}\int_{\mathbb S^2(r')}\left(D^{\ell+1}\tilde{\Omega}^k\Phi \right)^2(r')^{2}+\left(D^{\ell}\tilde{\Omega}^k\Phi\right)^2(r')^2 dAdr'\\
\leq&C\left(\sum_{j+m\leq \ell}\sum_{k=0}^2\int_{\Sigma_\tau\cap\{r\leq\frac{t^*}{2}\}}J^{N}_\mu\left(\partial_{t^*}^m\hat{Y}^j\Omega^k\Phi\right)n^\mu_{\Sigma_\tau}+\sum_{j=0}^{\ell-1}\sum_{k=0}^2\int_{\Sigma_\tau}\left(D^j\Box_{g_K}\left(\tilde{\Omega}^k\Phi\right)\right)^2\right).
\end{split}
\end{equation*}
By pigeonholing on this we also get that for some $\tilde{r}$ with $r-1\leq\tilde{r}\leq r$,
\begin{equation*}
\begin{split}
&\sum_{k=0}^2\int_{\mathbb S^2(\tilde{r})}\left(D\Omega^k\Phi\right)^2\tilde{r}^{2} dA\\
\leq&C\left(\sum_{j+m\leq \ell}\sum_{k=0}^2\int_{\Sigma_\tau\cap\{r\leq\frac{t^*}{2}\}}J^{N}_\mu\left(\partial_{t^*}^m\hat{Y}^j\Omega^k\Phi\right)n^\mu_{\Sigma_\tau}+\sum_{j=0}^{\ell-1}\sum_{k=0}^2\int_{\Sigma_\tau}\left(D^j\Box_{g_K}\left(\tilde{\Omega}^k\Phi\right)\right)^2\right).
\end{split}
\end{equation*}
\end{proof}
We also have pointwise estimates for $\Phi$ instead of $D\Phi$ if we use the conformal energy.
\begin{proposition}\label{SEinside}
Suppose $\Box_{g_K}\Phi=0$. For $r\leq\frac{t^*}{4}$,
\begin{equation*}
\begin{split}
|\Phi|^2\leq&C\tau^{-2}\left(\sum_{i+j\leq 2}\int_{\Sigma_\tau}J^{Z+N,w^Z}_\mu\left(\hat{Y}^i\partial_{t^*}^j\Phi\right)n^\mu_{\Sigma_\tau}+C\tau^2\sum_{i+j\leq 2}\int_{\Sigma_\tau\cap\{r\leq r^-_Y\}}J^N_\mu\left(\hat{Y}^i\partial_{t^*}^j\Phi\right)n^\mu_{\Sigma_\tau}\right).
\end{split}
\end{equation*}
\end{proposition}
\begin{proof}
By Sobolev Embedding in three dimensions, for $r\leq\frac{t^*}{4}$,
\begin{equation*}
\begin{split}
|\Phi|^2\leq&C\sum_{k=0}^{2}\int_{\Sigma_\tau\cap\{r\leq\frac{t^*}{4}\}} \left(D^k\Phi\right)^2.
\end{split}
\end{equation*}
Then, using the elliptic estimates in Propositions \ref{elliptic}.2 and \ref{elliptichorizon}, we have
\begin{equation*}
\begin{split}
|\Phi|^2\leq&C\sum_{i+j\leq 2}\int_{\Sigma_\tau\cap\{r\leq\frac{t^*}{2}\}} \Phi^2+J^N_\mu\left(\hat{Y}^i\partial_{t^*}^j\Phi\right)n^\mu_{\Sigma_\tau}.
\end{split}
\end{equation*}
Using Proposition \ref{Zlowerbound}, we can conclude the Proposition.
\end{proof}

We proceed to show that the pointwise estimate is better if we use the vector field commutator $S$. To this end, we first show that we can control a fixed $t^*$ quantity by an integrated quantity. The proof follows ideas in \cite{LS}, \cite{LKerr} and applies an integration in the direction of $S$.

\begin{proposition}\label{extradecay}
For any sufficiently regular $\Phi$, not necessarily satisfying any differential equations, and $\alpha_0$ a constant,
\begin{equation*}
\begin{split}
&\int_{\Sigma_{\tau}\cap\{r \leq\frac{\tau}{4}\}} r^{\alpha_0-2}\Phi^2\\
\leq &C\tau^{-1}\left(\iint_{\mathcal R((1.1)^{-1}\tau,\tau)\cap\{r \leq\frac{t^*}{3}\}} r^{\alpha_0-2}\Phi^2 +\iint_{\mathcal R((1.1)^{-1}\tau,\tau)\cap\{r \leq\frac{t^*}{3}\}} r^{\alpha_0-2} \left(S\Phi\right)^2\right).
\end{split}
\end{equation*}
\end{proposition}
\begin{proof}
To use the estimates for $S\Phi$, we need to integrate along integral curves of $S$. The following argument imitates that for proving improved decay for the homogeneous equation in \cite{LKerr}. We first find the integral curves by solving the ordinary differential equation
$$\frac{dr_S}{dt^*_S}=\frac{h(r_S)}{t^*_S}$$
where $h(r_S)$ is as in the definition of $S$. Hence the integral curves are given by
$$\frac{\exp\left(\int_{(r_S)_0}^{r_S} \frac{dr'_S}{h(r'_S)}\right)}{t^*_S}=\mbox{constant},$$
where $(r_S)_0>2M$ can be chosen arbitrarily.
Let $\sigma=t^*$, $\rho=\frac{\exp\left(\int_{(r_S)_0}^{r_S} \frac{dr_S'}{h(r_S')}\right)}{t_S^*}$ and consider $(\sigma,\rho, x^A, x^B)$ as a new system of coordinates. Notice that
$$\partial_\sigma=\frac{h(r_S)}{t_S^*}\partial_{r_S}+\partial_{t_S^*}=\frac{1}{t_S^*}S.$$
Now for each fixed $\rho$, we have
$$\Phi^2(\tau)\leq \Phi^2(\tau')+|\int_{\tau'}^\tau \frac{1}{\sigma}S(\Phi^2) d\sigma|.$$
Multiplying by $\rho^{\alpha}$ and integrating along a finite region of $\rho$, we get:
$$\int_{\rho_1}^{\rho_2} \Phi^2(\tau)\rho^{\alpha} d\rho\leq \int_{\rho_1}^{\rho_2} \Phi^2(\tau')\rho^{\alpha} d\rho+\int_{\rho_1}^{\rho_2} \int_{\tau'}^\tau |\frac{2\rho^{\alpha}}{\sigma}\Phi S\Phi |d\sigma d\rho.$$
We choose $\alpha$ so that $\alpha=0$ for $r\leq r^-_Y$ and $\alpha=\alpha_0$ for $r \geq R$ and smooth depending on $r$ in between.
We would like to change coordinates back to $(t^*_S,r_S, x_S^A, x_S^B)$. Notice that since $h(r_S)$ is everywhere positive, $(\rho,\tau)$ would correspond to a point with a larger value of $r_S$ than $(\rho,\tau')$. Therefore,
\begin{equation*}
\begin{split}
&\int_{2M}^{(r_S)_2} \Phi^2(\tau)\frac{\exp\left((1+\alpha)\int_{2M}^{r_S} \frac{dr_S'}{h(r_S')}\right)}{\tau h(r_S)}dr_S\\
\leq & \int_{2M}^{(r_S)_2} \Phi^2(\tau')\frac{\exp\left((1+\alpha)\int_{(r_S)_0}^{r_S} \frac{dr_S'}{h(r_S')}\right)}{\tau' h(r_S)}dr_S+\int_{\tau'}^\tau \int_{2M}^{(r_S)_2} |\frac{2}{\sigma}\Phi S\Phi |\frac{\exp\left((1+\alpha)\int_{(r_S)_0}^{r_S} \frac{dr_S'}{h(r_S')}\right)}{t^* h(r_S)}dr_S dt^*.
\end{split}
\end{equation*}
We have to compare $\frac{\exp\left((1+\alpha(r_S))\int_{(r_S)_0}^{r_S} \frac{dr'_S}{h(r'_S)}\right)}{h(r_S)}$ with the volume form. Very close to the horizon, $h(r_S)=r_S-2M$ and $\alpha(r)=0$. Hence 
$$\frac{\exp\left((1+\alpha)\int_{(r_S)_0}^{r_S} \frac{dr'_S}{h(r'_S)}\right)}{h(r_S)}=e^{\int_{(r_S)_0}^{r_S}\frac{dr'_S}{h(r'_S)}}\left(\frac{1}{r_S-2M}\right)\sim 1.$$
On the other hand, for $r\geq R$, $h(r_S)=(r_S+2M\log(r_S-2M)-3M-2M \log M )(1-\mu )$ and $\alpha(r_S)=\alpha_0$. In particular, for a sufficiently large choice of $R$, $h(r_S)\sim r_S$. Hence
$$\frac{\exp\left((1+\alpha)\int_{(r_S)_0}^{r_S} \frac{dr'_S}{h(r'_S)}\right)}{h(r_S)}\sim \frac{\exp\left((1+\alpha)\int_{(r_S)_0}^{r_S} \frac{dr'_S}{h(r'_S)}\right)}{r_S}\sim\left(\frac{r_S^{\alpha_0}}{R}\right)\sim r^{\alpha_0-2}.$$
The corresponding expression on the compact set $[r^-_Y,R]$ is obviously bounded. Hence, since the volume density both on a slice and on a spacetime region is $\sim r^2$, we have
\begin{equation*}
\begin{split}
&\int_{\Sigma_{\tau}\cap\{r <r_2\}} \frac{\Phi^2(\tau)}{\tau}r^{\alpha_0-2}\leq C\left(\int_{\Sigma_{\tau'}\cap\{r <r_2\}} \frac{\Phi^2(\tau')}{\tau'}r^{\alpha_0-2}+\iint_{\mathcal R(\tau',\tau)\cap\{r <r_2\}} r^{\alpha_0-2}|\frac{2}{(t^*)^2}\Phi S\Phi |\right).
\end{split}
\end{equation*}
This easily implies the following improved decay for the non-degenerate energy for $\tau'\in[(1.1)^{-1}\tau,\tau]$:
\begin{equation}\label{intS}
\begin{split}
&\int_{\Sigma_{\tau}\cap\{r <\frac{\tau}{4}\}} r^{\alpha_0-2}\Phi^2\leq C\tau^{-1}\left(\int_{\Sigma_{\tau'}\cap\{r <\frac{\tau'}{3}\}} r^{\alpha_0-2}\Phi^2 +\iint_{\mathcal R((1.1)^{-1}\tau,\tau)\cap\{r <\frac{t^*}{3}\}} r^{\alpha_0-2} \left(S\Phi\right)^2\right).
\end{split}
\end{equation}
By choosing an appropriate $\tilde{\tau}$, we have
$$\int_{\Sigma_{\tilde{\tau}}\cap\{r <\frac{\tilde{\tau}}{3}\}} r^{\alpha_0-2}\Phi^2\leq C\tau^{-1}\iint_{\mathcal R((1.1)^{-1}\tau,\tau)\cap\{r <\frac{t^*}{3}\}} r^{\alpha_0-2}\Phi^2.$$
Now, apply (\ref{intS}) with $\tau'=\tilde{\tau}$, we have 
\begin{equation*}
\begin{split}
&\int_{\Sigma_{\tau}\cap\{r \leq\frac{\tau}{4}\}} r^{-1-\delta}\Phi^2\\
\leq &C\tau\left(\int_{\Sigma_{\tilde{\tau}}\cap\{r <\frac{\tilde{\tau}}{3}\}} \frac{\Phi^2}{\tilde{\tau}}r^{\alpha_0-2}+\iint_{\mathcal R(\tilde{\tau},\tau)\cap\{r <\frac{t^*}{3}\}} r^{\alpha_0-2}|\frac{2}{(t^*)^2}\Phi S\Phi |\right)\\
\leq &C\tau^{-1}\left(\iint_{\mathcal R((1.1)^{-1}\tau,\tau)\cap\{r \leq\frac{t^*}{3}\}} r^{\alpha_0-2}\Phi^2 +\iint_{\mathcal R((1.1)^{-1}\tau,\tau)\cap\{r \leq\frac{t^*}{3}\}} r^{\alpha_0-2} \left(S\Phi\right)^2\right),
\end{split}
\end{equation*}
using Cauchy-Schwarz for the second term.
\end{proof}

By Sobolev Embedding, this would give an improved decay estimate in $t^*$ in the region $\{r\leq\frac{t^*}{4}\}$. For the application, we also need an improved decay in $r$, which we get by commuting with the angular momentum $\tilde{\Omega}$.
\begin{proposition}\label{inside1}
Suppose $\Box_{g_K}\Phi=G$. For $r\leq\frac{t^*}{4}$ and $\ell\geq 1$, we have 
\begin{equation*}
\begin{split}
|D^\ell\Phi|^2\leq &C(t^*)^{-1}r^{-1+\delta}\sum_{i+j\leq \ell-1}\sum_{k=0}^2\iint_{\mathcal R((1.1)^{-1}t^*,t^*)\cap\{r\leq\frac{t^*}{2}\}}\left(K^{X_1}\left(\hat{Y}^i\partial_{t^*}^{j}\tilde{\Omega}^k\Phi\right)+K^{X_1}\left(S\hat{Y}^i\partial_{t^*}^{j}\tilde{\Omega}^k\Phi\right)\right)\\
&+C(t^*)^{-1}r^{-1+\delta}\sum_{j=0}^{\ell-1}\sum_{k=0}^2\iint_{\mathcal R((1.1)^{-1}t^*,t^*)\cap\{r \leq\frac{t^*}{2}\}}r^{-1-\delta}\left(D^j\Box_{g_K}\left(\tilde{\Omega}^k\Phi\right)\right)^2.\\
\end{split}
\end{equation*}

\end{proposition}
\begin{proof}
Using a similar argument as before, except for choosing $\tilde{r}\leq r$, we have 
\begin{equation*}
\begin{split}
&r^{1-\delta}|D^\ell\Phi|^2\\
\leq &C\sum_{k=0}^{2}\int_{\mathbb S^2}\left(\tilde{\Omega}^kD^\ell\Phi\right)^2r^{1-\delta}dA\\
\leq&C\left(\sum_{k=0}^2\int_{\mathbb S^2(\tilde{r})}\left(D^\ell\Omega^k\Phi\right)^2\tilde{r}^{1-\delta} dA+\int_{\tilde{r}}^r\int_{\mathbb S^2(r')}\left(\left(D^{\ell+1}\Omega^k\Phi\right)^2+\left(D^\ell\Omega^k\Phi\right)^2\right)(r')^{1-\delta} dAdr'\right).\\
\end{split}
\end{equation*}
Using Proposition \ref{extradecay}, we have
\begin{equation*}
\begin{split}
&\int_{\tilde{r}}^r\int_{\mathbb S^2(r')}\left(D^\ell\Omega^k\Phi\right)^2(r')^{1-\delta} dAdr'\\
\leq&C\int_{\Sigma_\tau\cap\{r\leq\frac{\tau}{4}\}}r^{-1-\delta}\left(D^\ell\Omega^k\Phi\right)^2\\
\leq&C\tau^{-1}\left(\iint_{\mathcal R((1.1)^{-1}\tau,\tau)\cap\{r \leq\frac{t^*}{3}\}} r^{-1-\delta}\left(D^\ell\Omega^k\Phi\right)^2 +\iint_{\mathcal R((1.1)^{-1}\tau,\tau)\cap\{r \leq\frac{t^*}{3}\}} r^{-1-\delta} \left(SD^\ell\Omega^k\Phi\right)^2\right)
\end{split}
\end{equation*}

By first commuting $[D,S]$ and then using Proposition \ref{elliptic}.2 and \ref{elliptichorizon} on each fixed $t^*$ slice in the integral, we have
\begin{equation*}
\begin{split}
&\iint_{\mathcal R((1.1)^{-1}\tau,\tau)\cap\{r \leq\frac{t^*}{3}\}} r^{-1-\delta}\left(D^\ell\Omega^k\Phi\right)^2 +\iint_{\mathcal R((1.1)^{-1}\tau,\tau)\cap\{r \leq\frac{t^*}{3}\}} r^{-1-\delta} \left(SD^\ell\Omega^k\Phi\right)^2\\
\leq&C\sum_{i+j\leq\ell-1}\iint_{\mathcal R((1.1)^{-1}\tau,\tau)\cap\{r \leq\frac{t^*}{2}\}}r^{-1-\delta}\left(J^N_\mu\left(Y^i\partial_{t^*}^{j-i}\tilde{\Omega}^k\Phi\right)n^\mu_{\Sigma_\tau}+J^N_\mu\left(SY^i\partial_{t^*}^{j-i}\tilde{\Omega}^k\Phi\right)n^\mu_{\Sigma_\tau}\right)\\
&+C\sum_{j=0}^{\ell-1}\iint_{\mathcal R((1.1)^{-1}\tau,\tau)\cap\{r \leq\frac{t^*}{2}\}}r^{-1-\delta}\left(D^j\Box_{g_K}\left(\tilde{\Omega}^k\Phi\right)\right)^2\\
\leq&C\sum_{i+j\leq\ell-1}\iint_{\mathcal R((1.1)^{-1}\tau,\tau)\cap\{r \leq\frac{t^*}{2}\}}\left(K^{X_1}\left(Y^i\partial_{t^*}^{j-i}\tilde{\Omega}^k\Phi\right)+K^{X_1}\left(SY^i\partial_{t^*}^{j-i}\tilde{\Omega}^k\Phi\right)\right)\\
&+C\sum_{j=0}^{\ell-1}\iint_{\mathcal R((1.1)^{-1}\tau,\tau)\cap\{r \leq\frac{t^*}{2}\}}r^{-1-\delta}\left(D^j\Box_{g_K}\left(\tilde{\Omega}^k\Phi\right)\right)^2.\\
\end{split}
\end{equation*}
Therefore, we have
\begin{equation*}
\begin{split}
&r^{1-\delta}|D^\ell\Phi|^2\\
\leq&C\tau^{-1}\sum_{i+j\leq \ell-1}\sum_{k=0}^2\iint_{\mathcal R((1.1)^{-1}\tau,\tau)\cap\{r\leq\frac{t^*}{2}\}}\left(K^{X_1}\left(\hat{Y}^i\partial_{t^*}^{j}\tilde{\Omega}^k\Phi\right)+K^{X_1}\left(S\hat{Y}^i\partial_{t^*}^{j}\tilde{\Omega}^k\Phi\right)\right)\\
&+C\tau^{-1}\sum_{j=0}^{\ell-1}\sum_{k=0}^2\iint_{\mathcal R((1.1)^{-1}\tau,\tau)\cap\{r \leq\frac{t^*}{2}\}}r^{-1-\delta}\left(D^j\Box_{g_K}\left(\tilde{\Omega}^k\Phi\right)\right)^2.\\
\end{split}
\end{equation*}
\end{proof}
Similar ideas can be used to prove decay of $\Phi$ without derivatives, except for a loss in powers of $r$. This will not be used for the bootstrap argument, but will be used to prove the decay for $\Phi$ in the statement of Theorem 1.
\begin{proposition}\label{inside2}
Suppose $\Box_{g_K}\Phi=G$. For $r\leq\frac{t^*}{4}$, we have 
\begin{equation*}
\begin{split}
|\Phi|^2\leq &C(t^*)^{-1}r^{\delta}\sum_{k=0}^2\iint_{\mathcal R((1.1)^{-1}t^*,t^*)\cap\{r\leq\frac{t^*}{3}\}}\left(K^{X_1}\left(\tilde{\Omega}^k\Phi\right)+K^{X_1}\left(S\tilde{\Omega}^k\Phi\right)\right).
\end{split}
\end{equation*}

\end{proposition}
\begin{proof}
Fix $R$. Take $\tilde{r}\in [R, \frac{\tau}{5}]$, we have 
\begin{equation*}
\begin{split}
&r^{-\delta}|\Phi|^2\\
\leq &C\sum_{k=0}^{2}\int_{\mathbb S^2}\left(\tilde{\Omega}^k\Phi\right)^2r^{-\delta}dA\\
\leq&C\left(\sum_{k=0}^2\int_{\mathbb S^2(\tilde{r})}\left(\Omega^k\Phi\right)^2\tilde{r}^{-\delta} dA+|\int_{\tilde{r}}^r\int_{\mathbb S^2(r')}|\Omega^k\Phi D\Omega^k\Phi|^2(r')^{-\delta}+\left(\Omega^k\Phi\right)^2(r')^{-1-\delta} dAdr'|\right).\\
\end{split}
\end{equation*}
There exists $\tilde{r}\in [R, \frac{\tau}{5}]$ such that 
$$\int_{\mathbb S^2(\tilde{r})}\left(\Omega^k\Phi\right)^2\tilde{r}^{-\delta} dA\leq \tau^{-1}\int_{r_+}^{\frac{\tau}{4}}\int_{\mathbb S^2(r')}\left(\Omega^k\Phi\right)^2(r')^{-\delta} dAdr'.$$
Using Proposition \ref{extradecay}, we have
\begin{equation*}
\begin{split}
&\int_{r_+}^{\frac{\tau}{4}}\int_{\mathbb S^2(\tilde{r})}\left(\Omega^k\Phi\right)^2(r')^{-\delta} dAdr'\\
\leq&C\tau\int_{\Sigma_\tau\cap\{r\leq\frac{\tau}{4}\}}r^{-3-\delta}\left(\Omega^k\Phi\right)^2\\
\leq&C\left(\iint_{\mathcal R((1.1)^{-1}\tau,\tau)\cap\{r \leq\frac{t^*}{3}\}} r^{-3-\delta}\left(\Omega^k\Phi\right)^2 +\iint_{\mathcal R((1.1)^{-1}\tau,\tau)\cap\{r \leq\frac{t^*}{3}\}} r^{-3-\delta} \left(S\Omega^k\Phi\right)^2\right) \\
\leq&C\iint_{\mathcal R((1.1)^{-1}\tau,\tau)\cap\{r \leq\frac{t^*}{3}\}}\left(K^{X_1}\left(\tilde{\Omega}^k\Phi\right)+K^{X_1}\left(S\tilde{\Omega}^k\Phi\right)\right).
\end{split}
\end{equation*}
Using Proposition \ref{extradecay}, we also have
\begin{equation*}
\begin{split}
&|\int_{\tilde{r}}^r\int_{\mathbb S^2(r')}|\Omega^k\Phi D\Omega^k\Phi|^2(r')^{-\delta}+\left(\Omega^k\Phi\right)^2(r')^{-1-\delta} dAdr'|\\
\leq&C\int_{\Sigma_\tau\cap\{r\leq\frac{\tau}{4}\}}r^{-3-\delta}\left(\Omega^k\Phi\right)^2+\int_{\Sigma_\tau\cap\{r\leq\frac{\tau}{4}\}}r^{-1-\delta}\left(D\Omega^k\Phi\right)^2\\
\leq&C\tau^{-1}\left(\iint_{\mathcal R((1.1)^{-1}\tau,\tau)\cap\{r \leq\frac{t^*}{3}\}} r^{-3-\delta}\left(\Omega^k\Phi\right)^2 +\iint_{\mathcal R((1.1)^{-1}\tau,\tau)\cap\{r \leq\frac{t^*}{3}\}} r^{-3-\delta} \left(S\Omega^k\Phi\right)^2\right) \\
&+C\tau^{-1}\left(\iint_{\mathcal R((1.1)^{-1}\tau,\tau)\cap\{r \leq\frac{t^*}{3}\}} r^{-1-\delta}\left(D\Omega^k\Phi\right)^2 +\iint_{\mathcal R((1.1)^{-1}\tau,\tau)\cap\{r \leq\frac{t^*}{3}\}} r^{-1-\delta} \left(SD\Omega^k\Phi\right)^2\right)\\
\leq&C\tau^{-1}\iint_{\mathcal R((1.1)^{-1}\tau,\tau)\cap\{r \leq\frac{t^*}{3}\}}\left(K^{X_1}\left(\tilde{\Omega}^k\Phi\right)+K^{X_1}\left(S\tilde{\Omega}^k\Phi\right)\right).
\end{split}
\end{equation*}

Therefore, we have
\begin{equation*}
\begin{split}
|\Phi|^2
\leq&C\tau^{-1}r^{\delta}\sum_{k=0}^2\iint_{\mathcal R((1.1)^{-1}\tau,\tau)\cap\{r\leq\frac{t^*}{2}\}}\left(K^{X_1}\left(\tilde{\Omega}^k\Phi\right)+K^{X_1}\left(S\tilde{\Omega}^k\Phi\right)\right).\\
\end{split}
\end{equation*}
\end{proof}
\section{Bootstrap}\label{bootstrap}
Bootstrap assumptions (J): We first introduce the bootstrap assumptions corresponding to energy quantities on a fixed $t^*$ slice.
\begin{equation}\label{BA1}
  \sum_{i+j=16}A_j^{-1}\int_{\Sigma_\tau}J^{N}_\mu\left(\partial_{t^*}^i\tilde{\Omega}^j\Phi\right) n^{\mu}_{\Sigma_\tau}+\sum_{i+k=16}A_Y^{-1}\int_{\Sigma_\tau}J^{N}_\mu\left(\hat{Y}^k\partial_{t^*}^i\Phi\right) n^{\mu}_{\Sigma_\tau}\leq \epsilon \tau^{\eta_{16}}.
\end{equation}
\begin{equation}\label{BA2}
\begin{split}
 \sum_{i+j=15}A_j^{-1}\left(\int_{\Sigma_\tau} J^{Z+N,w^Z}_\mu\left(\partial_{t^*}^i\tilde{\Omega}^j\Phi\right) n^{\mu}_{\Sigma_\tau} +C\tau^2\int_{\Sigma_\tau\cap\{r\leq \frac{9\tau}{10}\}} J^{N}_\mu\left(\partial_{t^*}^i\tilde{\Omega}^j\Phi\right) n^{\mu}_{\Sigma_\tau}\right)&\\
+\sum_{i+k=15}A_{Y}^{-1}\tau^2\int_{\Sigma_\tau\cap\{r\leq r^+_Y\}}J^{N}_\mu\left(\hat{Y}^k\partial_{t^*}^i\Phi\right) n^{\mu}_{\Sigma_\tau}&\leq \epsilon \tau^{1+\eta_{15}}.
\end{split}
\end{equation}
\begin{equation}\label{BA3}
 \begin{split}
 \sum_{i+j\leq 14}A_j^{-1}\left(\int_{\Sigma_\tau} J^{Z+N,w^Z}_\mu\left(\partial_{t^*}^i\tilde{\Omega}^j\Phi\right) n^{\mu}_{\Sigma_\tau} +C\tau^2\int_{\Sigma_\tau\cap\{r\leq \frac{9\tau}{10}\}} J^{N}_\mu\left(\partial_{t^*}^i\tilde{\Omega}^j\Phi\right) n^{\mu}_{\Sigma_\tau}\right)&\\
+\sum_{i+k\leq 14}A_{Y}^{-1}\tau^2\int_{\Sigma_\tau\cap\{r\leq r^+_Y\}}J^{N}_\mu\left(\hat{Y}^k\partial_{t^*}^i\Phi\right) n^{\mu}_{\Sigma_\tau}&\leq \epsilon \tau^{\eta_{14}}.
\end{split}
\end{equation}
\begin{equation}\label{BA4}
  \sum_{i+j\leq 15}A_j^{-1}\int_{\Sigma_\tau} J^{N}_\mu\left(\partial_{t^*}^i\tilde{\Omega}^j\Phi\right) n^{\mu}_{\Sigma_\tau}\leq \epsilon .
\end{equation}
\begin{equation}\label{BA5}
  \sum_{i+j=13}A_{S,j}^{-1}\int_{\Sigma_\tau} J^{N}_\mu\left(S\partial_{t^*}^i\tilde{\Omega}^j\Phi\right) n^{\mu}_{\Sigma_\tau}+\sum_{i+k=13}A_{S,Y}^{-1}\int_{\Sigma_\tau}J^{N}_\mu\left(\hat{Y}^kS\partial_{t^*}^i\Phi\right) n^{\mu}_{\Sigma_\tau}\leq \epsilon \tau^{\eta_{S,13}}.
\end{equation}
\begin{equation}\label{BA6}
\begin{split}
 \sum_{i+j=12}A_{S,j}^{-1}\left(\int_{\Sigma_\tau} J^{Z+N,w^Z}_\mu\left(S\partial_{t^*}^i\tilde{\Omega}^j\Phi\right) n^{\mu}_{\Sigma_\tau} +C\tau^2\int_{\Sigma_\tau\cap\{r\leq \frac{9\tau}{10}\}} J^{N}_\mu\left(S\partial_{t^*}^i\tilde{\Omega}^j\Phi\right) n^{\mu}_{\Sigma_\tau}\right)&\\
+\sum_{i+k=12}A_{S,Y}^{-1}\tau^2\int_{\Sigma_\tau\cap\{r\leq r^+_Y\}}J^{N}_\mu\left(\hat{Y}^kS\partial_{t^*}^i\Phi\right) n^{\mu}_{\Sigma_\tau}&\leq \epsilon \tau^{1+\eta_{S,12}}.
\end{split}
\end{equation}
\begin{equation}\label{BA7}
\begin{split} 
\sum_{i+j\leq 11}A_{S,j}^{-1}\left(\int_{\Sigma_\tau} J^{Z+N,w^Z}_\mu\left(S\left(\partial_{t^*}^i\tilde{\Omega}^j\Phi\right)\right) n^{\mu}_{\Sigma_\tau} +C\tau^2\int_{\Sigma_\tau\cap\{r\leq \frac{9\tau}{10}\}} J^{N}_\mu\left(S\left(\partial_{t^*}^i\tilde{\Omega}^j\Phi\right)\right) n^{\mu}_{\Sigma_\tau}\right)&\\
+\sum_{i+k\leq 12}A_{S,Y}^{-1}\tau^2\int_{\Sigma_\tau\cap\{r\leq r^+_Y\}}J^{N}_\mu\left(\hat{Y}^kS\partial_{t^*}^i\Phi\right) n^{\mu}_{\Sigma_\tau}&\leq \epsilon \tau^{\eta_{S,11}}
\end{split}
\end{equation}
\begin{equation}\label{BA8}
  A_{S,j}^{-1}\sum_{j=0}^{12}\int_{\Sigma_\tau} J^{N}_\mu\left(S\partial_{t^*}^i\tilde{\Omega}^j\Phi\right) n^{\mu}_{\Sigma_\tau}\leq \epsilon .
\end{equation}
Bootstrap Assumptions (K): We also need bootstrap assumptions for the energy quantities in a spacetime slab.
\begin{equation}\label{BAK1}
\begin{split} 
\sum_{i+j=16}A_{X,j}^{-1}\iint_{\mathcal R(\tau_0,\tau)} \left(K^{X_0}\left(\partial_{t^*}^i\tilde{\Omega}^j\Phi\right)+K^{N}\left(\partial_{t^*}^i\tilde{\Omega}^j\Phi\right)\right)\leq \epsilon\tau^{\eta_{16}}.
\end{split}
\end{equation}
\begin{equation}\label{BAK1.4}
\sum_{i+j=15}A_{X,j}^{-1}\iint_{\mathcal R(\tau_0,\tau)} K^{X_1}\left(\partial_{t^*}^i\tilde{\Omega}^j\Phi\right)\leq \epsilon\tau^{\eta_{16}}.
\end{equation}
\begin{equation}\label{BAK1.5}
 \sum_{i+j\leq 15}A_{X,j}^{-1}\iint_{\mathcal R(\tau_0,\tau)} K^{X_0}\left(\partial_{t^*}^i\tilde{\Omega}^j\Phi\right)\leq \epsilon.
\end{equation}
\begin{equation}\label{BAK1.6}
 \sum_{i+j\leq 14}A_{X,j}^{-1}\iint_{\mathcal R(\tau_0,\tau)} K^{X_1}\left(\partial_{t^*}^i\tilde{\Omega}^j\Phi\right)\leq \epsilon.
\end{equation}
\begin{equation}\label{BAK2}
\begin{split}
 \sum_{i+j\leq 15}A_{X,j}^{-1}\iint_{\mathcal R((1.1)^{-1}\tau,\tau)\cap\{r\leq\frac{t^*}{2}\}}\left(K^{X_0}\left(\partial_{t^*}^{i}\tilde{\Omega}^j\Phi\right)+K^{N}\left(\partial_{t^*}^{i}\tilde{\Omega}^j\Phi\right)\right)\leq\epsilon\tau^{-1+\eta_{15}}.
\end{split}
\end{equation}
\begin{equation}\label{BAK3}
\sum_{i+j\leq 14}A_{X,j}^{-1}\iint_{\mathcal R((1.1)^{-1}\tau,\tau)\cap\{r\leq\frac{t^*}{2}\}} K^{X_1}\left(\partial_{t^*}^{i}\tilde{\Omega}^j\Phi\right)\leq\epsilon\tau^{-1+\eta_{15}}.
\end{equation}
\begin{equation}\label{BAK4}
\begin{split}
 \sum_{i+j\leq 14}A_{X,j}^{-1}\iint_{\mathcal R((1.1)^{-1}\tau,\tau)\cap\{r\leq\frac{t^*}{2}\}}\left(K^{X_0}\left(\partial_{t^*}^{i}\tilde{\Omega}^j\Phi\right)+K^{N}\left(\partial_{t^*}^{i}\tilde{\Omega}^j\Phi\right)\right)\leq\epsilon\tau^{-2+\eta_{14}}.
\end{split}
\end{equation}
\begin{equation}\label{BAK5}
 \sum_{i+j\leq 13}A_{X,j}^{-1}\iint_{\mathcal R((1.1)^{-1}\tau,\tau)\cap\{r\leq\frac{t^*}{2}\}}K^{X_1}\left(\partial_{t^*}^{i}\tilde{\Omega}^j\Phi\right)\leq\epsilon\tau^{-2+\eta_{14}}.
\end{equation}
\begin{equation}\label{BAK6}
\begin{split}
 \sum_{i+j=13}A_{S,X,j}^{-1}\iint_{\mathcal R(\tau_0,\tau)} K^{X_0}\left(S\partial_{t^*}^i\tilde{\Omega}^j\Phi\right)\leq \epsilon\tau^{\eta_{S,13}}.
\end{split}
\end{equation}
\begin{equation}\label{BAK6.5}
\begin{split}
 \sum_{i+j\leq 12}A_{S,X,j}^{-1}\iint_{\mathcal R(\tau_0,\tau)} K^{X_0}\left(S\partial_{t^*}^i\tilde{\Omega}^j\Phi\right)\leq \epsilon.
\end{split}
\end{equation}
\begin{equation}\label{BAK7}
\begin{split}
 \sum_{i+j+k\leq 12}A_{S,X,j}^{-1}\iint_{\mathcal R((1.1)^{-1}\tau,\tau)\cap\{r\leq\frac{t^*}{2}\}}K^{X_0}\left(S\partial_{t^*}^{i}\tilde{\Omega}^j\Phi\right)\leq\epsilon\tau^{-1+\eta_{S,12}}.
\end{split}
\end{equation}
\begin{equation}\label{BAK8}
\sum_{i+j\leq 11}A_{S,X,j}^{-1}\iint_{\mathcal R((1.1)^{-1}\tau,\tau)\cap\{r\leq\frac{t^*}{2}\}} K^{X_1}\left(S\partial_{t^*}^{i}\tilde{\Omega}^j\Phi\right)\leq\epsilon\tau^{-1+\eta_{S,12}}.
\end{equation}
\begin{equation}\label{BAK9}
\begin{split}
 \sum_{i+j\leq 11}A_{S,X,j}^{-1}\iint_{\mathcal R((1.1)^{-1}\tau,\tau)\cap\{r\leq\frac{t^*}{2}\}}K^{X_0}\left(S\partial_{t^*}^{i}\tilde{\Omega}^j\Phi\right)\leq\epsilon\tau^{-2+\eta_{S,11}}.
\end{split}
\end{equation}
\begin{equation}\label{BAK10}
 \sum_{i+j\leq 10}A_{S,X,j}^{-1}\iint_{\mathcal R((1.1)^{-1}\tau,\tau)\cap\{r\leq\frac{t^*}{2}\}}K^{X_1}\left(S\partial_{t^*}^{i}\tilde{\Omega}^j\Phi\right)\leq\epsilon\tau^{-2+\eta_{S,11}}.
\end{equation}

Bootstrap Assumptions (P): We also introduce bootstrap assumptions for the pointwise behavior.
For $r\geq\frac{t^*}{4}$,
\begin{equation}\label{BAP1}
\sum_{j=0}^{13}|\Gamma^j\Phi|^2\leq BA\epsilon r^{-2}(t^*)^{1+\eta_{14}}.
\end{equation}
\begin{equation}\label{BAP1.5}
\sum_{j=0}^{13}|D\Gamma^j\Phi|^2\leq BA\epsilon.
\end{equation}
\begin{equation}\label{BAP2}
\sum_{\ell=1}^{13-j}\sum_{j=0}^{12}|D^\ell\Gamma^j\Phi|^2\leq BA\epsilon r^{-2}.
\end{equation}
\begin{equation}\label{BAP3}
\sum_{j=0}^8|D\Gamma^j\Phi|^2\leq BA\epsilon r^{-2}(t^*)^{\eta_{14}}(1+|u|)^{-2}.
\end{equation}
\begin{equation}\label{BAP4}
\sum_{j=0}^8|\bar{D}\Gamma^{j}\Phi|^2\leq BA\epsilon r^{-2}(t^*)^{-2+\eta_{14}}.
\end{equation}
\begin{equation}\label{BAP5}
\sum_{j=0}^{6}|S\Gamma^j\Phi|^2\leq B_SA\epsilon r^{-2}(t^*)^{\eta_{S,11}}.
\end{equation}
\begin{equation}\label{BAP6}
\sum_{j=0}^8|DS\Gamma^j\Phi|^2\leq B_SA\epsilon r^{-2}.
\end{equation}
\begin{equation}\label{BAP7}
\sum_{j=0}^6|DS\Gamma^j\Phi|^2\leq B_SA\epsilon r^{-2}(t^*)^{\eta_{S,11}}(1+|u|)^{-2}.
\end{equation}
For $r\leq\frac{t^*}{4}$,
\begin{equation}\label{BAPI1}
\sum_{j=0}^{13}|\Gamma^j\Phi|^2\leq BA\epsilon (t^*)^{-1+\eta_{14}}.
\end{equation}
\begin{equation}\label{BAPI1.1}
\sum_{\ell=1}^{14-j}\sum_{j=0}^{13}|D^\ell\Gamma^j\Phi|^2\leq BA\epsilon (t^*)^{-1+\eta_{14}}.
\end{equation}
\begin{equation}\label{BAPI1.2}
\sum_{\ell=1}^{13-j}\sum_{j=0}^{12}|D^\ell\Gamma^j\Phi|^2\leq BA\epsilon (t^*)^{-2+\eta_{14}}.
\end{equation}
\begin{equation}\label{BAPI2}
\sum_{\ell=1}^{9-j}\sum_{j=0}^8|D^\ell\Gamma^j\Phi|^2\leq BA\epsilon (t^*)^{-3+\eta_{S,11}}r^{-1+\delta}.
\end{equation}
\begin{equation}\label{BAPI3}
\sum_{j=0}^{6}|S\Gamma^j\Phi|^2\leq B_SA\epsilon (t^*)^{-2+\eta_{S,11}}.
\end{equation}
\begin{equation}\label{BAPI4}
\sum_{\ell=1}^{7-j}\sum_{j=0}^{6}|D^\ell S\Gamma^j\Phi|^2\leq B_SA\epsilon r^{-2}(t^*)^{-2+\eta_{S,11}}.
\end{equation}
\begin{remark}
Notice that in general, for most of the bootstrap assumptions on $\Phi$, there is a corresponding one on $S\Phi$. The argument to retrieve these assumptions are quite similar, we only have to estimate the commutator term (in a manner similar to \cite{LS}, \cite{LKerr}) and track the appropriate constants. 
\end{remark}
\begin{remark}
Notice that all this assumption are satisfied initially by the assumption of the Theorem.
\end{remark}

\begin{remark}
We will bootstrap to improve the constants $A_j$, $A_{X,j}$, $A_S$, $A_{S,X,j}$, $A_Y$, $A_{S,Y}$ and $B$. The constant $B$ is only used for the bootstrap of the pointwise estimates. The constants satisfy
$$1\ll B\sim B_S\ll A_0\ll A_{X,0}\ll A_1\ll ...\ll A_{16}\ll A_{X,16}\ll A_Y\ll A_{S,0}\ll A_{S,X,0}\ll ... \ll A_{S,X,13}\ll A_{S,Y}.$$
We will use $A$ as a shorthand to denote the maximum of all these constants, i.e., $A_{S,Y}$. We will always assume by taking $\epsilon$ small that 
$$A\epsilon\ll 1.$$
Moreover, we set the constants so that
$$\frac{A_{j-1}}{A_j}\ll\frac{A_{X,j-1}}{A_j}\ll\delta'\eta^{-1}.$$
$$\delta\sim\delta'\ll\frac{A_j}{A_{X,j}}.$$
The $\eta$'s, on the other hand, satisfy
$$\delta\sim\delta'\ll\eta_{16}\ll\eta_{15}\ll\eta_{14}\ll\eta_{S,13}\ll\eta_{S,12}\ll\eta_{S,11}.$$
The $\eta$'s are chosen so that
$$\frac{A_j}{A_{X,j}}\ll\eta_{14}\ll 1\quad\mbox{for all }j.$$
$\epsilon$ will be much smaller than any combinations of the other constants.
\end{remark}
We will use energy estimates and decay estimates to eventually close the bootstrap. In order to derive the estimates, we consider equations for $\Gamma^k\Phi$. We now introduce the notations that will facilitate the discussion below.
\begin{definition}
Define $G_k=\displaystyle\sum_{|j|=k}|\Box_{g_K}\left(\Gamma^j\Phi\right)|$, $U_k=\displaystyle\sum_{|j|=k}|[\Box_{g_K},\Gamma^j]\Phi|$ and $N_k=\displaystyle\sum_{|j|=k}|\Gamma^j\left(\Box_{g_K}\Phi\right)|$. 
\end{definition}
\begin{definition}
Define $G_{\leq k}=\displaystyle\sum_{|j|\leq k}|\Box_{g_K}\left(\Gamma^j\Phi\right)|$, $U_{\leq k}=\displaystyle\sum_{|j|\leq k}|[\Box_{g_K},\Gamma^j]\Phi|$ and $N_{\leq k}=\displaystyle\sum_{|j|\leq k}|\Gamma^j\left(\Box_{g_K}\Phi\right)|$.
\end{definition}
In order to keep track of the constants, we define also
\begin{definition}
Define $U_{k,j}=|[\Box_{g_K},\partial_{t^*}^{k-j}\tilde{\Omega}^j]\Phi|$ and $U_{\leq k,\leq j}=\displaystyle\sum_{j'\leq j, k'\leq k} U_{k',j'}$.
\end{definition}

\begin{remark}
 We will refer to $G$ as the inhomogeneous term, $U$ as the commutator term and $N$ as the nonlinear term. Clearly we have $G_k\leq U_k+N_k$ and $G_{\leq k}\leq U_{\leq k}+N_{\leq k}$
\end{remark}

We now estimate the inhomogeneous terms that will appear in the analysis several times below. It is necessary to study the commutator terms and the nonlinear terms together because the estimates for each depend on the estimates for the other when we use elliptic estimates.
\begin{proposition}\label{Uestprop}
$U_k$ satisfies the following estimates:
\begin{equation}\label{Uest}
\begin{split}
&\int_{\Sigma_\tau}r^{\alpha}\left(D^\ell U_{k,j}\right)^2\\
\leq& C\left(\sum_{m=0}^{k+\ell-i}\sum_{i=0}^{j-1}\int_{\Sigma_\tau\cap\{r\geq R_\Omega-1\}}r^{\alpha-4}J^{N}_\mu\left(\partial_{t^*}^m\tilde{\Omega}^i\Phi\right)n^\mu_{\Sigma_\tau}+\sum_{m=0}^{k+\ell-i-1}\sum_{i=0}^{j-1}\int_{\Sigma_\tau}r^{\alpha-4}\left(D^m N_i\right)^2\right)\quad\mbox{for }\alpha\leq 4,
\end{split}
\end{equation}
where it is understood that $\displaystyle\sum_{i=0}^{-1}=0$.
\end{proposition}
\begin{proof}
The commutator terms are estimated in \cite{LS}. Notice that since $\tilde{\Omega}$ is supported away from the trapped set, there is no loss of derivatives in using the integrated decay estimate. We have that $U_{k,j}$ supported in $\{r\geq R_\Omega\}$ and that
$$|U_{k,j}|\leq C\sum_{i=0}^j|\partial_{t^*}^{k-j}[\Box_{g_K},\tilde{\Omega}^i]\Phi|\leq C\sum_{i=0}^{j-1}r^{-2}\left(|D^2\partial_{t^*}^{k-j}\tilde{\Omega}^i\Phi|+|D\partial_{t^*}^{k-j}\tilde{\Omega}^i\Phi|\right),$$
and therefore
$$|D^\ell U_{k,j}|\leq C\sum_{m=1}^{\ell+2}\sum_{i=0}^{j-1}r^{-2}|D^m\partial_{t^*}^{k-j}\tilde{\Omega}^i\Phi|\leq C\sum_{m=1}^{k+\ell-j+2}\sum_{i=0}^{j-1}r^{-2}|D^m\tilde{\Omega}^i\Phi|,$$
where, as in the statement of the Proposition it is understood that the sum vanishes if $j=0$.
Hence, using the elliptic estimate for $\{r\geq R_\Omega\}$, i.e., Proposition \ref{ellipticoutside},
\begin{equation}\label{Uinduct}
\begin{split}
&\int_{\Sigma_\tau}r^{\alpha}\left(D^\ell U_{k,j}\right)^2\\
\leq&C\sum_{m=1}^{k+\ell-j+2}\sum_{i=0}^{j-1}\int_{\Sigma_\tau\cap\{r\geq R_\Omega\}}r^{\alpha-4}\left(D^m\tilde{\Omega}^i\Phi\right)^2\\
\leq&C\sum_{m=0}^{k+\ell-i}\sum_{i=0}^{j-1}\int_{\Sigma_\tau\cap\{r\geq R_\Omega-1\}}r^{\alpha-4}J^N_\mu\left(\partial_{t^*}^m\tilde{\Omega}^i\Phi\right)n^\mu_{\Sigma_{\tau}}\\
&+C\sum_{m=0}^{k+\ell-i-1}\sum_{i=0}^{j-1}\int_{\Sigma_\tau\cap\{r\geq R_\Omega-1\}}r^{\alpha-4}\left(\left(D^{m}U_{i,i}\right)^2+\left(D^{m}N_{i}\right)^2\right).\\
\end{split}
\end{equation}
Now we can estimate $U_k$ by induction:  Fix any $k$ and we will induct on $j$. By definition, $U_{k,0}=0$. Now, assume that for all $k+\ell\leq 16$ and for some $j_0\geq 1$, we have
\begin{equation*}
\begin{split}
&\sum_{k+\ell\leq M, j\leq min\{j_0-1,k\}}\int_{\Sigma_\tau}r^{\alpha}\left(D^\ell U_{k,j}\right)^2\\
\leq& C\left(\sum_{m=0}^{M-i}\sum_{i=0}^{j_0-2}\int_{\Sigma_\tau\cap\{r\geq R_\Omega-1\}}r^{\alpha-4}J^{N}_\mu\left(\partial_{t^*}^m\tilde{\Omega}^i\Phi\right)n^\mu_{\Sigma_\tau}+\sum_{m=0}^{M-1-i}\sum_{i=0}^{j_0-1}\int_{\Sigma_\tau}r^{\alpha-4}\left(D^m N_i\right)^2\right).
\end{split}
\end{equation*}
for all $\alpha\leq 4$. Then, using (\ref{Uinduct}), we have that for $k+\ell\leq 16$, and $j_0\leq k$,
\begin{equation*}
\begin{split}
&\int_{\Sigma_\tau}r^{\alpha}\left(D^\ell U_{k,j_0}\right)^2\\
\leq& C\left(\sum_{m=0}^{k+\ell-i}\sum_{i=0}^{j_0-1}\int_{\Sigma_\tau\cap\{r\geq R_\Omega-1\}}r^{\alpha-4}J^{N}_\mu\left(\partial_{t^*}^m\tilde{\Omega}^i\Phi\right)n^\mu_{\Sigma_\tau}+\sum_{m=0}^{k+\ell-i-1}\sum_{i=0}^{j_0-1}\int_{\Sigma_\tau}r^{\alpha-4}\left(D^m N_i\right)^2\right).
\end{split}
\end{equation*}
Hence, (\ref{Uest}) holds.
\end{proof}

We now estimate the nonlinear term $N_k$. Since $N_k$ is at least quadratic, we do not need to be precise about the constants $A$ and we will always estimate with the maximum $A$.
\begin{proposition}\label{Nestprop}
$N_k$ satisfies the following estimates for fixed $t^*=\tau$:
\begin{equation*}
\begin{split}
\sum_{k+\ell= 16}\int_{\Sigma_\tau}\left(D^\ell N_k\right)^2\leq CBA^2\epsilon^2\tau^{-2+\eta_{16}},
\end{split}
\end{equation*}
\begin{equation*}
\begin{split}
\sum_{k+\ell\leq 15}\int_{\Sigma_\tau}\left(D^\ell N_k\right)^2\leq CBA^2\epsilon^2\tau^{-2}.
\end{split}
\end{equation*}
$N_k$ also satisfy the following estimates for when integrated over $t^*\in [(1.1)^{-1}\tau,\tau]$:
\begin{equation*}
\begin{split}
\sum_{k+\ell= 15}\iint_{\mathcal R((1.1)^{-1}\tau,\tau)}r^{-1-\delta}\left(D^\ell N_k\right)^2\leq CBA^2\epsilon^2\tau^{-2+\eta_{16}},
\end{split}
\end{equation*}
\begin{equation*}
\begin{split}
\sum_{k+\ell\leq 14}\iint_{\mathcal R((1.1)^{-1}\tau,\tau)}r^{-1-\delta}\left(D^\ell N_k\right)^2\leq CBA^2\epsilon^2\tau^{-2}.
\end{split}
\end{equation*}
\end{proposition}
\begin{proof}
Here, there is no need to distinguish between the good and bad derivatives.

$$|D^\ell N_k|\leq |D^\ell\Gamma^k\left(\Lambda_i D\Phi D\Phi\right)|+|\Gamma^k \mathcal C|.$$
We claim that the most important terms will be those that are quadratic in $D^{j+1}\Gamma^i\Phi$ or cubic of the form $$\left(D^{j_1+1}\Gamma^{i_1}\Phi \right)\left(D^{j_2+1}\Gamma^{i_2}\Phi\right)\left(\Gamma^{i_3}\Phi\right)$$ with $i_1+j_1,i_2+j_2\leq 8$. For by the assumptions every term has the form
$$\left(D^{j_1}\Gamma^{i_1}\Phi \right)\left(D^{j_2}\Gamma^{i_2}\Phi\right)\left(D^{j_3}\Gamma^{i_3}\Phi\right)\left(D^{j_4}\Gamma^{i_4}\Phi\right)...\left(D^{j_r}\Gamma^{i_r}\Phi\right),$$
with $r\geq 2$, at least two $j$'s $\geq 1$ and $i+j\leq 9$ for all but at most one factor. If all factors $i+j\leq 9$ or the factor that $i+j>9$ has $i\geq 1$, we can reduce to the case $D^{j_1+1}\Gamma^{i_1}\Phi D^{j_2+1}\Gamma^{i_2}\Phi$ by putting all other factors in $L^\infty$ using bootstrap assumptions (\ref{BAP1}), (\ref{BAP1.5}), (\ref{BAPI1}) and (\ref{BAPI1.1}). If the factor that $i+j>9$ has $i=0$ we reduce to 
$$\left(D^{j_1+1}\Gamma^{i_1}\Phi \right)\left(D^{j_2+1}\Gamma^{i_2}\Phi\right)\left(\Gamma^{i_3}\Phi\right)$$
again by putting all other factors in $L^\infty$ using bootstrap assumptions (\ref{BAP1}), (\ref{BAP1.5}), (\ref{BAPI1}) and (\ref{BAPI1.1}). 
\begin{equation*}
\begin{split}
&\int_{\Sigma_\tau}\left(D^\ell N_k\right)^2\\
\leq &C\left(\sup\sum_{i_1+j_1\leq 8}|D^{j_1+1}\Gamma^{i_1}\Phi|^2\right)\sum_{j_2=0}^\ell\sum_{i_2=0}^k\int_{\Sigma_\tau}| D^{j_2+1}\Gamma^{i_2}\Phi|^2\\
&+C\left(\sup\sum_{i_1+j_1\leq 8}|D^{j_1+1}\Gamma^{i_1}\Phi|^2\right)\left(\sup\sum_{i_2+j_2\leq 8}r^2| D^{j_2+1}\Gamma^{i_2}\Phi|^2\right)\sum_{i_3=0}^k\int_{\Sigma_\tau}r^{-2}| \Gamma^{i_3}\Phi|^2\\
\leq&CBA\epsilon\tau^{-2}\sum_{i=1}^{\ell+1}\sum_{j=0}^k\int_{\Sigma_\tau} \left(D^i\Gamma^j\Phi\right)^2\\
&\quad\mbox{using Hardy's inequality Proposition \ref{Hardy}}\\
\leq&CBA\epsilon\tau^{-2}\sum_{i+m=0}^{\ell}\sum_{j=0}^k\int_{\Sigma_\tau} J^N_\mu\left(\partial_{t^*}^m\Gamma^j\hat{Y}^i\Phi\right)n^\mu_{\Sigma_{\tau}}+CBA\epsilon\tau^{-2}\sum_{i+j\leq k+\ell-1}\int_{\Sigma_\tau}\left(\left(D^iU_{\leq j}\right)^2+\left(D^iN_{\leq j}\right)^2\right)\\
&\quad\mbox{using the elliptic estimates in Propositions \ref{elliptic}, \ref{elliptichorizon}}\\
\leq&CBA\epsilon\tau^{-2}\sum_{i=0}^{\ell}\sum_{j=0}^k\int_{\Sigma_\tau} J^N_\mu\left(\partial_{t^*}^m\Gamma^j\hat{Y}^i\Phi\right)n^\mu_{\Sigma_{\tau}}+CBA\epsilon\tau^{-2}\sum_{i+j\leq k+\ell-1}\int_{\Sigma_\tau}\left(D^iN_{\leq j}\right)^2,\\
\end{split}
\end{equation*}

where we have used Proposition \ref{Uestprop} in the last step.
Now, a simple induction would show that 
\begin{equation*}
\begin{split}
\sum_{k+\ell= 16}\int_{\Sigma_\tau}\left(D^\ell N_k\right)^2\leq CBA^2\epsilon^2\tau^{-2+\eta},\quad\mbox{and}
\end{split}
\end{equation*}
\begin{equation*}
\begin{split}
\sum_{k+\ell\leq 15}\int_{\Sigma_\tau}\left(D^\ell N_k\right)^2\leq CBA^2\epsilon^2\tau^{-2}.
\end{split}
\end{equation*}
We now move on the the terms integrated over $t^*\in [(1.1)^{-1}\tau,\tau]$. Arguing as before, and noticing that the elliptic estimate in Proposition \ref{elliptic} also allows weights in $r$, we have
\begin{equation*}
\begin{split}
&\iint_{\mathcal R((1.1)^{-1}\tau,\tau)}r^{-1-\delta}\left(D^\ell N_k\right)^2\\
\leq&CBA\epsilon\tau^{-2}\sum_{i+m=0}^{\ell}\sum_{j=0}^k\iint_{\mathcal R((1.1)^{-1}\tau,\tau)} r^{-1-\delta}J^N_\mu\left(\partial_{t^*}^m\Gamma^j\hat{Y}^i\Phi\right)n^\mu_{\Sigma_{\tau}}\\
&+CBA\epsilon\tau^{-2}\sum_{i+j\leq k+\ell-1}\iint_{\mathcal R((1.1)^{-1}\tau,\tau)}r^{-1-\delta}\left(\left(D^iU_{\leq j}\right)^2+\left(D^iN_{\leq j}\right)^2\right)\\
\leq&CBA\epsilon\tau^{-2}\sum_{i=0}^{\ell}\sum_{j=0}^k\iint_{\mathcal R((1.1)^{-1}\tau,\tau)} r^{-1-\delta}J^N_\mu\left(\partial_{t^*}^m\Gamma^j\hat{Y}^i\Phi\right)n^\mu_{\Sigma_{\tau}}\\
&+CBA\epsilon\tau^{-2}\sum_{i+j\leq k+\ell-1}\iint_{\mathcal R((1.1)^{-1}\tau,\tau)}r^{-1-\delta}\left(D^iN_{\leq j}\right)^2\\
\leq&CBA\epsilon\tau^{-2}\sum_{i=0}^{\ell}\sum_{j=0}^k\iint_{\mathcal R((1.1)^{-1}\tau,\tau)} K^{X_1}\left(\partial_{t^*}^m\Gamma^j\hat{Y}^i\Phi\right)\\
&+CBA\epsilon\tau^{-2}\sum_{i+j\leq k+\ell-1}\iint_{\mathcal R((1.1)^{-1}\tau,\tau)}r^{-1-\delta}\left(D^iN_{\leq j}\right)^2\\
\end{split}
\end{equation*}
Now, the bootstrap assumptions (\ref{BAK1.4}), (\ref{BAK1.6}) and an induction on $k+\ell$ would conclude the Proposition.
\end{proof}

Now, the estimates for $N_k$ will also give improved estimates for $U_k$ via Proposition \ref{Uestprop}:
\begin{proposition}\label{U}
The following estimates for $U_k$ on a fixed $t^*$ slice hold for $\alpha\leq 2$:
$$\sum_{k+\ell=16}\int_{\Sigma_\tau}r^{\alpha}\left(D^\ell U_{k,j}\right)^2\leq CA_{j-1}\epsilon\tau^{\eta_{16}},$$
$$\sum_{k+\ell=15}\int_{\Sigma_\tau}r^{\alpha}\left(D^\ell U_{k,j}\right)^2\leq CA_{j-1}\epsilon\tau^{-1+\eta_{15}},$$
$$\sum_{k+\ell\leq 14}\int_{\Sigma_\tau}r^{\alpha}\left(D^\ell U_{k,j}\right)^2\leq CA_{j-1}\epsilon\tau^{-2+\eta_{14}}.$$
The following estimates for $U_k$ integrated on $[(1.1)^{-1}\tau,\tau]$ also hold for $\alpha\leq 1+\delta$:
$$\sum_{k+\ell=16}\iint_{\mathcal R((1.1)^{-1}\tau,\tau)}r^{\alpha}\left(D^\ell U_{k,j}\right)^2\leq CA_{X,j-1}\epsilon\tau^{\eta_{16}},$$
$$\sum_{k+\ell=15}\iint_{\mathcal R((1.1)^{-1}\tau,\tau)}r^{\alpha}\left(D^\ell U_{k,j}\right)^2\leq CA_{X,j-1}\epsilon\tau^{-1+\eta_{15}},$$
$$\sum_{k+\ell\leq 14}\iint_{\mathcal R((1.1)^{-1}\tau,\tau)}r^{\alpha}\left(D^\ell U_{k,j}\right)^2\leq CA_{X,j-1}\epsilon\tau^{-2+\eta_{14}}.$$
\end{proposition}
\begin{proof}
We first prove the Proposition for the estimates for the constant in $\tau$ terms. By Proposition \ref{Uestprop},
\begin{equation*}
\begin{split}
&\int_{\Sigma_\tau}r^{\alpha}\left(D^\ell U_{k,j}\right)^2\\
\leq& C\left(\sum_{m=0}^{k+\ell-i}\sum_{i=0}^{j-1}\int_{\Sigma_\tau\cap\{r\geq R_\Omega-1\}}r^{\alpha-4}J^{N}_\mu\left(\partial_{t^*}^m\tilde{\Omega}^i\Phi\right)n^\mu_{\Sigma_\tau}+\sum_{m=0}^{k+\ell-1-i}\sum_{i=0}^{j-1}\int_{\Sigma_\tau}r^{\alpha-4}\left(D^m N_i\right)^2\right)\quad\mbox{for }\alpha\leq 4.
\end{split}
\end{equation*}
The second term satisfy the required estimate by Proposition \ref{Nestprop}.
We estimate the first term. By (\ref{BA1}),
$$\sum_{m=0}^{16-i}\sum_{i=0}^{j-1}\int_{\Sigma_\tau\cap\{r\geq R_\Omega-1\}}r^{\alpha-4}J^{N}_\mu\left(\partial_{t^*}^m\tilde{\Omega}^i\Phi\right)n^\mu_{\Sigma_\tau}\leq CA_{j-1}\epsilon\tau^{\eta_{16}}.$$
By (\ref{BA2}) and Proposition \ref{Zlowerbound},
\begin{equation*}
\begin{split}
&\sum_{m=0}^{15-i}\sum_{i=0}^{j-1}\int_{\Sigma_\tau\cap\{r\geq R_\Omega-1\}}r^{\alpha-4}J^{N}_\mu\left(\partial_{t^*}^m\tilde{\Omega}^i\Phi\right)n^\mu_{\Sigma_\tau}\\
\leq&C\left(\sum_{m=0}^{15-i}\sum_{i=0}^{j-1}\int_{\Sigma_\tau\cap\{r\leq \frac{\tau}{2}\}}J^{N}_\mu\left(\partial_{t^*}^m\tilde{\Omega}^i\Phi\right)n^\mu_{\Sigma_\tau}+\sum_{m=0}^{15-i}\sum_{i=0}^{j-1}\tau^{-2}\int_{\Sigma_\tau\cap\{r\geq \frac{\tau}{2}\}}J^{N}_\mu\left(\partial_{t^*}^m\tilde{\Omega}^i\Phi\right)n^\mu_{\Sigma_\tau}\right)\\
\leq&CA_{j-1}\epsilon\tau^{-1+\eta_{15}}.
\end{split}
\end{equation*}
By (\ref{BA3}) and Proposition \ref{Zlowerbound}, 
\begin{equation*}
\begin{split}
&\sum_{m=0}^{14-i}\sum_{i=0}^{j-1}\int_{\Sigma_\tau\cap\{r\geq R_\Omega-1\}}r^{\alpha-4}J^{N}_\mu\left(\partial_{t^*}^m\tilde{\Omega}^i\Phi\right)n^\mu_{\Sigma_\tau}\\
\leq&C\left(\sum_{m=0}^{14-i}\sum_{i=0}^{j-1}\int_{\Sigma_\tau\cap\{r\leq \frac{\tau}{2}\}}J^{N}_\mu\left(\partial_{t^*}^m\tilde{\Omega}^i\Phi\right)n^\mu_{\Sigma_\tau}+\sum_{m=0}^{14-i}\sum_{i=0}^{j-1}\tau^{-2}\int_{\Sigma_\tau\cap\{r\geq \frac{\tau}{2}\}}J^{N}_\mu\left(\partial_{t^*}^m\tilde{\Omega}^i\Phi\right)n^\mu_{\Sigma_\tau}\right)\\
\leq&CA_{j-1}\epsilon\tau^{-2+\eta_{14}}.
\end{split}
\end{equation*}
For the integrated terms, we similarly have, by Proposition \ref{Uestprop},
\begin{equation*}
\begin{split}
&\iint_{\mathcal R((1.1)^{-1}\tau,\tau)}r^{\alpha}\left(D^\ell U_{k,j}\right)^2\\
\leq& C\left(\sum_{m=0}^{k+\ell-i}\sum_{i=0}^{j-1}\iint_{\mathcal R((1.1)^{-1}\tau,\tau)\cap\{r\geq R_\Omega-1\}}r^{\alpha-4}J^{N}_\mu\left(\partial_{t^*}^m\tilde{\Omega}^i\Phi\right)n^\mu_{\Sigma_\tau}\right.\\
&\left.\quad\quad+\sum_{m=0}^{k+\ell-1-i}\sum_{i=0}^{j-1}\iint_{\mathcal R((1.1)^{-1}\tau,\tau)}r^{\alpha-4}\left(D^m N_i\right)^2\right)\quad\mbox{for }\alpha\leq 4.
\end{split}
\end{equation*}
The second term can be estimated by Proposition \ref{Nestprop}. Notice that $r^{-1+\delta}J^N_\mu\left(\Phi\right)n^\mu_{\Sigma_\tau}\leq K^{X_0}\left(\Phi\right)$. Hence, following the argument above for the fixed $\tau$ case, we would have proved the Proposition for the case $\alpha\leq 1-\delta$. Nevertheless, with more care, we can improve to $\alpha\leq 1+\delta$. 
\begin{equation*}
\begin{split}
&\sum_{m=0}^{k+\ell-i}\sum_{i=0}^{j-1}\iint_{\mathcal R((1.1)^{-1}\tau,\tau)\cap\{r\geq R_\Omega\}}r^{\alpha-4}J^{N}_\mu\left(\partial_{t^*}^m\tilde{\Omega}^i\Phi\right)n^\mu_{\Sigma_{t^*}}\\
\leq&C\sum_{m=0}^{k+\ell-i}\sum_{i=0}^{j-1}\iint_{\mathcal R((1.1)^{-1}\tau,\tau)\cap\{r\leq \frac{t^*}{2}\}}K^{X_0}\left(\partial_{t^*}^m\tilde{\Omega}^i\Phi\right) \\
&+C\sum_{m=0}^{k+\ell-i}\sum_{i=0}^{j-1}\iint_{\mathcal R((1.1)^{-1}\tau,\tau)\cap\{r\geq \frac{t^*}{2}\}}r^{-2+2\delta}K^{X_0}\left(\partial_{t^*}^m\tilde{\Omega}^i\Phi\right).
\end{split}
\end{equation*}
For $k+\ell=16$, this is bounded by $CA_{X,j-1}\epsilon\tau^{\eta_{16}}$ by (\ref{BAK1}). For $k+\ell=15$, this is bounded by $CA_{X,j-1}\epsilon\tau^{-1+\eta_{15}}$ by (\ref{BAK2}) and (\ref{BAK1.5}). For $k+\ell\leq 14$, this is bounded by $CA_{X,j-1}\epsilon\tau^{-2+\eta_{14}}$ by (\ref{BAK4}) and (\ref{BAK1.5}) since $2\delta\leq \eta_{14}$.
\end{proof}

While the above is sufficient to recover the bootstrap assumptions for the pointwise bounds, we will need improvements to achieve the energy bounds. For the improvements, we study separately the region $r\leq\frac{t^*}{4}$, $\frac{t^*}{4}\leq r\leq\frac{9t^*}{10}$ and $r\geq\frac{9t^*}{10}$. For the region $r\geq\frac{t^*}{4}$, we will only show the improvement for $N_k$ instead of the derivatives of $N_k$. Various complications would arise in estimating the derivatives of $N_k$. For the region $r\leq\frac{t^*}{4}$, however, we will estimate also the derivatives of $N_k$ as they will be necessary to estimate the error terms arising from commuting with the red-shift vector field.

\begin{proposition}
\begin{equation*}
\sum_{\ell+k=16}\int_{\Sigma_\tau\cap\{r\leq\frac{\tau}{4}\}}r^{1-\delta}\left(D^\ell N_k\right)^2 \leq CBA^2\epsilon^2\tau^{-3+\eta_{S,{11}}+\eta_{16}}.
\end{equation*}
\begin{equation*}
\sum_{\ell+k=15}\int_{\Sigma_\tau\cap\{r\leq\frac{\tau}{4}\}}r^{1-\delta}\left(D^\ell N_k\right)^2 \leq CBA^2\epsilon^2\tau^{-4+\eta_{S,{11}}+\eta_{15}}.
\end{equation*}
\begin{equation*}
\sum_{\ell+k\leq 14}\int_{\Sigma_\tau\cap\{r\leq\frac{\tau}{4}\}}r^{1-\delta}\left(D^\ell N_k\right)^2 \leq CBA^2\epsilon^2\tau^{-5+\eta_{S,{11}}+\eta_{14}}.
\end{equation*}
\end{proposition}
\begin{proof}
As before, we only have to estimate terms quadratic in $D^j\Gamma^i\Phi$ with $j\geq 1$ or cubic of the form $$\left(D^{j_1+1}\Gamma^{i_1}\Phi \right)\left(D^{j_2+1}\Gamma^{i_2}\Phi\right)\left(\Gamma^{i_3}\Phi\right)$$ with $i_1+j_1,i_2+j_2\leq 8$.
\begin{equation}\label{Nlocal}
\begin{split}
&\int_{\Sigma_\tau}r^{1-\delta}\left(D^\ell N_k\right)^2\\
\leq &C\left(\sup_{r\leq\frac{\tau}{4}}\sum_{i_1+j_1\leq 8}r^{1-\delta}|D^{j_1+1}\Gamma^{i_1}\Phi|^2\right)\sum_{j_2=0}^\ell\sum_{i_2=0}^k\int_{\Sigma_\tau\cap\{r\leq\frac{\tau}{4}\}}| D^{j_2+1}\Gamma^{i_2}\Phi|^2\\
&+C\left(\sup_{r\leq\frac{\tau}{4}}\sum_{i_1+j_1\leq 8}r^{1-\delta}|D^{j_1+1}\Gamma^{i_1}\Phi|^2\right)^2\sum_{i_3=0}^k\tau^{2\delta}\int_{\Sigma_\tau\cap\{r\leq\frac{\tau}{4}\}}r^{-2}| \Gamma^{i_3}\Phi|^2\\
\leq&CBA\epsilon\tau^{-3+\eta_{S,11}}\sum_{i=1}^{\ell+1}\sum_{j=0}^k\int_{\Sigma_\tau\cap\{r\leq\frac{\tau}{4}\}} \left(D^i\Gamma^j\Phi\right)^2+CBA^2\epsilon^2\tau^{-6+2\eta_{S,11}+2\delta}\sum_{i_3=0}^k\int_{\Sigma_\tau}\left( D\Gamma^{i_3}\Phi\right)^2\\
&\mbox{by Hardy's inequality. Notice now that the second term has more decay than we need,}\\
&\mbox {so we will drop it from now on.}\\
\leq&CBA\epsilon\tau^{-3+\eta_{S,11}}\sum_{i+m=0}^{\ell}\sum_{j=0}^k\int_{\Sigma_\tau\cap\{r\leq\frac{\tau}{2}\}} J^N_\mu\left(\partial_{t^*}^m\Gamma^j\hat{Y}^i\Phi\right)n^\mu_{\Sigma_{\tau}}\\
&+CBA\epsilon\tau^{-3+\eta_{S,11}}\sum_{i=0}^{\ell-1}\sum_{j=0}^k\int_{\Sigma_\tau}\left(\left(D^iU_{j}\right)^2+\left(D^iN_{j}\right)^2\right)\\
\leq&CBA\epsilon\tau^{-3+\eta_{S,11}}\sum_{i+j\leq k+\ell}\int_{\Sigma_\tau} J^N_\mu\left(\Gamma^j\hat{Y}^i\Phi\right)n^\mu_{\Sigma_{\tau}}+CBA\epsilon\tau^{-3+\eta_{S,11}}\sum_{i+j\leq k+\ell-1}\int_{\Sigma_\tau}\left(D^iN_{j}\right)^2,\\
\end{split}
\end{equation}
The Proposition would follow from an induction on $k+\ell$ and the bootstrap assumptions (\ref{BA1}), (\ref{BA2}), (\ref{BA3}). The $k+\ell=0$ case also follows from the above computation as we have adopted the notation that $\displaystyle\sum_{i+j\leq -1} =0$.
\end{proof}
We now move to the region $\{\frac{t^*}{4}\leq r\leq\frac{9t^*}{10}\}$. In this region, $u\sim t^*$, and therefore we can exploit the decay in the variable $u$ given by the estimates from the conformal energy.
\begin{proposition}
$$\int_{\Sigma_\tau\cap\{\frac{\tau}{4}\leq r\leq\frac{9\tau}{10}\}}N_{16}^2\leq CBA^2\epsilon^2\tau^{-4+\eta_{14}+\eta_{16}},$$
$$\int_{\Sigma_\tau\cap\{\frac{\tau}{4}\leq r\leq\frac{9\tau}{10}\}}N_{15}^2\leq CBA^2\epsilon^2\tau^{-5+\eta_{14}+\eta{15}},$$
and
$$\sum_{j=0}^{14}\int_{\Sigma_\tau\cap\{\frac{\tau}{4}\leq r\leq\frac{9\tau}{10}\}}N_{j}^2\leq CBA^2\epsilon^2\tau^{-6+2\eta_{14}}.$$
\end{proposition}

\begin{proof}
Arguing as before, we see that the main terms for the nonlinearity are those that are quadratic in $D\Phi$ or those that are cubic with the form $\Gamma^{i_3}\Phi D\Gamma^{i_1}\Phi D\Gamma^{i_2}\Phi$ with $i_1, i_2\leq 8$. The quadratic terms can be estimated:
\begin{equation*}
\begin{split}
&\sum_{i_1=0}^{\lfloor\frac{k}{2}\rfloor}\sum_{i_2=0}^{k}\int_{\Sigma_\tau\cap\{\frac{\tau}{4}\leq r\leq\frac{9\tau}{10}\}}|{D}\Gamma^{i_1}\Phi D\Gamma^{i_2}\Phi|^2\\
\leq &C \left(\sum_{i_1=0}^8 \sup_{\frac{\tau}{4}\leq r\leq\frac{9\tau}{10}}|{D}\Gamma^{i_1}\Phi|^2\right)\sum_{i_2=0}^{k}\int_{\Sigma_\tau\cap\{\frac{\tau}{4}\leq r\leq\frac{9\tau}{10}\}}| D\Gamma^{i_2}\Phi|^2\\
\leq&CAB\epsilon\tau^{-4+\eta_{14}}\sum_{i=0}^{k}\int_{\Sigma_\tau\cap\{\frac{\tau}{4}\leq r\leq\frac{9\tau}{10}\}}| D\Gamma^{i}\Phi|^2.
\end{split}
\end{equation*}
The particular cubic term can be estimated as follows: 
\begin{equation*}
\begin{split}
&\sum_{i_1, i_2=0}^{\lfloor\frac{k}{2}\rfloor}\sum_{i_3=0}^{k}\int_{\Sigma_\tau\cap\{\frac{\tau}{4}\leq r\leq\frac{9\tau}{10}\}}|\Gamma^{i_3}\Phi D\Gamma^{i_1}\Phi D\Gamma^{i_2}\Phi|^2\\
\leq &C \left(\sum_{i_1=0}^8 \sup_{\frac{\tau}{4}\leq r\leq\frac{9\tau}{10}}|{D}\Gamma^{i_1}\Phi|^2\right)^2\sum_{i_3=0}^{k}\int_{\Sigma_\tau\cap\{\frac{\tau}{4}\leq r\leq\frac{9\tau}{10}\}}| \Gamma^{i_3}\Phi|^2\\
\leq&CA^2B^2\epsilon^2\tau^{-8+2\eta_{14}}\sum_{i=0}^{k}\int_{\Sigma_\tau\cap\{\frac{\tau}{4}\leq r\leq\frac{9\tau}{10}\}}| \Gamma^{i}\Phi|^2.
\end{split}
\end{equation*}
In principle, for $k\leq 15$, we can then control the last term using the conformal energy. For $k=16$, however, conformal energy is not available, and we need to use Hardy's inequality.
\begin{equation*}
\begin{split}
&A\epsilon\tau^{-8+2\eta_{14}}\sum_{i=0}^{k}\int_{\Sigma_\tau\cap\{\frac{\tau}{4}\leq r\leq\frac{9\tau}{10}\}}| \Gamma^{i}\Phi|^2\\
\leq &CA\epsilon\tau^{-6+2\eta_{14}}\sum_{i=0}^{k}\int_{\Sigma_\tau}r^{-2}| \Gamma^{i}\Phi|^2\\
\leq &CA\epsilon\tau^{-6+2\eta_{14}}\sum_{i=0}^{k}\int_{\Sigma_\tau} J^N_\mu\left(\Gamma^i\Phi\right) n^\mu_{\Sigma_\tau}.
\end{split}
\end{equation*}
The estimates now follow from the Bootstrap Assumptions (\ref{BA1}), (\ref{BA2}), (\ref{BA3}).
\end{proof}
For many applications, we only need a much weaker estimate on $N_k$. We write down the following Proposition which corresponds to the estimates that will be proved for the quantities involving $S$. This would allow a unified approach in dealing with many estimates with or without $S$.
\begin{proposition}\label{NI}
\begin{equation*}
\begin{split}
\int_{\Sigma_\tau\cap\{r\leq\frac{9\tau}{10}\}}r^{1-\delta}N_{16}^2\leq CA^2\epsilon^2\tau^{-3+\eta_{S,11}+\eta_{16}}
\end{split}
\end{equation*}
\begin{equation*}
\begin{split}
\int_{\Sigma_\tau\cap\{r\leq\frac{9\tau}{10}\}}r^{1-\delta}N_{\leq 15}^2\leq CA^2\epsilon^2\tau^{-4+\eta_{S,11}+\eta_{15}}
\end{split}
\end{equation*}
\end{proposition}
We now move to the estimates for $N_k$ in the region $\{r\geq \frac{9t^*}{10}\}$. Here, we need to exploit the null condition:
\begin{proposition}\label{NO}
For $\alpha=0$ or $2$,
$$\int_{\Sigma_\tau\cap\{r\geq\frac{9\tau}{10}\}}N_{16}^2\leq CBA^2\epsilon^2\tau^{-2+\eta_{16}},$$
$$\int_{\Sigma_\tau\cap\{r\geq\frac{9\tau}{10}\}}r^{\alpha}N_{15}^2\leq CBA^2\epsilon^2\tau^{-3+\alpha+\eta_{15}},$$
and
$$\sum_{j=0}^{14}\int_{\Sigma_\tau\cap\{r\geq\frac{9\tau}{10}\}}r^{\alpha}N_j^2\leq CBA^2\epsilon^2\tau^{-4+\alpha+\eta_{14}}.$$
\end{proposition}
\begin{proof}
Following the argument before, we reduce to quadratic and cubic terms. This time, however, the null condition plays a crucial role. For the quadratic terms, we need to consider
$$\left(\bar{D}\Gamma^{i_1}\Phi D\Gamma^{i_2}\Phi\right),\left(D\Gamma^{i_1}\Phi \bar{D}\Gamma^{i_2}\Phi\right),r^{-1} \left(D\Gamma^{i_1}\Phi D\Gamma^{i_2}\Phi\right),$$
where $i_1\geq i_2$.
For the cubic terms, we need to consider
$$\left({D}\Gamma^{i_1}\Phi D\Gamma^{i_2}\Phi D\Gamma^{i_3}\Phi\right), \left(\bar{D}\Gamma^{i_1}\Phi D\Gamma^{i_2}\Phi \Gamma^{i_3}\Phi\right).$$
Notice that for the first cubic term can be dominated pointwise by quadratic terms of the third type listed above using the bootstrap assumptions (\ref{BAP2}) and (\ref{BAPI2}). The second cubic term can also be dominated by the first two types of quadratic terms if $i_3\leq 13$ by (\ref{BAP1}) and (\ref{BAPI1}). We can thus assume $i_3>13$ and hence $i_1, i_2\leq 8$. We now estimate the quadratic terms
\begin{equation*}
\begin{split}
&\sum_{i_2=0}^{\lfloor\frac{k}{2}\rfloor}\sum_{i_1=0}^{k-j}\int_{\Sigma_\tau\cap\{r\geq\frac{9\tau}{10}\}}r^{\alpha}\left(|\bar{D}\Gamma^{i_1}\Phi D\Gamma^{i_2}\Phi|^2+|D\Gamma^{i_1}\Phi \bar{D}\Gamma^{i_2}\Phi|^2+r^{-2} |D\Gamma^{i_1}\Phi D\Gamma^{i_2}\Phi|^2\right)\\
\leq&C\left(\sup_{r\geq\frac{9\tau}{10}}\sum_{i_2=0}^{8}r^2|D\Gamma^{i_2}\Phi|^2\right)\left(\sum_{i_1=0}^{k}\int_{\Sigma_\tau\cap\{r\geq\frac{9\tau}{10}\}} r^{\alpha-2}|\bar{D}\Gamma^{i_1}\Phi|^2\right)\\
&+C\left(\sup_{r\geq\frac{9\tau}{10}}\sum_{i_2=0}^{8}r^2|\bar{D}\Gamma^{i_2}\Phi|^2\right)\left(\sum_{i_1=0}^{k}\int_{\Sigma_\tau\cap\{r\geq\frac{9\tau}{10}\}}r^{\alpha-2}| D\Gamma^{i_1}\Phi|^2\right)\\
&+C\tau^{-2}\left(\sup_{r\geq\frac{9\tau}{10}}\sum_{i_2=0}^{8}r^2|D\Gamma^{i_2}\Phi|^2\right)\left(\sum_{i_1=0}^{k}\int_{\Sigma_\tau\cap\{r\geq\frac{9\tau}{10}\}} r^{\alpha-2}|{D}\Gamma^{i_1}\Phi|^2\right)\\
\leq &CAB\epsilon\left(\sum_{i=0}^{k}\int_{\Sigma_\tau\cap\{r\geq\frac{9\tau}{10}\}} r^{\alpha-2}|\bar{D}\Gamma^{i}\Phi|^2\right)+CAB\epsilon\tau^{-2+\eta_{14}}\left(\sum_{i=0}^{k}\int_{\Sigma_\tau\cap\{r\geq\frac{9\tau}{10}\}} r^{\alpha-2}|{D}\Gamma^{i}\Phi|^2\right).
\end{split}
\end{equation*}
We then estimate the particular cubic term:
\begin{equation*}
\begin{split}
&\sum_{i_3=0}^{k}\sum_{i_1, i_2=0}^{8}\int_{\Sigma_\tau\cap\{r\geq\frac{9\tau}{10}\}}r^{\alpha}\left(\bar{D}\Gamma^{i_1}\Phi D\Gamma^{i_2}\Phi \Gamma^{i_3}\Phi\right)^2\\
\leq&C\left(\sup_{r\geq\frac{9\tau}{10}}\sum_{i_1=0}^{8}r^2|D\Gamma^{i_2}\Phi|^2\right)\left(\sup_{r\geq\frac{9\tau}{10}}\sum_{i_2=0}^{8}r^2|\bar{D}\Gamma^{i_2}\Phi|^2\right)\left(\sum_{i_3=0}^{k}\int_{\Sigma_\tau\cap\{r\geq\frac{9\tau}{10}\}} r^{\alpha-4}|\Gamma^{i_1}\Phi|^2\right)\\
\leq &CA^2B^2\epsilon^2\tau^{-2+\eta_{14}}\left(\sum_{i=0}^{k}\int_{\Sigma_\tau\cap\{r\geq\frac{9\tau}{10}\}} r^{\alpha-4}|\Gamma^{i}\Phi|^2\right).
\end{split}
\end{equation*}
Therefore,
\begin{equation*}
\begin{split}
&\int_{\Sigma_\tau\cap\{r\geq\frac{9\tau}{10}\}}N_{16}^2\\
\leq &CAB\epsilon\tau^{-2}\left(\sum_{i=0}^{16}\int_{\Sigma_\tau\cap\{r\geq\frac{9\tau}{10}\}} |\bar{D}\Gamma^{i}\Phi|^2\right)+CAB\epsilon\tau^{-4+\eta_{14}}\left(\sum_{i=0}^{16}\int_{\Sigma_\tau\cap\{r\geq\frac{9\tau}{10}\}} |{D}\Gamma^{i}\Phi|^2\right)\\
&+CA^2B^2\epsilon^2\tau^{-4+\eta_{14}}\left(\sum_{i=0}^{16}\int_{\Sigma_\tau\cap\{r\geq\frac{9\tau}{10}\}} r^{-2}|\Gamma^{i}\Phi|^2\right)\\
\leq &\left(CAB\epsilon\tau^{-2}+\left(CAB\epsilon+CA^2B^2\epsilon^2\right)\tau^{-4+\eta_{14}}\right)\left(\sum_{i=0}^{16}\int_{\Sigma_\tau\cap\{r\geq\frac{9\tau}{10}\}} |{D}\Gamma^{i}\Phi|^2\right)\\
\leq &CA^2B\epsilon^2\tau^{-2+\eta_{16}}.
\end{split}
\end{equation*}
and
\begin{equation*}
\begin{split}
&\int_{\Sigma_\tau\cap\{r\geq\frac{9\tau}{10}\}}N_k^2\\
\leq &CAB\epsilon\tau^{-4}\left(\sum_{i=0}^{k}\int_{\Sigma_\tau\cap\{r\geq\frac{9\tau}{10}\}} \tau^2|\bar{D}\Gamma^{i}\Phi|^2\right)+CAB\epsilon\tau^{-4+\eta_{14}}\left(\sum_{i=0}^{k}\int_{\Sigma_\tau\cap\{r\geq\frac{9\tau}{10}\}} |{D}\Gamma^{i}\Phi|^2\right)\\
&+CA^2B^2\epsilon^2\tau^{-4+\eta_{14}}\left(\sum_{i=0}^{k}\int_{\Sigma_\tau\cap\{r\geq\frac{9\tau}{10}\}} r^{-2}|\Gamma^{i}\Phi|^2\right)\\
\leq &CAB\epsilon\tau^{-4}\left(\sum_{i=0}^{k}\int_{\Sigma_\tau\cap\{r\geq\frac{9\tau}{10}\}} \tau^2|\bar{D}\Gamma^{i}\Phi|^2\right)+\left(CAB\epsilon+CA^2B^2\epsilon^2\right)\tau^{-4+\eta_{14}}\left(\sum_{i=0}^{k}\int_{\Sigma_\tau} |{D}\Gamma^{i}\Phi|^2\right)\\
\end{split}
\end{equation*}
and
\begin{equation*}
\begin{split}
&\int_{\Sigma_\tau\cap\{r\geq\frac{9\tau}{10}\}}r^{2}N_k^2\\ 
\leq &CAB\epsilon\tau^{-2}\left(\sum_{i=0}^{k}\int_{\Sigma_\tau\cap\{r\geq\frac{9\tau}{10}\}} \tau^2|\bar{D}\Gamma^{i}\Phi|^2\right)+CAB\epsilon\tau^{-2+\eta_{14}}\left(\sum_{i=0}^{k}\int_{\Sigma_\tau\cap\{r\geq\frac{9\tau}{10}\}} |{D}\Gamma^{i}\Phi|^2\right)\\
&+CA^2B^2\epsilon^2\tau^{-2+\eta_{14}}\left(\sum_{i=0}^{k}\int_{\Sigma_\tau\cap\{r\geq\frac{9\tau}{10}\}} r^{-2}|\Gamma^{i}\Phi|^2\right)\\
\leq &CAB\epsilon\tau^{-2}\left(\sum_{i=0}^{k}\int_{\Sigma_\tau\cap\{r\geq\frac{9\tau}{10}\}} \tau^2|\bar{D}\Gamma^{i}\Phi|^2\right)+\left(CAB\epsilon+CA^2B^2\epsilon^2\right)\tau^{-2+\eta_{14}}\left(\sum_{i=0}^{k}\int_{\Sigma_\tau} |{D}\Gamma^{i}\Phi|^2\right)\\
\end{split}
\end{equation*}
The conclusion follows from Proposition \ref{Zlowerbound} and the bootstrap assumptions (\ref{BA1}), (\ref{BA2}), (\ref{BA3}).
\end{proof}

From the the estimates for $U_k$ and $N_k$ and the $L^2-L^\infty$ estimates in the last section, we get the following pointwise bounds:
\begin{proposition}\label{pointwise}
For $r\geq\frac{t^*}{4}$,
\begin{equation}\label{P1}
\sum_{j=0}^{13}|\Gamma^j\Phi|^2\leq \frac{B}{2}A\epsilon r^{-2}(t^*)^{1+\eta_{15}}.
\end{equation}
\begin{equation}\label{P1.5}
\sum_{j=0}^{13}|D\Gamma^j\Phi|^2\leq \frac{B}{2}A\epsilon.
\end{equation}
\begin{equation}\label{P2}
\sum_{\ell=1}^{13-j}\sum_{j=0}^{12}|D^\ell\Gamma^j\Phi|^2\leq \frac{B}{2}A\epsilon r^{-2}.
\end{equation}
\begin{equation}\label{P3}
\sum_{j=0}^8|D\Gamma^j\Phi|^2\leq \frac{B}{2}A\epsilon r^{-2}(t^*)^{\eta_{14}}(1+|u|)^{-2}.
\end{equation}
\begin{equation}\label{P4}
\sum_{j=0}^8|\bar{D}\Gamma^{j}\Phi|^2\leq \frac{B}{2}A\epsilon r^{-2}(t^*)^{-2+\eta_{14}}.
\end{equation}
For $r\leq\frac{t^*}{4}$,
\begin{equation}\label{PI1}
\sum_{j=0}^{13}|\Gamma^j\Phi|^2\leq \frac{B}{2}A\epsilon (t^*)^{-1+\eta_{15}}.
\end{equation}
\begin{equation}\label{PI1.1}
\sum_{\ell=1}^{14-j}\sum_{j=0}^{13}|D^\ell\Gamma^j\Phi|^2\leq BA\epsilon (t^*)^{-1+\eta_{15}}.
\end{equation}
\begin{equation}\label{PI1.2}
\sum_{\ell=1}^{13-j}\sum_{j=0}^{12}|D^\ell\Gamma^j\Phi|^2\leq BA\epsilon (t^*)^{-2+\eta_{14}}.
\end{equation}
\begin{equation}\label{PI2}
\sum_{\ell=1}^{9-j}\sum_{j=0}^8|D^\ell\Gamma^j\Phi|^2\leq \frac{B}{2}A\epsilon r^{-1+\delta}(t^*)^{-3+\eta_{S,11}}.
\end{equation}
\end{proposition}
\begin{proof}
(\ref{P1}) is immediate from Proposition \ref{rnoderivatives} and the bootstrap assumptions (\ref{BA2}) and (\ref{BA3}).

By Proposition \ref{SE},
\begin{equation*}
\begin{split}
\sum_{j=0}^{13}|D\Gamma^j\Phi|^2
\leq& C\left(\sum_{k=0}^{15}\int_{\Sigma_\tau} J^{N}_\mu\left(\Gamma^k\Phi\right) n^\mu_{\Sigma_\tau}+\sum_{k=0}^1\int_{\Sigma_\tau} \left(D^kG_{\leq 13}\right)^2\right).
\end{split}
\end{equation*}
Hence we get (\ref{P1.5}) by bootstrap assumption (\ref{BA4}) and Propositions \ref{Nestprop} and \ref{U}. The constant is improved since $A\epsilon\ll 1$ and $C\ll B$.

By Proposition \ref{r}, 
\begin{equation*}
\begin{split}
\sum_{\ell=1}^{13-j}\sum_{j=0}^{12}|D^\ell\Gamma^j\Phi|^2\leq &Cr^{-2}\left(\sum_{m=0}^{13-j}\sum_{k=0}^2\sum_{j=0}^{12}\int_{\Sigma_\tau} J^{N}_\mu\left(\partial_{t^*}^m\Omega^k\Gamma^j\Phi\right) n^\mu_{\Sigma_\tau}+\sum_{m+k\leq 10}\int_{\Sigma_\tau}\left(D^m G_{\leq k}\right)^2\right)\\
\leq&Cr^{-2}\left(\sum_{j=0}^{11}\int_{\Sigma_\tau} J^{N}_\mu\left(\Gamma^j\Phi\right) n^\mu_{\Sigma_\tau}+\sum_{m+k\leq 10}\int_{\Sigma_\tau}\left(D^m G_{\leq k}\right)^2\right)\\
\end{split}
\end{equation*}
We hence get (\ref{P2}) by bootstrap assumption (\ref{BA4}) and Propositions \ref{Nestprop} and \ref{U}. The constant is improved since $A\epsilon\ll 1$ and $C\ll B$.

By Proposition \ref{ru},
for $r\geq\frac{t^*}{4}$, we have 
\begin{equation*}
\begin{split}
&\sum_{j=0}^8|D\Gamma^j\Phi|^2\\
\leq& Cr^{-2}\left(1+|u|\right)^{-2}\sum_{m=0}^{1}\sum_{k=0}^2\sum_{j=0}^8\left(\int_{\Sigma_\tau}J^{Z+CN}_\mu\left(\partial^m_{t^*}\tilde{\Omega}^k\Gamma^j\Phi\right) n^\mu_{\Sigma_\tau}+C\tau^2\int_{\Sigma_\tau\cap\{r\leq r^-_Y\}}J^{N}_\mu\left(\partial^m_{t^*}\tilde{\Omega}^k\Gamma^j\Phi\right) n^\mu_{\Sigma_\tau}\right)\\
&+Cr^{-2}\sum_{k=0}^2\sum_{j=0}^8\int_{\Sigma_\tau\cap\{u'\sim u\}\cap\{r\geq\frac{\tau}{4}\}} \left(\Box_{g_K}\left(\tilde{\Omega}^k\Gamma^j\Phi\right)\right)^2\\
\leq& CA\epsilon\tau^{\eta_{14}}r^{-2}\left(1+|u|\right)^{-2}+Cr^{-2}\int\int_{\Sigma_\tau\cap\{u'\sim u\}\cap\{r\geq\frac{\tau}{4}\}} G_{\leq 10}^2\\
&\quad\mbox{by bootstrap assumption (\ref{BA3})}\\
\leq& CA\epsilon\tau^{\eta_{14}}r^{-2}\left(1+|u|\right)^{-2}+CA\epsilon(t^*)^{-2+\eta_{14}}r^{-2}+CA^2B\epsilon^2(t^*)^{-2+\eta_{14}}r^{-2}\\
&\quad\mbox{by Propositions \ref{NI}, \ref{NO} and \ref{U}}\\
\leq &\frac{B}{2}A\epsilon r^{-2}(t^*)^{\eta_{14}}(1+|u|)^{-2}.
\end{split}
\end{equation*}
Hence we have proved (\ref{P3}).

By Proposition \ref{rv}, for $r\geq\frac{t^*}{4}$, we have
\begin{equation*}
\begin{split}
\sum_{j=0}^8|\bar{D}\Gamma^j\Phi|^2\leq&C r^{-4}\sum_{k=0}^2\sum_{j=0}^8\sum_{i+m\leq 1}\left(\int_{\Sigma_{\tau}} J^N_\mu\left(S^i\partial_{t^*}^m\Gamma^j\Phi\right)n^\mu_{\Sigma_\tau}+\int_{\Sigma_\tau}J^{Z+CN}_\mu\left(\partial^m_{t^*}\tilde{\Omega}^k\Gamma^j\Phi\right) n^\mu_{\Sigma_\tau}\right.\\
&\left.\quad\quad\quad\quad\quad\quad+C\tau^2\int_{\Sigma_\tau\cap\{r\leq r^-_Y\}}J^{N}_\mu\left(\partial^m_{t^*}\tilde{\Omega}^k\Gamma^j\Phi\right) n^\mu_{\Sigma_\tau}+\int_{\Sigma_{\tau}} \left(\Box_{g_K}\left(\tilde{\Omega}^k\Gamma^j\Phi\right)\right)^2\right)\\
&+Cr^{-2}\sum_{k=0}^2\sum_{j=0}^8\int_{\Sigma_\tau\cap\{r\geq\frac{\tau}{2}\}} \left(\Box_{g_K}\left(\tilde{\Omega}^k\Gamma^j\Phi\right)\right)^2\\
\leq &CA\epsilon r^{-4}(t^*)^{\eta_{14}}+CA\epsilon(t^*)^{-2+\eta_{14}}r^{-2}+CA^2B\epsilon^2(t^*)^{-2+\eta_{14}}r^{-2}\\
\leq &\frac{B}{2}A\epsilon r^{-2}(t^*)^{-2+\eta_{14}}.
\end{split}
\end{equation*}
Hence we have proved (\ref{P4}) and completed the proof for $r\geq\frac{t^*}{4}$. We now move to the pointwise estimates in the region $r\leq\frac{t^*}{4}$. (\ref{PI1}) follows directly from Proposition \ref{SEinside} and the bootstrap assumptions (\ref{BA2}) and (\ref{BA3}). By Proposition \ref{sSobolev},
\begin{equation*}
\begin{split}
\sum_{\ell=1}^{14-j}\sum_{j=0}^{13}|D^\ell\Gamma^j\Phi|^2\leq C\left(\sum_{i+j\leq 15}\int_{\Sigma_\tau\cap\{r\leq\frac{t^*}{2}\}}J^{N}_\mu\left(\hat{Y}^i \Gamma^j\Phi\right)n^\mu_{\Sigma_\tau}+\sum_{\ell=1}^{14-j}\sum_{j=0}^{13}\int_{\Sigma_\tau}(D^\ell G_{\leq j})^2\right),
\end{split}
\end{equation*}
where we have used the fact that $[\hat{Y},\Gamma]=0$. Hence (\ref{PI1.1}) follows from the bootstrap assumptions (\ref{BA2}) and (\ref{BA3}) and Propositions \ref{Nestprop} and \ref{U}. The proof of (\ref{PI1.2}) follows similarly as (\ref{PI1.1}): by Proposition \ref{sSobolev}, 
\begin{equation*}
\begin{split}
\sum_{\ell=1}^{13-j}\sum_{j=0}^{12}|D^\ell\Gamma^j\Phi|^2\leq C\left(\sum_{i+j\leq 14}\int_{\Sigma_\tau\cap\{r\leq\frac{t^*}{2}\}}J^{N}_\mu\left(\hat{Y}^i \Gamma^j\Phi\right)n^\mu_{\Sigma_\tau}+\sum_{\ell=1}^{13-j}\sum_{j=0}^{13}\int_{\Sigma_\tau}(D^\ell G_{\leq j})^2\right),
\end{split}
\end{equation*}
Hence (\ref{PI1.2}) follows from the bootstrap assumptions (\ref{BA3}) and Propositions \ref{Nestprop} and \ref{U}.

Finally, by Proposition \ref{extradecay}, for $r\leq\frac{t^*}{4}$, we have
\begin{equation*}
\begin{split}
\sum_{\ell=1}^{9-j}\sum_{j=0}^8|D^\ell\Gamma^j\Phi|^2\leq &C(t^*)^{-1}r^{-1+\delta}\sum_{i+j\leq 10}\iint_{\mathcal R((1.1)^{-1}\tau,\tau)\cap\{r\leq\frac{t^*}{2}\}}\left(K^{X_1}\left(Y^i\Gamma^j\Phi\right)+K^{X_1}\left(SY^i\Gamma^j\Phi\right)\right)\\
&+C(t^*)^{-1}r^{-1+\delta}\sum_{i+j\leq 10}\iint_{\mathcal R((1.1)^{-1}\tau,\tau)\cap\{r \leq\frac{t^*}{2}\}}r^{-1-\delta}\left(D^iG_j\right)^2\\
\leq&CA\epsilon\tau^{-3+\eta_{S,11}}r^{-1+\delta}+CA^2\epsilon^2\tau^{-3+\eta_{14}}r^{-1+\delta}\\
\leq&\frac{B}{2}A\epsilon\tau^{-3+\eta_{S,11}}r^{-1+\delta},
\end{split}
\end{equation*}
where in the third line we have used the bootstrap assumptions (\ref{BAK5}) and (\ref{BAK10}) and Propositions \ref{NI}, \ref{NO} and \ref{U}. The only caveat is that when using (\ref{BAK10}), the vector fields $\hat{Y}$ and $S$ are in different order. However, since $[S,\hat{Y}]\sim D$, we can estimate the commutator term by (\ref{BAK5}).
\end{proof}
Now, we have proved the $L^{\infty}$ bounds, we will replace the constant $B$ in the bootstrap assumption (\ref{BAP1}), (\ref{BAP1.5}), (\ref{BAP2}), (\ref{BAP3}), (\ref{BAP4}), (\ref{BAPI1}), (\ref{BAPI2}) by $C$ in the sequel. Notice that we have originally assumed $B\ll A_0$ and therefore $C\ll A_0$ still holds. We now proceed to recover the bootstrap assumptions (K) that do not involve the commutators $Y$ or $S$. We first retrieve (\ref{BAK2})-(\ref{BAK5}).  Notice also that we will retrieve (\ref{BAK1}) and (\ref{BAK1.5}) later together with (\ref{BA1}) and (\ref{BA4}).

\begin{proposition}\label{K}
\begin{equation}\label{K2}
 \sum_{i+j\leq 15}A_{X,j}^{-1}\iint_{\mathcal R((1.1)^{-1}\tau,\tau)\cap\{r\leq\frac{t^*}{2}\}}\left(K^{X_0}\left(\partial_{t^*}^{i}\tilde{\Omega}^j\Phi\right)+K^{N}\left(\partial_{t^*}^{i}\tilde{\Omega}^j\Phi\right)\right)\leq\frac{\epsilon}{2}\tau^{-1+\eta_{15}}.
\end{equation}
\begin{equation}\label{K3}
\sum_{i+j\leq 14}A_{X,j}^{-1}\iint_{\mathcal R((1.1)^{-1}\tau,\tau)\cap\{r\leq\frac{t^*}{2}\}} K^{X_1}\left(\partial_{t^*}^{i}\tilde{\Omega}^j\Phi\right)\leq\frac{\epsilon}{2}\tau^{-1+\eta_{15}}.
\end{equation}
\begin{equation}\label{K4}
 \sum_{i+j\leq 14}A_{X,j}^{-1}\iint_{\mathcal R((1.1)^{-1}\tau,\tau)\cap\{r\leq\frac{t^*}{2}\}}\left(K^{X_0}\left(\partial_{t^*}^{i}\tilde{\Omega}^j\Phi\right)+K^{N}\left(\partial_{t^*}^{i}\tilde{\Omega}^j\Phi\right)\right)\leq\frac{\epsilon}{2}\tau^{-2+\eta_{14}}.
\end{equation}
\begin{equation}\label{K5}
 \sum_{i+j\leq 13}A_{X,j}^{-1}\iint_{\mathcal R((1.1)^{-1}\tau,\tau)\cap\{r\leq\frac{t^*}{2}\}}K^{X_1}\left(\partial_{t^*}^{i}\tilde{\Omega}^j\Phi\right)\leq\frac{\epsilon}{2}\tau^{-2+\eta_{14}}.
\end{equation}
\end{proposition}
\begin{proof}
We first prove the estimates involving $X_0$, i.e.,  (\ref{K2}) and (\ref{K4}). By Proposition \ref{X0}.1 and the Remark following it and the fact that $|\partial_{t^*}^mN_k|\leq |N_{\leq k+m}|$, we have
\begin{equation*}
\begin{split}
& \sum_{i+j\leq 15}A_{X,j}^{-1}\iint_{\mathcal R((1.1)^{-1}\tau,\tau)\cap\{r\leq\frac{t^*}{2}\}}\left(K^{X_0}\left(\partial_{t^*}^{i}\tilde{\Omega}^j\Phi\right)+K^{N}\left(\partial_{t^*}^{i}\tilde{\Omega}^j\Phi\right)\right)\\
\leq& C\sum_{i+j\leq 15}A_{X,j}^{-1}\left(\tau^{-2}\int_{\Sigma_{(1.1)^{-1}\tau}} J^{Z+N,w^Z}_\mu\left(\partial_{t^*}^{i}\tilde{\Omega}^j\Phi\right)n^\mu_{\Sigma_{(1.1)^{-1}\tau}}+ C\int_{\Sigma_{(1.1)^{-1}\tau}\cap\{r\leq r^-_Y\}} J^{N}_\mu\left(\partial_{t^*}^{i}\tilde{\Omega}^j\Phi\right)n^\mu_{\Sigma_{(1.1)^{-1}\tau}}\right)\\
&+C\sum_{i+j\leq 15}A_{X,j}^{-1}\left(\iint_{\mathcal R((1.1)^{-1}\tau-1,\tau+1)\cap\{r\leq\frac{9t^*}{10}\}}r^{1+\delta}N_{\leq 16}^2+\sup_{t^*\in [(1.1)^{-1}\tau-1,\tau+1]}\int_{\Sigma_{t^*}\cap\{|r-3M|\leq\frac{M}{8}\}} N_{\leq 16}^2\right.\\
&\left.+\iint_{\mathcal R((1.1)^{-1}\tau-1,\tau+1)\cap\{r\leq\frac{9t^*}{10}\}}r^{1+\delta}\left( U_{\leq 15,\leq j}\right)^2\right)\\
\leq& C\sum_{i+j\leq 15}\left(A_{X,j}^{-1}A_j\epsilon\tau^{-1+\eta_{15}}+A_{X,j}^{-1}A^2\epsilon^2\tau^{-2+\eta_{S,11}+\eta_{16}+2\delta}+CA_{X,j}^{-1}A_{X,j-1}\epsilon\tau^{-1+\eta_{15}}\right)\\
\leq&\frac{\epsilon}{2}\tau^{-1+\eta_{15}},
\end{split}
\end{equation*}
by Propositions \ref{NI}, \ref{NO} and \ref{U}. Notice that our integrated estimates for $U$ in Proposition \ref{U} is only for $[(1.1)^{-1}\tau,\tau]$. Nevertheless, for the region $[(1.1)^{-1}\tau-1,(1.1)^{-1}\tau]\cap[\tau,\tau+1]$, we can use the integrate over the fixed $\tau$ estimate in the same Proposition. By Proposition \ref{X0}.1 and the Remark following it and the fact that $|\partial_{t^*}^mN_k|\leq |N_{\leq k+m}|$, we have
\begin{equation*}
\begin{split}
& \sum_{i+j\leq 14}A_{X,j}^{-1}\iint_{\mathcal R((1.1)^{-1}\tau,\tau)\cap\{r\leq\frac{t^*}{2}\}}\left(K^{X_0}\left(\partial_{t^*}^{i}\tilde{\Omega}^j\Phi\right)+K^{N}\left(\partial_{t^*}^{i}\tilde{\Omega}^j\Phi\right)\right)\\
\leq& C\sum_{i+j\leq 14}A_{X,j}^{-1}\left(\tau^{-2}\int_{\Sigma_{(1.1)^{-1}\tau}} J^{Z+N,w^Z}_\mu\left(\partial_{t^*}^{i}\tilde{\Omega}^j\Phi\right)n^\mu_{\Sigma_{(1.1)^{-1}\tau}}+ C\int_{\Sigma_{(1.1)^{-1}\tau}\cap\{r\leq r^-_Y\}} J^{N}_\mu\left(\partial_{t^*}^{i}\tilde{\Omega}^j\Phi\right)n^\mu_{\Sigma_{(1.1)^{-1}\tau}}\right)\\
&+C\sum_{i+j\leq 14}A_{X,j}^{-1}\left(\iint_{\mathcal R((1.1)^{-1}\tau-1,\tau+1)\cap\{r\leq\frac{9t^*}{10}\}}r^{1+\delta}N_{\leq 15}^2+\sup_{t^*\in [(1.1)^{-1}\tau-1,\tau+1]}\int_{\Sigma_{t^*}\cap\{|r-3M|\leq\frac{M}{8}\}} N_{\leq 15}^2\right.\\
&\left.+\iint_{\mathcal R((1.1)^{-1}\tau-1,\tau+1)\cap\{r\leq\frac{9t^*}{10}\}}r^{1+\delta}\left(U_{\leq 14,\leq j}\right)^2\right)\\
\leq& C\sum_{i+j\leq 14}\left(A_{X,j}^{-1}A_j\epsilon\tau^{-2+\eta_{14}}+A_{X,j}^{-1}A^2\epsilon^2\tau^{-3+\eta_{S,11}+\eta_{15}+2\delta}+A_{X,j}^{-1}A_{X,j-1}\epsilon\tau^{-2+\eta_{14}}\right)\\
\leq&\frac{\epsilon}{2}\tau^{-2+\eta_{14}}.
\end{split}
\end{equation*}
The proof of (\ref{K3}) and (\ref{K5}) proceeds in an identical manner. Notice that using Proposition \ref{X0}.2, the right hand side when we estimate (\ref{K3}) (respectively (\ref{K5})) is identical to that when we estimate (\ref{K2}) (respectively (\ref{K4})).
\end{proof}

Now we move on to retrieving the bootstrap assumptions (J) with better constants:
\begin{proposition}\label{J}
\begin{equation}\label{J1}
  \sum_{i+j=16}A_j^{-1}\int_{\Sigma_\tau}J^{N}_\mu\left(\partial_{t^*}^i\tilde{\Omega}^j\Phi\right) n^{\mu}_{\Sigma_\tau}\leq \frac{\epsilon}{4} \tau^{\eta_{16}}.
\end{equation}
\begin{equation}\label{J4}
  \sum_{i+j\leq 15}A_j^{-1}\int_{\Sigma_\tau} J^{N}_\mu\left(\partial_{t^*}^i\tilde{\Omega}^j\Phi\right) n^{\mu}_{\Sigma_\tau}\leq \frac{\epsilon}{2}.
\end{equation}
\begin{equation}\label{K1.0.1}
 \sum_{i+j=16}A_{X,j}^{-1} \left(\iint_{\mathcal R(\tau_0,\tau)}K^{X_0}\left(\partial_{t^*}^i\tilde{\Omega}^j\Phi\right)+\iint_{\mathcal R(\tau_0,\tau)\cap\{r\leq r^-_Y\}}K^{N}\left(\partial_{t^*}^i\tilde{\Omega}^j\Phi\right)\right)\leq \frac{\epsilon}{2}\tau^{\eta_{16}}.
\end{equation}
\begin{equation}\label{K1.4}
\sum_{i+j=15}A_{X,j}^{-1}\iint_{\mathcal R(\tau_0,\tau)} K^{X_1}\left(\partial_{t^*}^i\tilde{\Omega}^j\Phi\right)\leq \frac{\epsilon}{2}\tau^{\eta_{16}}.
\end{equation}
\begin{equation}\label{K1.5}
 \sum_{i+j\leq 15}A_{X,j}^{-1}\iint_{\mathcal R(\tau_0,\tau)} K^{X_0}\left(\partial_{t^*}^i\tilde{\Omega}^j\Phi\right)\leq \frac{\epsilon}{2}.
\end{equation}
\begin{equation}\label{K1.6}
 \sum_{i+j\leq 14}A_{X,j}^{-1}\iint_{\mathcal R(\tau_0,\tau)} K^{X_1}\left(\partial_{t^*}^i\tilde{\Omega}^j\Phi\right)\leq \epsilon.
\end{equation}
\end{proposition}
\begin{proof}
We will prove the slightly stronger statements with $A_{X,j}$ replaced by $A_j$. Using Proposition \ref{bddcom} and \ref{bddcom2}, we have
\begin{equation*}
\begin{split}
&\sum_{i+j=16}A_{j}^{-1}\left(\int_{\Sigma_\tau}J^{N}_\mu\left(\partial_{t^*}^i\tilde{\Omega}^j\Phi\right) n^{\mu}_{\Sigma_\tau} +\iint_{\mathcal R(\tau_0,\tau)} K^{X_0}\left(\partial_{t^*}^i\tilde{\Omega}^j\Phi\right)+\iint_{\mathcal R(\tau_0,\tau)\cap\{r\leq r^-_Y\}}K^{N}\left(\partial_{t^*}^i\tilde{\Omega}^j\Phi\right)\right)\\
\leq &C\sum_{i+j=16}A_{j}^{-1}\left(\int_{\Sigma_{\tau_0}} J^{N}_\mu\left(\partial_{t^*}^{i}\tilde{\Omega}^j\Phi\right)n^\mu_{\Sigma_{\tau'}}+\left(\int_{\tau_0-1}^{\tau+1}\left(\int_{\Sigma_{t^*}} N_{16}^2\right)^{\frac{1}{2}}dt^*\right)^2+\iint_{\mathcal R(\tau_0-1,\tau+1)}N_{16}^2\right.\\
&\left.+\iint_{\mathcal R(\tau_0-1,\tau+1)}r^{1+\delta}U_{16,j}^2+\sup_{t^*\in [\tau_0-1,\tau+1]}\int_{\Sigma_{t^*}\cap\{|r-3M|\leq\frac{M}{8}\}} U_{16,j}^2\right)\\
\leq &C\sum_{i+j=16}A_{j}^{-1}\left(\epsilon+A^2\epsilon^2\eta_{16}^{-1}\tau^{\eta_{16}}+A_{X,j-1}\epsilon\tau^{\eta_{16}}\right)\\
\leq &\frac{\epsilon}{4}\tau^{\eta_{16}}.
\end{split}
\end{equation*}

We now turn to the estimates for $\displaystyle\sum_{j=0}^{15}|\Gamma^j\Phi|$. We have
\begin{equation*}
\begin{split}
&\sum_{i+j\leq 15}A_{j}^{-1}\left(\int_{\Sigma_\tau}J^{N}_\mu\left(\partial_{t^*}^i\tilde{\Omega}^j\Phi\right) n^{\mu}_{\Sigma_\tau} +\iint_{\mathcal R(\tau_0,\tau)} K^{X_0}\left(\partial_{t^*}^i\tilde{\Omega}^j\Phi\right)\right)\\
\leq &C\sum_{i+j\leq 15}A_j^{-1}\left(\int_{\Sigma_{\tau_0}} J^{N}_\mu\left(\partial_{t^*}^{i}\tilde{\Omega}^j\Phi\right)n^\mu_{\Sigma_{\tau_0}}+\left(\int_{\tau_0-1}^{\tau+1}\left(\int_{\Sigma_{t^*}} N_{\leq 15}^2\right)^{\frac{1}{2}}dt^*\right)^2+\iint_{\mathcal R(\tau_0-1,\tau+1)}N_{\leq 15}^2\right.\\
&\left.+\iint_{\mathcal R(\tau_0-1,\tau+1)}r^{1+\delta}U_{\leq 15,\leq j}^2+\sup_{t^*\in [\tau_0-1,\tau+1]}\int_{\Sigma_{t^*}\cap\{|r-3M|\leq\frac{M}{8}\}} U_{\leq 15,\leq j}^2\right)\\
\leq& C\sum_{i+j\leq 15}A_j^{-1}\left(\epsilon+A^2\epsilon^2+CA_{X,j-1}\epsilon\right)\\
\leq& \frac{\epsilon}{2}.
\end{split}
\end{equation*}
It now remains to show (\ref{K1.4}) and (\ref{K1.6}). By Proposition \ref{bddcom3} they can be estimated by exactly the same terms as (\ref{K1.0.1}) and (\ref{K1.5}) respectively. The Proposition hence follows.

\end{proof}
We now move on to control the conformal energy and close the part of the bootstrap assumption (\ref{BA2}) without $\hat{Y}$.
\begin{proposition}\label{Z1}
\begin{equation}\label{J2}
\begin{split}
 \sum_{i+j=15}A_j^{-1}\left(\int_{\Sigma_\tau} J^{Z+N,w^Z}_\mu\left(\partial_{t^*}^i\tilde{\Omega}^j\Phi\right) n^{\mu}_{\Sigma_\tau} +C\tau^2\int_{\Sigma_\tau\cap\{r\leq \frac{9\tau}{10}\}} J^{N}_\mu\left(\partial_{t^*}^i\tilde{\Omega}^j\Phi\right) n^{\mu}_{\Sigma_\tau}\right)\leq \frac{\epsilon}{4} \tau^{1+\eta_{15}}.
\end{split}
\end{equation}
\end{proposition}

\begin{proof}
By Proposition \ref{conformalenergy},
\begin{equation*}
\begin{split}
&\sum_{i+j= 15}A_j^{-1}\left(\int_{\Sigma_{\tau}} J^{Z,w^Z}_\mu\left(\partial_{t^*}^{i}\tilde{\Omega}^j\Phi\right)n^\mu_{\Sigma_{{\tau}}}+C\tau^2\int_{\Sigma_{\tau}\cap\{r\leq \frac{9\tau}{10}\}} J^{N}_\mu\left(\partial_{t^*}^{i}\tilde{\Omega}^j\Phi\right)n^\mu_{\Sigma_{\tau}}\right)\\
\leq &C\sum_{i+j= 15}A_j^{-1}\left(\int_{\Sigma_{\tau_0}} J^{Z+CN,w^Z}_\mu\left(\partial_{t^*}^{i}\tilde{\Omega}^j\Phi\right)n^\mu_{\Sigma_{{\tau_0}}}+\delta'\iint_{\mathcal R(\tau_0,\tau)\cap\{r\leq \frac{t^*}{2}\}} (t^*)^2K^{X_0}\left(\partial_{t^*}^{i}\tilde{\Omega}^j\Phi\right)\right.\\
&\left. +\left(\delta'+a\right)\iint_{\mathcal R(\tau_0,\tau)\cap\{r\leq r^-_Y\}}(t^*)^2K^N\left(\partial_{t^*}^{i}\tilde{\Omega}^j\Phi\right)+(\delta')^{-1}\iint_{\mathcal R(\tau_0,\tau)} t^* r^{-1+\delta}K^{X_1}\left(\partial_{t^*}^{i}\tilde{\Omega}^j\Phi\right)\right.\\
&\left.+(\delta')^{-1}\sum_{m=0}^1\iint_{\mathcal R(\tau_0,\tau)\cap\{r\leq\frac{9t^*}{10}\}} (t^*)^2r^{1+\delta}\left(\partial_{t^*}^m N_{ 15}\right)^2+(\delta')^{-1}\iint_{\mathcal R(\tau_0,\tau)\cap\{r\leq\frac{9t^*}{10}\}} (t^*)^2r^{1+\delta} U_{ 15,  j}^2\right. \\
&\left.+(\delta')^{-1}\left(\int_{\tau_0}^{\tau}\left(\int_{\Sigma_{t^*}\cap\{r\geq \frac{t^*}{2}\}}r^2 G_{ 15, j}^2 \right)^{\frac{1}{2}}dt^*\right)^2+(\delta')^{-1}\sup_{t^*\in [\tau_0,\tau]}\int_{\Sigma_{t^*}\cap\{r^-_Y\leq r\leq \frac{25M}{8}\}} (t^*)^2 N_{15}^2\right). \\
\end{split}
\end{equation*}
We will estimate the terms one by one.  First, the term with the initial data, i.e., the very first term, is clearly bounded by $C(\sum_j A_j^{-1})\epsilon$. Second, we consider the two $(t^*)^2K$ terms on the second line. To this end, we define as before $\tau_0\leq \tau_1\leq ...\leq \tau_n=\tau$ with $\tau_{i+1}\leq (1.1)\tau_i$ and $n\sim\log(\tau-\tau_0)$ is the minimum such that this can be done. Thus, these two terms can be bounded, using the bootstrap assumption (\ref{BAK2}),
\begin{equation*}
\begin{split}
 &\delta'\sum_{i+j=15}A_j^{-1}\iint_{\mathcal R(\tau_0,\tau)\cap\{r\leq \frac{t^*}{2}\}} (t^*)^2K^{X_0}\left(\partial_{t^*}^{i}\tilde{\Omega}^j\Phi\right)+\left(\delta'+a\right)\iint_{\mathcal R(\tau_0,\tau)\cap\{r\leq r^-_Y\}}(t^*)^2K^N\left(\partial_{t^*}^{i}\tilde{\Omega}^j\Phi\right)\\
\leq &C\sum_{i+j=15}A_j^{-1}\sum_{k=0}^{n-1}\left(\delta'\tau_k^2\iint_{\mathcal R(\tau_k,\tau_{k+1})\cap\{r\leq \frac{t^*}{2}\}} K^{X_0}\left(\partial_{t^*}^{i}\tilde{\Omega}^j\Phi\right)+\left(\delta'+a\right)\tau_k^2\iint_{\mathcal R(\tau_0,\tau)\cap\{r\leq r^-_Y\}}K^N\left(\partial_{t^*}^{i}\tilde{\Omega}^j\Phi\right)\right)\\
\leq& C\left(\sum_j\frac{A_{X,j}}{A_j}\right)\epsilon\left(2\delta'+a\right)\tau^{1+\eta_{15}}.
\end{split}
\end{equation*}
This is acceptable since $a,\delta'\ll\frac{A_j}{A_{X,j}}$.
Third, the term with $t^* r^{-1+\delta}K$ can be bounded using the bootstrap assumption (\ref{BAK1.4}),
\begin{equation*}
\begin{split}
&(\delta')^{-1}\sum_{i+j=15}A_j^{-1}\iint_{\mathcal R(\tau_0,\tau)} t^*  r^{-1+\delta}K^{X_1}\left(\partial_{t^*}^{i}\tilde{\Omega}^j\Phi\right) \\
\leq & C(\delta')^{-1}\sum_{i+j=15}A_j^{-1}\sum_{k=0}^{n-1}\tau_k\iint_{\mathcal R(\tau_k,\tau_{k+1})} K^{X_1}\left(\partial_{t^*}^{i}\tilde{\Omega}^j\Phi\right)
\leq C(\delta')^{-1}\left(\sum_j\frac{A_{X,j}}{A_j}\right)\tau^{1+\eta_{16}}.
\end{split}
\end{equation*}
This is acceptable since $\eta_{16}\ll\eta_{15}$ and therefore the constant can be improved for $\tau$ large.
Fourth, the integrals involving $N_{15}$ can be bounded using Propositions \ref{NI} and \ref{NO}:
\begin{equation*}
\begin{split}
&C(\delta')^{-1}A_0^{-1}\sum_{m=0}^1\iint_{\mathcal R(\tau_0,\tau)\cap\{r\leq\frac{9t^*}{10}\}} (t^*)^2r^{1+\delta}\left(\partial_{t^*}^m N_{15}\right)^2 \\
\leq &CA^2A_0^{-1}\epsilon^2(\delta')^{-1}\int_{\tau_0}^\tau \left(t^*\right)^{-1+\eta_{S,11}+\eta_{15}+2\delta}dt^*\leq CA^2A_0^{-1}\epsilon^2(\delta')^{-1}\tau^{\eta_{S,11}+\eta_{15}+2\delta},
\end{split}
\end{equation*}
$$C(\delta')^{-1}A_0^{-1}\left(\int_{\tau_0}^{\tau}\left(\int_{\Sigma_{t^*}\cap\{r\geq \frac{t^*}{2}\}}r^2 N_{15}^2 \right)^{\frac{1}{2}}dt^*\right)^2\leq CA^2A_0^{-1}\epsilon^2(\delta')^{-1}\tau^{1+\eta_{15}},$$
$$C(\delta')^{-1}A_0^{-1}\sup_{t^*\in [\tau_0,\tau]}\int_{\Sigma_{t^*}\cap\{r^-_Y\leq r\leq \frac{25M}{8}\}} (t^*)^2 N_{15}^2\leq CA^2A_0^{-1}\epsilon^2(\delta')^{-1}.$$
These are all acceptable since $\epsilon$ would beat all the constants.
Fifth, for the commutator terms $U_{15,j}^2$, we estimate by Proposition \ref{U}, 
$$(\delta')^{-1}A_j^{-1}\iint_{\mathcal R(\tau_0,\tau)} (t^*)^2r^{1+\delta}\left(U_{15,j}\right)^2\leq C\sum_{i=0}^{n-1}(\delta')^{-1}\iint_{\mathcal R(\tau_i,\tau_{i+1})} \tau_i^2r^{1+\delta}\left(U_{ 15, j}\right)^2\leq C(\delta')^{-1}\frac{A_{X,j-1}}{A_j}\tau^{1+\eta_{15}},$$
$$(\delta')^{-1}A_j^{-1}\left(\int_{\tau_0}^{\tau}\left(\int_{\Sigma_{t^*}\cap\{r\geq \frac{t^*}{2}\}}r^2 U_{ 15,j}^2 \right)^{\frac{1}{2}}dt^*\right)^2\leq C(\delta')^{-1}\frac{A_{j-1}}{A_j}\tau.$$
Since 
$$\frac{A_{j-1}}{A_j}\ll\frac{A_{X,j-1}}{A_j}\ll\delta',$$
all terms are acceptable.
\end{proof}

With $14$ or less derivatives, the conformal energy behaves better. We now close the part of the bootstrap assumption (\ref{BA3}) without $\hat{Y}$.
\begin{proposition}\label{Z2}
\begin{equation}\label{J3}
 \begin{split}
 \sum_{i+j\leq 14}A_j^{-1}\left(\int_{\Sigma_\tau} J^{Z+N,w^Z}_\mu\left(\partial_{t^*}^i\tilde{\Omega}^j\Phi\right) n^{\mu}_{\Sigma_\tau} +C\tau^2\int_{\Sigma_\tau\cap\{r\leq \frac{9\tau}{10}\}} J^{N}_\mu\left(\partial_{t^*}^i\tilde{\Omega}^j\Phi\right) n^{\mu}_{\Sigma_\tau}\right)\leq \frac{\epsilon}{4} \tau^{\eta_{14}}.
\end{split}
\end{equation}
\end{proposition}
\begin{proof}
By Proposition \ref{conformalenergy}, and noticing that $U$ is supported away from $\{|r-3M|\leq\frac{M}{8}\}$, we have
\begin{equation*}
\begin{split}
&\sum_{i+j\leq 14}A_j^{-1}\left(\int_{\Sigma_\tau} J^{Z+N,w^Z}_\mu\left(\partial_{t^*}^i\tilde{\Omega}^j\Phi\right) n^{\mu}_{\Sigma_\tau} +C\tau^2\int_{\Sigma_\tau\cap\{r\leq \frac{9\tau}{10}\}} J^{N}_\mu\left(\partial_{t^*}^i\tilde{\Omega}^j\Phi\right) n^{\mu}_{\Sigma_\tau}\right)\\
\leq &C\sum_{i+j\leq 14}A_j^{-1}\left(\int_{\Sigma_{\tau_0}} J^{Z+CN,w^Z}_\mu\left(\partial_{t^*}^{i}\tilde{\Omega}^j\Phi\right)n^\mu_{\Sigma_{{\tau_0}}}+\delta'\iint_{\mathcal R(\tau_0,\tau)\cap\{r\leq \frac{t^*}{2}\}} (t^*)^2K^{X_0}\left(\partial_{t^*}^{i}\tilde{\Omega}^j\Phi\right)\right.\\
&\left. +\left(\delta'+a\right)\iint_{\mathcal R(\tau_0,\tau)\cap\{r\leq r^-_Y\}}(t^*)^2K^N\left(\partial_{t^*}^{i}\tilde{\Omega}^j\Phi\right)+(\delta')^{-1}\iint_{\mathcal R(\tau_0,\tau)} t^* r^{-1+\delta}K^{X_1}\left(\partial_{t^*}^{i}\tilde{\Omega}^j\Phi\right)\right.\\
&\left.+(\delta')^{-1}\sum_{m=0}^1\iint_{\mathcal R(\tau_0,\tau)\cap\{r\leq\frac{9t^*}{10}\}} (t^*)^2r^{1+\delta}\left(\partial_{t^*}^m N_{\leq 15}\right)^2+(\delta')^{-1}\iint_{\mathcal R(\tau_0,\tau)\cap\{r\leq\frac{9t^*}{10}\}} (t^*)^2r^{1+\delta} U_{\leq 15, \leq  j}^2\right. \\
&\left.+(\delta')^{-1}\left(\int_{\tau_0}^{\tau}\left(\int_{\Sigma_{t^*}\cap\{r\geq \frac{t^*}{2}\}}r^2 G_{\leq 15,\leq j}^2 \right)^{\frac{1}{2}}dt^*\right)^2+(\delta')^{-1}\sup_{t^*\in [\tau_0,\tau]}\int_{\Sigma_{t^*}\cap\{r^-_Y\leq r\leq \frac{25M}{8}\}} (t^*)^2 N_{\leq 15}^2\right). \\
\end{split}
\end{equation*}
As before, we estimate each term one by one. First, the term with initial data is clearly bounded by $C\sum_jA_j\epsilon$. Second, the $(t^*)^2K$ terms can be bounded, using (\ref{BAK4}), and dividing the interval into $\tau_0< \tau_1<...<\tau_n=\tau$ as before by 
$$C\frac{A_{X,j}}{A_j}\epsilon\left(2\delta'+a\right)\sum_{i=0}^{n-1}\tau_i^{\eta_{14}}\leq C\frac{A_{X,j}}{A_j}\epsilon\eta_{14}^{-1}\left(2\delta+a\right)\tau^{\eta_{14}}.$$ 
This is acceptable since $a,\delta'\ll\frac{A_j}{A_{X,j}}$. Third, the $t^*r^{-1+\delta}K$ term can be bounded, using the bootstrap assumptions (\ref{BAK1.6}) and (\ref{BAK3}), by 
\begin{equation*}
\begin{split}
&(\delta')^{-1}\sum_{i+j\leq 14}A_j^{-1}\iint_{\mathcal R(\tau_0,\tau)} t^*  r^{-1+\delta}K^{X_1}\left(\partial_{t^*}^{i}\tilde{\Omega}^j\Phi\right) \\
\leq & (\delta')^{-1}\sum_{i+j\leq 14}A_j^{-1}\left(\iint_{\mathcal R(\tau_0,\tau)\cap\{r\leq\frac{t^*}{2}\}} t^*  K^{X_1}\left(\partial_{t^*}^{i}\tilde{\Omega}^j\Phi\right)+\iint_{\mathcal R(\tau_0,\tau)\cap\{r\geq \frac{t^*}{2}\}} t^*  r^{-1+\delta}K^{X_1}\left(\partial_{t^*}^{i}\tilde{\Omega}^j\Phi\right)\right) \\
\leq & C(\delta')^{-1}\sum_{i+j\leq 14}A_j^{-1}\sum_{k=0}^{n-1}\left(\tau_k\iint_{\mathcal R(\tau_k,\tau_{k+1})\cap\{r\leq\frac{t^*}{2}\}} K^{X_1}\left(\partial_{t^*}^{i}\tilde{\Omega}^j\Phi\right)+\tau_k^\delta\iint_{\mathcal R(\tau_k,\tau_{k+1})\cap\{r\geq \frac{t^*}{2}\}} K^{X_1}\left(\partial_{t^*}^{i}\tilde{\Omega}^j\Phi\right) \right) \\
\leq& C(\delta')^{-1}\sum_j\frac{A_{X,j}}{A_j}\eta_{15}^{-1}\tau^{\eta_{15}},
\end{split}
\end{equation*}
which is acceptable for $\tau$ large since $\eta_{15}\ll\eta_{14}$.

Fourth, the integrals involving $N_{\leq 14}$ can be bounded using Propositions \ref{NI} and \ref{NO} by noticing that $|\partial_{t^*}N_{\leq 14}|\leq CN_{\leq 15}$:
\begin{equation*}
\begin{split}
&C(\delta')^{-1}A_0^{-1}\sum_{m=0}^1\iint_{\mathcal R(\tau_0,\tau)\cap\{r\leq\frac{9t^*}{10}\}} (t^*)^2r^{1+\delta}\left(\partial_{t^*}^m N_{\leq 14}\right)^2 \\
\leq &CA^2A_0^{-1}\epsilon^2(\delta')^{-1}\int_{\tau_0}^\tau \left(t^*\right)^{-2+\eta_{S,11}+\eta_{15}+2\delta}dt^*\leq CA^2A_0^{-1}\epsilon^2(\delta')^{-1},
\end{split}
\end{equation*}
$$C(\delta')^{-1}A_0^{-1}\left(\int_{\tau_0}^{\tau}\left(\int_{\Sigma_{t^*}\cap\{r\geq \frac{t^*}{2}\}}r^2 N_{15}^2 \right)^{\frac{1}{2}}dt^*\right)^2\leq CA^2A_0^{-1}\epsilon^2(\delta')^{-1}\eta_{15}^{-1}\tau^{\eta_{15}},$$
$$C(\delta')^{-1}A_0^{-1}\sup_{t^*\in [\tau_0,\tau]}\int_{\Sigma_{t^*}\cap\{r^-_Y\leq r\leq \frac{25M}{8}\}} (t^*)^2 N_{15}^2\leq CA^2A_0^{-1}\epsilon^2(\delta')^{-1},$$
which is acceptable since $\epsilon\ll A\delta\eta_{15}^{-1}$.
Fifth, for the commutator terms $U_{\leq 14,\leq j}^2$, we estimate by Proposition \ref{U}, 
$$(\delta')^{-1}A_j^{-1}\iint_{\mathcal R(\tau_0,\tau)} (t^*)^2r^{1+\delta}\left(U_{\leq 14,\leq j}\right)^2\leq C\sum_{i=0}^{n-1}(\delta')^{-1}A_{X,j-1} \tau_i^2\tau_i^{-2+\eta_{14}}\leq C(\delta')^{-1}\eta_{14}^{-1}\frac{A_{X,j-1}}{A_j}\tau^{\eta_{14}},$$
$$(\delta')^{-1}A_j^{-1}\left(\int_{\tau_0}^{\tau}\left(\int_{\Sigma_{t^*}\cap\{r\geq \frac{t^*}{2}\}}r^2 U_{\leq 14,j}^2 \right)^{\frac{1}{2}}dt^*\right)^2\leq C(\delta')^{-1}\frac{A_{j-1}}{A_{j}}\eta_{14}^{-1}\tau.$$
Since 
$$\frac{A_{j-1}}{A_j}\ll\frac{A_{X,j-1}}{A_j}\ll\delta'\eta_{14}^{-1},$$
all terms are acceptable.
\end{proof}

We now consider terms involving commutation with $\hat{Y}$ and recover the bootstrap assumptions (\ref{BA1}), (\ref{BA2}) and (\ref{BA3}).
\begin{proposition}\label{Yest}
\begin{equation*}
  \sum_{i+k=16}\int_{\Sigma_\tau}J^{N}_\mu\left(\hat{Y}^k\partial_{t^*}^i\Phi\right) n^{\mu}_{\Sigma_\tau}\leq \frac{A_Y}{4}\epsilon \tau^{\eta_{16}},
\end{equation*}
and
\begin{equation*}
\begin{split}
 &\sum_{i+k=15}\tau^2\int_{\Sigma_\tau\cap\{r\leq r^-_Y\}}J^{N}_\mu\left(\hat{Y}^k\partial_{t^*}^i\Phi\right) n^{\mu}_{\Sigma_\tau} \leq \frac{A_Y}{4}\epsilon \tau^{1+\eta_{15}},
\end{split}
\end{equation*}
and
\begin{equation*}
 \begin{split}
 \sum_{i+k\leq 14}\tau^2\int_{\Sigma_\tau\cap\{r\leq r^-_Y\}}J^{N}_\mu\left(\hat{Y}^k\partial_{t^*}^i\Phi\right) n^{\mu}_{\Sigma_\tau}&\leq \frac{A_Y}{4}\epsilon \tau^{\eta_{14}}.
\end{split}
\end{equation*}
\end{proposition}
\begin{proof}
The idea is to use Proposition 13 and use the fact that it gives the control of an integrated in time quantity. From this we can extract a good slice to improve the constant.
By Proposition \ref{commYcontrol},
\begin{equation*}
\begin{split}
&\sum_{i+k=16}\left(\int_{\Sigma_\tau\cap\{r\leq r^+_Y\}} J^{N}_\mu\left(\hat{Y}^{k}\partial_{t^*}^i\Phi\right)n^\mu_{\Sigma_\tau} + \iint_{\mathcal R(\tau',\tau)\cap\{r\leq r^-_Y\}}J_\mu^{N}\left(\hat{Y}^{k}\partial_{t^*}^i\Phi\right)n^\mu_{\Sigma_{t^*}}\right)\\
\leq &C\left(\sum_{j+m\leq 16}\int_{\Sigma_{\tau'}\cap\{r\leq r^+_Y\}} J^N_\mu\left(\partial_{t^*}^j\hat{Y}^m\Phi\right)n^\mu_{\Sigma_{\tau'}}+\sum_{j=0}^{16}\int_{\Sigma_{\tau}\cap\{r\leq r^+_Y\}} J^N_\mu\left(\partial_{t^*}^j\Phi\right)n^\mu_{\Sigma_{\tau}}\right.\\
&\left.+\sum_{j=0}^{16}\iint_{\mathcal R(\tau',\tau)\cap\{ r\leq \frac{23M}{8}\}} J^N_\mu\left(\partial_{t^*}^j\Phi\right)n^\mu_{\Sigma_{t^*}}+\sum_{i+j\leq 16}\iint_{\mathcal R(\tau',\tau)\cap\{r\leq\frac{23M}{8}\}}\left(D^iG_{\leq j,0}\right)^2\right)\\
\leq&CA_Y(\tau')^{\eta_{16}}+CA_0\tau^{\eta_{16}}+CA^2\epsilon^2(\tau')^{-1+\eta_{16}},\\
\end{split}
\end{equation*}
by (\ref{BA1}), (\ref{BAK1}) and Proposition \ref{Nestprop}. Take $\tau'=\tau-A_0$. Then
\begin{equation*}
\begin{split}
&\sum_{i+k=16}\iint_{\mathcal R(\tau-A_0,\tau)\cap\{r\leq r^-_Y\}}J_\mu^{N}\left(\hat{Y}^{k}\partial_{t^*}^i\Phi\right)n^\mu_{\Sigma_{t^*}}\leq CA_Y\tau^{\eta_{16}}+CA_0\tau^{\eta_{16}}+CA^2\epsilon^2\tau^{-1+\eta_{16}}.\\
\end{split}
\end{equation*}
Hence there is some $\tilde{\tau}\in [\tau-A_0,\tau]$ such that 
\begin{equation*}
\begin{split}
\sum_{i+k=16}\int_{\Sigma_{\tilde{\tau}}\cap\{r\leq r^-_Y\}} J^{N}_\mu\left(\hat{Y}^{k}\partial_{t^*}^i\Phi\right)n^\mu_{\Sigma_{\tilde{\tau}}}\leq CA_YA_0^{-1}\tau^{\eta_{16}}+C\tau^{\eta_{16}}+CA^2\epsilon^2\tau^{-1+\eta_{16}}.\\
\end{split}
\end{equation*}
We also have by (\ref{BA1}) and the elliptic estimates in Proposition \ref{elliptic},
\begin{equation*}
\begin{split}
&\sum_{i+k=16}\int_{\Sigma_{\tilde{\tau}}\cap\{r^-_Y\leq r\leq r^+_Y\}} J^{N}_\mu\left(\hat{Y}^{k}\partial_{t^*}^i\Phi\right)n^\mu_{\Sigma_{\tilde{\tau}}} \\
\leq&C \sum_{i=0}^{16}\int_{\Sigma_{\tilde{\tau}}\cap\{r^-_Y\leq r\leq \frac{t^*}{2}\}} J^{N}_\mu\left(\partial_{t^*}^i\Phi\right)n^\mu_{\Sigma_{\tilde{\tau}}}+\sum_{i+j\leq 15}\int_{\Sigma_{\tilde{\tau}}}(D^iG_{\leq j,0})^2\\
\leq &CA_0\tau^{\eta_{16}}+CA^2\epsilon^2\tau^{-1+\eta_{16}}.\\
\end{split}
\end{equation*}
Now reapply Proposition \ref{commYcontrol}, from $\tilde{\tau}$ to $\tau$, we get 
\begin{equation*}
\begin{split}
&\sum_{i+k=16}\int_{\Sigma_\tau\cap\{r\leq r^+_Y\}} J^{N}_\mu\left(\hat{Y}^{k}\partial_{t^*}^i\Phi\right)n^\mu_{\Sigma_\tau}\\
\leq &C\sum_{j+m\leq 16}\int_{\Sigma_{\tilde{\tau}}\cap\{r\leq r^+_Y\}} J^N_\mu\left(\partial_{t^*}^j\hat{Y}^m\Phi\right)n^\mu_{\Sigma_{\tilde{\tau}}}+CA_0\tau^{\eta_{16}}+CA^2\epsilon^2\\
\leq&CA_YA_0^{-1}\tau^{\eta_{16}}+CA_0\tau^{\eta_{16}}+CA^2\epsilon^2\tau^{-1+\eta_{16}}.\\
\end{split}
\end{equation*}
Since $C\ll A_0\ll A_Y$, we get the first statement in the Proposition.
The derivations for the other bounds are identical, with the constants and exponents replaced appropriately.
\end{proof}
From this we can also derive some integrated estimates for $Y^k\Gamma^j\Phi$. This will be useful in controlling the commutator $[\Box_{g_K},S]$.
\begin{proposition}\label{Yintest}
\begin{equation*}
  \sum_{i+k=16}\iint_{\mathcal R((1.1)^{-1}\tau,\tau)\cap\{r\leq r^-_Y\}}J^{N}_\mu\left(\hat{Y}^k\partial_{t^*}^i\Phi\right) n^{\mu}_{\Sigma_\tau}\leq A_Y\epsilon \tau^{\eta_{16}},
\end{equation*}
and
\begin{equation*}
\begin{split}
 &\sum_{i+k=15}\tau^2\iint_{\mathcal R((1.1)^{-1}\tau,\tau)\cap\{r\leq r^-_Y\}}J^{N}_\mu\left(\hat{Y}^k\partial_{t^*}^i\Phi\right) n^{\mu}_{\Sigma_\tau} \leq A_Y\epsilon \tau^{1+\eta_{15}},
\end{split}
\end{equation*}
and
\begin{equation*}
 \begin{split}
 \sum_{i+k\leq 14}\tau^2\iint_{\mathcal R((1.1)^{-1}\tau,\tau)\cap\{r\leq r^-_Y\}}J^{N}_\mu\left(\hat{Y}^k\partial_{t^*}^i\Phi\right) n^{\mu}_{\Sigma_\tau}&\leq A_Y\epsilon \tau^{\eta_{14}}.
\end{split}
\end{equation*}
\end{proposition}
\begin{proof}
This is a direct consequence of Proposition \ref{commYcontrol}, \ref{Nestprop}, \ref{Yest}, as well as the bootstrap assumptions (\ref{BA1}), (\ref{BA2}) and (\ref{BA3}).
\end{proof}

We will finally proceed to the quantities associated to the vector field $S$. Recall from \cite{LKerr} that for large values of $r$
\begin{equation*}
\begin{split}
&|[\Box_{g_K},S]\Phi-\left(2+\frac{r^*\mu}{r}\right)\Box_g \Phi -\frac{2}{r}\left(\frac{r^*}{r}-1-\frac{2r^*\mu}{r}\right)\partial_{r^*}\Phi-2\left(\left(\frac{r^*}{r}-1\right)-\frac{3r^*\mu}{2r}\right)\lapp\Phi|\\
\leq &Ca r^{-2}(\sum_{k=1}^2|D^k\Phi|).
\end{split}
\end{equation*}
and that for finite values of $r$, we have
\begin{equation*}
\begin{split}
&|[\Box_{g_K},S]\Phi|\leq C(\sum_{k=1}^2|D^k\Phi|).
\end{split}
\end{equation*}
Moreover, all the coefficients in the commutator term obey the same estimates (with a different constant) upon differentiation. 
Therefore, 
$$\Box_{g_K}\left(S\Gamma^k\Phi\right)=V_k+S(U_k)+S(N_k),$$
where $$\left(D^\ell V_k\right)^2\leq C r^{-4}\left(\log r\right)^2\left(\sum_{j=1}^{\ell+1}\left(D^j\Gamma^{k+1}\Phi\right)^2+ \sum_{j=1}^{\ell+2} \left(D^j\Gamma^k\Phi\right)^2\right).$$
We would now estimate these three terms separately. We first estimate the $V_k$ terms:
\begin{proposition}\label{V}
For $\alpha\leq 2$,
\begin{equation*}
\begin{split}
\sum_{\ell+k\leq 13}\int_{\Sigma_\tau}r^\alpha \left(D^\ell V_{\leq k}\right)^2\leq CA_{Y}\epsilon\tau^{-2+\eta_{14}+\delta}.
\end{split}
\end{equation*}
For $\alpha\leq 1+\delta$,
\begin{equation*}
\begin{split}
\sum_{\ell+k=13}\iint_{\mathcal R((1.1)^{-1}\tau,\tau)}r^\alpha \left(D^\ell V_{\leq k}\right)^2\leq CA_{Y}\epsilon\tau^{-1+\eta_{15}+\delta}.
\end{split}
\end{equation*}
\begin{equation*}
\begin{split}
\sum_{\ell+k\leq 12}\iint_{\mathcal R((1.1)^{-1}\tau,\tau)}r^\alpha \left(D^\ell V_{\leq k}\right)^2\leq CA_{Y}\epsilon\tau^{-2+\eta_{14}+\delta}.
\end{split}
\end{equation*}
\end{proposition}
\begin{proof}
By the elliptic estimates in Propositions \ref{elliptic} and \ref{elliptichorizon}, we have, for $\alpha\leq 2$,
\begin{equation*}
\begin{split}
&\sum_{\ell+k\leq 13}\int_{\Sigma_\tau}r^\alpha \left(D^\ell V_{\leq k}\right)^2\\
\leq &C \sum_{i+j\leq 12}\int_{\Sigma_\tau}r^{\alpha-4+\delta}J^N_\mu\left(\hat{Y}^i\Gamma^{j}\Phi\right)n^\mu_{\Sigma_\tau}+C\sum_{i+j\leq 11}\int_{\Sigma_\tau}r^{\alpha-4+\delta} \left(D^i G_{\leq j}\right)^2\\
\leq&C\left(A_{Y}\epsilon+A^2\epsilon^2\right)\tau^{-2+\eta_{14}+\delta},
\end{split}
\end{equation*}
where we have used Propositions \ref{Nestprop} and \ref{U}, the bootstrap assumptions (\ref{BA3}) (for $r\leq \frac{9t^*}{10}$) and (\ref{BA4}) (for $r\geq \frac{9t^*}{10}$).

By the elliptic estimates in Propositions \ref{elliptic} and \ref{elliptichorizon}, we have
\begin{equation*}
\begin{split}
&\iint_{\mathcal R((1.1)^{-1}\tau,\tau)}r^\alpha \left(D^\ell V_{\leq k}\right)^2\\
\leq &C\sum_{i+j=0}^{\ell+k+1}\iint_{\mathcal R((1.1)^{-1}\tau,\tau)}r^{\alpha-4+\delta}J^N_\mu\left(\hat{Y}^i\Gamma^{j}\Phi\right)n^\mu_{\Sigma_\tau}+\sum_{i+j=0}^{\ell+k}\iint_{\mathcal R((1.1)^{-1}\tau,\tau)}r^{\alpha-4+\delta} \left(D^i G_{\leq j}\right)^2\\
\end{split}
\end{equation*}
We first consider the case $\ell+k=13$. For the first term, we divide into $r\leq\frac{t^*}{2}$ (which we estimate by (\ref{BAK3}) and Proposition \ref{Yintest}) and $r\geq\frac{t^*}{2}$ (which we estimate using the extra decay in $r$ by (\ref{BAK1.5})). The second term contains the $U_k$ and the $N_k$ part. The $U_k$ part can be estimated by Proposition \ref{U}. The $N_k$ part can be estimated by Proposition \ref{Nestprop}. The $\ell+k\leq 12$ case in completely analogous, replacing the bootstrap assumption (\ref{BAK3}) by (\ref{BAK5}).
\end{proof}

We then proceed to the estimates for $S(U_k)$. Notice that when we prove the estimates for the derivatives for $S(U_k)$, the derivatives for $S(N_k)$ will be involved. Like the proof of the estimates for $U_k$, we will first prove estimates for the derivatives of $S(U_k)$ depending on $S(N_k)$, and close the estimates after we control $S(N_k)$.

\begin{proposition}\label{SUprop}
The following estimates for $S(U_k)$ on a fixed $t^*$ slice hold for $\alpha\leq 2$:
\begin{equation*}
\begin{split}
\sum_{k+\ell=13}\int_{\Sigma_\tau}r^\alpha \left(D^\ell\left(S(U_{k,j})\right)\right)^2\leq CA_{S,j-1}\epsilon\tau^{\eta_{S,13}}+\sum_{m=1}^{13-j}\int_{\Sigma_\tau}(D^mS(N_{\leq j-1}))^2.
\end{split}
\end{equation*}
\begin{equation*}
\begin{split}
\sum_{k+\ell=12}\int_{\Sigma_\tau}r^\alpha \left(D^\ell\left(S(U_{k,j})\right)\right)^2\leq CA_{S,j-1}\epsilon\tau^{-1+\eta_{S,12}}+\sum_{m=1}^{12-j}\int_{\Sigma_\tau}(D^mS(N_{\leq j-1}))^2.
\end{split}
\end{equation*}
\begin{equation*}
\begin{split}
\sum_{k+\ell\leq 11}\int_{\Sigma_\tau}r^\alpha \left(D^\ell\left(S(U_{k,j})\right)\right)^2\leq CA_{S,j-1}\epsilon\tau^{-2+\eta_{S,11}}+\sum_{m=1}^{11-j}\int_{\Sigma_\tau}(D^mS(N_{\leq j-1}))^2.
\end{split}
\end{equation*}
The following estimates for $S(U_k)$ integrated on $[(1.1)^{-1}\tau,\tau]$ also hold for $\alpha\leq 1+\delta$:
\begin{equation*}
\begin{split}
\sum_{k+\ell=13}\iint_{\mathcal R((1.1)^{-1}\tau,\tau)}r^\alpha \left(D^\ell\left(S(U_{k,j})\right)\right)^2\leq CA_{S,X,j-1}\epsilon\tau^{\eta_{S,13}}+\sum_{m=1}^{13-j}\iint_{\mathcal R((1.1)^{-1}\tau,\tau)}r^{-2}(D^mS(N_{\leq j-1}))^2.
\end{split}
\end{equation*}
\begin{equation*}
\begin{split}
\sum_{k+\ell=12}\iint_{\mathcal R((1.1)^{-1}\tau,\tau)}r^\alpha \left(D^\ell\left(S(U_{k,j})\right)\right)^2\leq CA_{S,X,j-1}\epsilon\tau^{-1+\eta_{S,12}}+\sum_{m=1}^{12-j}\iint_{\mathcal R((1.1)^{-1}\tau,\tau)}r^{-2}(D^mS(N_{\leq j-1}))^2.
\end{split}
\end{equation*}
\begin{equation*}
\begin{split}
\sum_{k+\ell\leq 11}\iint_{\mathcal R((1.1)^{-1}\tau,\tau)}r^\alpha \left(D^\ell\left(S(U_{k,j})\right)\right)^2\leq CA_{S,X,j-1}\epsilon\tau^{-2+\eta_{S,11}}+\sum_{m=1}^{11-j}\iint_{\mathcal R((1.1)^{-1}\tau,\tau)}r^{-2}(D^mS(N_{\leq j-1}))^2.
\end{split}
\end{equation*}
\end{proposition}

\begin{proof}
Notice that $D^\ell\left(S(U_{k,j})\right)$ is supported in $\{r\geq R_\Omega\}$ and satisfies 
\begin{equation*}
\begin{split}
|D^\ell\left(S(U_{k,j})\right)|&\leq C\sum_{m=1}^{\ell+2}\sum_{i=0}^{j-1}r^{-2}\left(|D^m S\partial_{t^*}^{k-j}\tilde{\Omega}^i\Phi|+|D^m\partial_{t^*}^{k-j}\tilde{\Omega}^i\Phi|\right)\\
&\leq C\sum_{m=1}^{\ell+k-j+2}\sum_{i=0}^{j-1}r^{-2}\left(|D^m S\tilde{\Omega}^i\Phi|+|D^m\tilde{\Omega}^i\Phi|\right).
\end{split}
\end{equation*}
We can ignore the last term because it appears already in $D^\ell U_{k,j}$ and can be estimated by Proposition \ref{U}.
\begin{equation*}
\begin{split}
&\int_{\Sigma_\tau}r^\alpha \left(D^\ell\left(S(U_{k,j})\right)\right)^2\\
\leq&C\sum_{m=1}^{\ell+k-j+2}\sum_{i=0}^{j-1}\int_{\Sigma_\tau\cap\{r\geq R_\Omega\}}r^{\alpha-4}\left(D^m S\tilde{\Omega}^i\Phi\right)^2\\
\leq&C\sum_{m=1}^{\ell+k-j+1}\sum_{i=0}^{j-1}\int_{\Sigma_\tau\cap\{r\geq R_\Omega-1\}}r^{\alpha-4}J^N_\mu\left(\partial_{t^*}^m S\tilde{\Omega}^i\Phi\right)n^\mu_{\Sigma_\tau}+C\sum_{m=1}^{\ell+k-j}\sum_{i=0}^{j-1}\int_{\Sigma_\tau}r^{-2}\left(D^m\Box_{g_K}(S\tilde{\Omega}^i\Phi)\right)^2\\
\leq&C\sum_{m=1}^{\ell+k-j+1}\sum_{i=0}^{j-1}\int_{\Sigma_\tau\cap\{r\geq R_\Omega-1\}}r^{\alpha-4}J^N_\mu\left(\partial_{t^*}^m S\tilde{\Omega}^i\Phi\right)n^\mu_{\Sigma_\tau}\\
&+C\sum_{m=1}^{\ell+k-j}\int_{\Sigma_\tau}r^{-2}\left((D^mS(U_{\leq j-1,\leq j-1}))^2+(D^mS(N_{\leq j-1}))^2+(D^mV_{\leq j-1})^2\right).
\end{split}
\end{equation*}
We now apply the bootstrap assumptions. By bootstrap assumption (\ref{BA5}) and Proposition \ref{V},
$$\sum_{k+\ell=13}\int_{\Sigma_\tau}r^\alpha \left(D^\ell\left(S(U_{k,j})\right)\right)^2\leq CA_{S,j-1}\epsilon\tau^{\eta_{S,13}}+\sum_{m=1}^{13-j}\int_{\Sigma_\tau}\left((D^mS(U_{\leq j-1,\leq j-1}))^2+(D^mS(N_{\leq j-1}))^2\right).$$
By bootstrap assumption (\ref{BA4}) (for $r\geq\frac{t^*}{2}$), (\ref{BA6}) (for $r\leq\frac{t^*}{2}$) and Proposition \ref{V},
$$\sum_{k+\ell=12}\int_{\Sigma_\tau}r^\alpha \left(D^\ell\left(S(U_{k,j})\right)\right)^2\leq CA_{S,j-1}\epsilon\tau^{-1+\eta_{S,12}}+\sum_{m=1}^{12-j}\int_{\Sigma_\tau}\left((D^mS(U_{\leq j-1,\leq j-1}))^2+(D^mS(N_{\leq j-1}))^2\right).$$
By bootstrap assumption (\ref{BA4}) (for $r\geq\frac{t^*}{2}$), (\ref{BA7}) (for $r\leq\frac{t^*}{2}$) and Proposition \ref{V},
$$\sum_{k+\ell\leq 11}\int_{\Sigma_\tau}r^\alpha \left(D^\ell\left(S(U_{k,j})\right)\right)^2\leq CA_{S,j-1}\epsilon\tau^{-2+\eta_{S,11}}+\sum_{m=1}^{11-j}\int_{\Sigma_\tau}\left((D^mS(U_{\leq j-1,\leq j-1}))^2+(D^mS(N_{\leq j-1}))^2\right).$$
Noticing that $U_{k,0}=0$, we can conclude the first three statements in the Proposition using an induction in $j$. (See the proof of Proposition \ref{Uestprop}).
For the integrated in time estimate, we note that $r^{-1-\delta}J^N_\mu\left(\Gamma^i\Phi\right)\leq CK^{X_0}\left(\Gamma^i\Phi\right)$ and use the bootstrap assumptions (\ref{BAK6}), (\ref{BAK6.5}), (\ref{BAK7}) and (\ref{BAK9}) (See proof of Propositions \ref{Uestprop} and \ref{U}).
\end{proof}

We then move on to the $S(N_k)$ terms, first we will prove an estimate for the derivatives of $S(N_k)$. The decay rate here is not optimal, but would be sufficient to close the bootstrap argument. Our approach here is to prove the decay rate that is driven only by the pointwise decay of $D^\ell\Phi$ but not by that of $D^\ell S\Phi$. The latter can, in principle, be done by similar methods, but we will skip it since it will not be necessary. In subsequent propositions, we will then prove refined decay rate for $S(N_k)$ (without derivatives) as well as for $D^\ell S(N_k)$ restricted to the region $r\leq\frac{t^*}{4}$.

\begin{proposition}\label{SNestprop}
$S(N_{k})$ satisfies the following estimates for any fixed $t^*=\tau$:
\begin{equation*}
\sum_{k+\ell =13}\int_{\Sigma_\tau} (D^\ell S(N_{k}))^2 \leq CB_SA^2\epsilon^2 \tau^{-2+\eta_{S,11}}
\end{equation*}
\begin{equation*}
\sum_{k+\ell \leq 12}\int_{\Sigma_\tau} (D^\ell S(N_{k}))^2 \leq CA^2\epsilon^2 \tau^{-2+\eta_{14}+\delta}
\end{equation*}
$S(N_{k})$ also satisfies the following integrated estimates over $t^*\in [(1.1)^{-1}\tau,\tau]$:
\begin{equation*}
\sum_{k+\ell =12}\iint_{\mathcal R((1.1)^{-1}\tau,\tau)} r^{-1-\delta}(D^\ell S(N_{k}))^2 \leq CB_SA^2\epsilon^2 \tau^{-2+\eta_{S,11}}
\end{equation*}
\begin{equation*}
\sum_{k+\ell \leq 11}\iint_{\mathcal R((1.1)^{-1}\tau,\tau)} r^{-1-\delta}(D^\ell S(N_{k}))^2 \leq CA^2\epsilon^2 \tau^{-2+\eta_{14}+\delta}
\end{equation*}
\end{proposition}
\begin{proof}
We would like to do a reduction similar to how we estimated $N_k$. Clearly, only the quadratic and cubic terms matter and we only need to consider terms that contain $S$ for the other terms are already controlled by the estimates of $N_k$. We will denote terms that are already in $N_k$ by ``good terms''. The only cubic terms that are relevant are those which contain $S\Gamma^{i}\Phi$ since in the terms with $D^{j+1}S\Gamma^i\Phi$, we can put all but one other factors in $L^\infty$ using the bootstrap assumptions (\ref{BAP1}), (\ref{BAP2}), (\ref{BAPI1}) and (\ref{BAPI1.1}). Notice also that the conditions for $D_\Phi\Lambda_0$, $D_\Phi\Lambda_1$ and $D_\Phi\mathcal C$ in the definition of the null condition guarantee that the bounds do not deteriorate if $S$ acts on the coefficients. The relevant terms are
$$(D^{j_1}S\Gamma^{i_1}\Phi)(D^{j_2}\Gamma^{i_2}\Phi),\quad\mbox{and}$$
$$(D^{j_1}\Gamma^{i_1}\Phi)(D^{j_2}\Gamma^{i_2}\Phi)(S\Gamma^{i_3}\Phi).$$

We first treat the case that $k+\ell\leq 12$. In this case we will always put factors without $S$ in $L^\infty$.
\begin{equation*}
\begin{split}
&\int_{\Sigma_\tau} (D^\ell S(N_{k}))^2 \\
\leq&C\sum_{i_1+i_2+j_1+j_2\leq k+\ell+2, j_1,j_2\geq 1}\int_{\Sigma_\tau} (D^{j_1}S\Gamma^{i_1}\Phi)^2(D^{j_2}\Gamma^{i_2}\Phi)^2 \\
&+C\sum_{i_1+i_2+i_3+j_1+j_2\leq 14,j_1,j_2\geq 1}\int_{\Sigma_\tau} (D^{j_1}\Gamma^{i_1}\Phi)^2(D^{j_2}\Gamma^{i_2}\Phi)^2(S\Gamma^{i_3}\Phi)^2+\mbox{good terms} \\
\leq&C\left(\sum_{i+j\leq k+\ell+1, j\geq 1}\sup \left(D^{j}\Gamma^{i}\Phi\right)^2\right)\sum_{i+j\leq k+\ell+1, j\geq 1}\int_{\Sigma_\tau} (D^{j}S\Gamma^{i}\Phi)^2\\
&+C\left(\sum_{i+j\leq k+\ell+1, j\geq 1}\sup \left(D^{j}\Gamma^{i}\Phi\right)^2\right)\left(\sum_{i+j\leq k+\ell+1, j\geq 1}\sup r^2\left(D^{j}\Gamma^{i}\Phi\right)^2\right)\sum_{i\leq k+\ell}\int_{\Sigma_\tau} r^{-2}(S\Gamma^{i}\Phi)^2 \\
&+\mbox{good terms}\\
\leq&CA\epsilon\tau^{-2+\eta_{14}}\sum_{i+j\leq k+\ell+1, j\geq 1}\int_{\Sigma_\tau} (D^{j}S\Gamma^{i}\Phi)^2 +CA^2\epsilon^2\tau^{-2+\eta_{14}}\sum_{i\leq k+\ell}\int_{\Sigma_\tau} J^N_\mu(S\Gamma^{i}\Phi)n^\mu_{\Sigma_\tau}+\mbox{good terms}\\
&\quad\mbox{using the bootstrap assumption (\ref{BAP1.5}), (\ref{BAP2}), (\ref{BAP6}), (\ref{BAPI1.2}) and (\ref{BAPI4}) and Proposition \ref{Hardy}}\\
\leq&CA\epsilon\tau^{-2+\eta_{14}}\sum_{i\leq k+\ell}\int_{\Sigma_{\tau}} \left(J^N_\mu\left(S\Gamma^i\Phi\right)n^\mu_{\Sigma_\tau}+J^N_\mu\left(\Gamma^i\Phi\right)n^\mu_{\Sigma_\tau}\right) \\
&+CA\epsilon\tau^{-2+\eta_{14}}\sum_{i+j\leq k+\ell-1}\int_{\Sigma_{\tau}} \left((D^i U_{\leq j})^2+(D^i S(U_{\leq j}))^2+\left(D^i N_{\leq j}\right)^2+\left(D^i S(N_{\leq j})\right)^2+\left(D^i V_{\leq j}\right)^2\right).
\end{split}
\end{equation*}
We now apply the estimates for the inhomogeneous terms, i.e., Propositions \ref{Nestprop}, \ref{U}, \ref{V} \ref{SUprop}. Since $k+\ell\leq 12$,: 
\begin{equation*}
\begin{split}
&\int_{\Sigma_\tau} (D^\ell S(N_{k}))^2 \\
\leq&CA\epsilon\tau^{-2+\eta_{14}}\sum_{i\leq 12}\int_{\Sigma_{\tau}} \left(J^N_\mu\left(S\Gamma^i\Phi\right)n^\mu_{\Sigma_\tau}+J^N_\mu\left(\Gamma^i\Phi\right)n^\mu_{\Sigma_\tau}\right) +CA^2\epsilon^2\tau^{-2+\eta_{14}+\delta} \\
&+CA\epsilon\tau^{-2+\eta_{14}}\sum_{i+j\leq k+\ell-1}\int_{\Sigma_{\tau}} \left(D^i S(N_{\leq j})\right)^2.
\end{split}
\end{equation*}
The desired estimates then follow from an induction, together with the bootstrap assumptions (\ref{BA4}) and (\ref{BA8}), since according to this notation $\displaystyle\sum_{i+j=0}^{-1}=0$.

We then treat the case that $k+\ell=13$. In this case it is possible to have $14$ derivatives falling on the factor with $\Phi$ and hence cannot be controlled in $L^\infty$. However, in this scenario, we must have
$$\sum_{i+j=14}(DS\Phi)(D^{j}\Gamma^{i}\Phi)$$
and therefore $DS\Phi$ can be controlled in $L^\infty$ by the bootstrap assumptions (\ref{BAP6}) and (\ref{BAPI4}). In short, we have
\begin{equation*}
\begin{split}
&\sum_{k+\ell =13}\int_{\Sigma_\tau} (D^\ell S(N_{k}))^2 \\
\leq&C\sum_{i_1+i_2+j_1+j_2\leq 15, j_1,j_2\geq 1}\int_{\Sigma_\tau} (D^{j_1}S\Gamma^{i_1}\Phi)^2(D^{j_2}\Gamma^{i_2}\Phi)^2 \\
&+C\sum_{i_1+i_2+i_3+j_1+j_2\leq 15,j_1,j_2\geq 1}\int_{\Sigma_\tau} (D^{j_1}\Gamma^{i_1}\Phi)^2(D^{j_2}\Gamma^{i_2}\Phi)^2(S\Gamma^{i_3}\Phi)^2 +\mbox{good terms}\\
\leq&C\left(\sum_{i+j\leq 14, j\geq 1}\sup \left(D^{j}\Gamma^{i}\Phi\right)^2\right)\sum_{i+j\leq 14, j\geq 1}\int_{\Sigma_\tau} (D^{j}S\Gamma^{i}\Phi)^2+ \left(\sup\left(DS\Phi\right)^2\right)\sum_{i+j=14, j\geq 1}\int_{\Sigma_\tau}(D^{j}\Gamma^{i}\Phi)^2\\
&+C\left(\sum_{i+j\leq 14, j\geq 1}\sup \left(D^{j}\Gamma^{i}\Phi\right)^2\right)\left(\sum_{i+j\leq 14, j\geq 1}\sup r^2\left(D^{j}\Gamma^{i}\Phi\right)^2\right)\sum_{i\leq 13}\int_{\Sigma_\tau} r^{-2}(S\Gamma^{i}\Phi)^2+\mbox{good terms}\\
\leq&CA\epsilon\tau^{-2+\eta_{14}}\sum_{i+j\leq 14, j\geq 1}\int_{\Sigma_\tau} (D^{j}S\Gamma^{i}\Phi)^2+ CB_SA\epsilon \tau^{-2+\eta_{S,11}}\sum_{i+j=14,j\geq 1}\int_{\Sigma_\tau}(D^{j}\Gamma^{i}\Phi)^2\\
&+CA^2\epsilon^2\tau^{-2+\eta_{14}}\sum_{i\leq 13}\int_{\Sigma_\tau} J^N_\mu(S\Gamma^{i}\Phi)n^\mu_{\Sigma_\tau}+\mbox{good terms}\\
&\quad\mbox{using the bootstrap assumption (\ref{BAP1.5}), (\ref{BAP2}), (\ref{BAP6}), (\ref{BAPI1.2}) and (\ref{BAPI4}) and Proposition \ref{Hardy}}\\
\leq&CA\epsilon\tau^{-2+\eta_{14}}\sum_{i\leq 13}\int_{\Sigma_{\tau}} J^N_\mu\left(S\Gamma^i\Phi\right)n^\mu_{\Sigma_\tau}+CB_SA\epsilon\tau^{-2+\eta_{S,11}}\sum_{i\leq 13}\int_{\Sigma_{\tau}}J^N_\mu\left(\Gamma^i\Phi\right)n^\mu_{\Sigma_\tau} \\
&+CA\epsilon\tau^{-2+\eta_{14}}\sum_{\ell+k\leq 12}\int_{\Sigma_{\tau}} \left((D^\ell U_{\leq k})^2+(D^iS(U_{\leq k}))^2+\left(D^\ell N_{\leq k}\right)^2+\left(D^\ell S(N_{\leq k})\right)^2+\left(D^\ell V_{\leq k}\right)^2\right).
\end{split}
\end{equation*}
We now apply the estimates for the inhomogeneous terms, i.e., Propositions \ref{Nestprop}, \ref{U}, \ref{V}, \ref{SUprop}: 
\begin{equation*}
\begin{split}
&\sum_{k+\ell =13}\int_{\Sigma_\tau} (D^\ell S(N_{k}))^2 \\
\leq&CA\epsilon\tau^{-2+\eta_{14}}\sum_{i\leq 13}\int_{\Sigma_{\tau}} J^N_\mu\left(S\Gamma^i\Phi\right)n^\mu_{\Sigma_\tau}+CB_SA\epsilon\tau^{-2+\eta_{S,11}}\sum_{i\leq 13}\int_{\Sigma_{\tau}}J^N_\mu\left(\Gamma^i\Phi\right)n^\mu_{\Sigma_\tau} \\
&+CA\epsilon\tau^{-2+\eta_{14}}\sum_{\ell+k\leq 12}\int_{\Sigma_{\tau}} \left(D^\ell S(N_{\leq k})\right)^2,
\end{split}
\end{equation*}
which is acceptable.
The estimates for the integrated in $t^*$ terms are proved analogously, noting that the elliptic estimate in Proposition \ref{elliptic} would allow for weight in $r$ and use the second parts of Propositions \ref{Nestprop}, \ref{U}, \ref{V} and \ref{SUprop}.
\end{proof}
This would allow us to close the estimates for $S(U_k)$ from Proposition \ref{SUprop}.

\begin{proposition}\label{SU}
The following estimates for $S(U_k)$ on a fixed $t^*$ slice hold for $\alpha\leq 2$:
\begin{equation*}
\begin{split}
\sum_{k+\ell=13}\int_{\Sigma_\tau}r^\alpha \left(D^\ell\left(S(U_{k,j})\right)\right)^2\leq CA_{S,j-1}\epsilon\tau^{\eta_{S,13}}.
\end{split}
\end{equation*}
\begin{equation*}
\begin{split}
\sum_{k+\ell=12}\int_{\Sigma_\tau}r^\alpha \left(D^\ell\left(S(U_{k,j})\right)\right)^2\leq CA_{S,j-1}\epsilon\tau^{-1+\eta_{S,12}}.
\end{split}
\end{equation*}
\begin{equation*}
\begin{split}
\sum_{k+\ell\leq 11}\int_{\Sigma_\tau}r^\alpha \left(D^\ell\left(S(U_{k,j})\right)\right)^2\leq CA_{S,j-1}\epsilon\tau^{-2+\eta_{S,11}}.
\end{split}
\end{equation*}
The following estimates for $S(U_k)$ integrated on $[(1.1)^{-1}\tau,\tau]$ also hold for $\alpha\leq 1+\delta$:
\begin{equation*}
\begin{split}
\sum_{k+\ell=13}\iint_{\mathcal R((1.1)^{-1}\tau,\tau)}r^\alpha \left(D^\ell\left(S(U_{k,j})\right)\right)^2\leq CA_{S,X,j-1}\epsilon\tau^{\eta_{S,13}}.
\end{split}
\end{equation*}
\begin{equation*}
\begin{split}
\sum_{k+\ell=12}\iint_{\mathcal R((1.1)^{-1}\tau,\tau)}r^\alpha \left(D^\ell\left(S(U_{k,j})\right)\right)^2\leq CA_{S,X,j-1}\epsilon\tau^{-1+\eta_{S,12}}.
\end{split}
\end{equation*}
\begin{equation*}
\begin{split}
\sum_{k+\ell\leq 11}\iint_{\mathcal R((1.1)^{-1}\tau,\tau)}r^\alpha \left(D^\ell\left(S(U_{k,j})\right)\right)^2\leq CA_{S,X,j-1}\epsilon\tau^{-2+\eta_{S,11}}.
\end{split}
\end{equation*}
\end{proposition}
\begin{proof}
This follows directly from Proposition \ref{SUprop} and \ref{SNestprop}.
\end{proof}
In the region $\{r\leq\frac{t^*}{4}\}$, we have refined decay rates for $D^\ell S(N_k)$:

\begin{proposition}\label{SNII}
\begin{equation*}
\sum_{k+\ell=13}\int_{\Sigma_\tau\cap\{r\leq\frac{t^*}{4}\}}r^{1-\delta} (D^\ell S(N_{k}))^2 \leq CA^2\epsilon^2 \tau^{-3+\eta_{S,11}}.
\end{equation*}
\begin{equation*}
\sum_{k+\ell\leq 12}\int_{\Sigma_\tau\cap\{r\leq\frac{t^*}{4}\}}r^{1-\delta} (D^\ell S(N_k))^2 \leq CA^2\epsilon^2 \tau^{-4+\eta_{S,12}+\eta_{S,11}}.
\end{equation*}
\end{proposition}
\begin{proof}
Take $k+\ell\leq 13$. Notice that $|[D,S]\Phi|\leq C|D\Phi|$.

We would like to do a reduction similar to how we estimated $N_k$. Clearly, only the quadratic and cubic terms matter and we only need to consider terms that contain $S$ for the other terms are already controlled by the estimates of $N_k$. Notice also as before that the conditions in the null condition guarantee that the bounds do not deteriorate if $S$ acts on the coefficients. The relevant terms are
$$(D^{j_1}S\Gamma^{i_1}\Phi)(D^{j_2}\Gamma^{i_2}\Phi)\quad j_1,j_2\geq 1,$$
$$(D^{j_1}S\Gamma^{i_1}\Phi)(D^{j_2}\Gamma^{i_2}\Phi)(\Gamma^{i_3}\Phi), \quad j_1,j_2\geq 1, i_3> 8\quad\mbox{and}$$
$$(D^{j_1}\Gamma^{i_1}\Phi)(D^{j_2}\Gamma^{i_2}\Phi)(S\Gamma^{i_3}\Phi)\quad j_1,j_2\geq 1, i_3 >8.$$
We first tackle the quadratic terms:
\begin{equation*}
\begin{split}
&\sum_{i_1+\ell_1\leq 7, \ell_1\geq 1}\sum_{i_2+j_2\leq k+\ell+1, j_2\geq 1}\int_{\Sigma_\tau\cap\{r\leq\frac{\tau}{4}\}}r^{1-\delta} \left(|{D}^{j_1}S\Gamma^{i_1}\Phi D^{j_2}\Gamma^{i_2}\Phi|^2+|{D}^{j_1}\Gamma^{i_1}\Phi D^{j_2}S\Gamma^{i_2}\Phi|^2\right)\\
\leq &C\left(\sup_{r\leq\frac{\tau}{4}}\sum_{i+j \leq 7, j\geq 1}r^{1-\delta}|{D}^{j}S\Gamma^{i}\Phi|^2\right)\left(\sum_{i+j\leq k+\ell+1, j\geq 1}\int_{\Sigma_\tau\cap\{r\leq\frac{\tau}{4}\}}| D^{j}\Gamma^{i}\Phi|^2\right)\\
&+C\left(\sup_{r\leq\frac{\tau}{4}}\sum_{i+j \leq 7, j\geq 1}r^{1-\delta}|{D}^j\Gamma^{i}\Phi|^2\right)\left(\sum_{i+j\leq k+\ell+1, j\geq 1}\int_{\Sigma_\tau\cap\{r\leq\frac{\tau}{4}\}}| D^j S\Gamma^{i}\Phi|^2\right)\\
\leq&CA\epsilon\tau^{-2+\eta_{S,11}}\sum_{i+j\leq k+\ell}\int_{\Sigma_\tau\cap\{r\leq\frac{9\tau}{10}\}}J^N_\mu(\hat{Y}^j\Gamma^{i}\Phi)+CA\epsilon\tau^{-3+\eta_{S,11}}\sum_{i+j\leq k+\ell}\int_{\Sigma_\tau\cap\{r\leq\frac{9\tau}{10}\}}J^N_\mu(\hat{Y}^jS\Gamma^{i}\Phi) \\
&+CA\epsilon\tau^{-2+\eta_{S,11}}\sum_{i+j\leq k+\ell-1}\int_{\Sigma_\tau\cap\{r\leq\frac{9\tau}{10}\}} \left((D^i U_j)^2+(D^i N_j)^2\right) \\
&+CA\epsilon\tau^{-3+\eta_{S,11}}\sum_{i+j\leq k+\ell-1}\int_{\Sigma_\tau\cap\{r\leq\frac{9\tau}{10}\}} \left((D^iS( U_j))^2+(D^i S( N_j))^2+(D^iV_j)^2\right)
\end{split}
\end{equation*}
by the bootstrap assumptions (\ref{BAP2}) and (\ref{BAP4}) and the elliptic estimates Proposition \ref{elliptic} and \ref{elliptichorizon}. Since $k+\ell-1\leq 12$, the inhomogeneous terms can be bounded using Proposition \ref{Nestprop}, \ref{U}, \ref{V}, \ref{SNestprop} and \ref{SU} to be
$$\leq CA^2\epsilon^2\tau^{-4+\eta_{S,12}+\eta_{S,11}}.$$
We then move on to the cubic terms:
\begin{equation*}
\begin{split}
&\sum_{i_1+j_1\leq 7, j_1\geq 1}\sum_{i_2+j_2\leq 7, j_2\geq 1}\sum_{i_3=0}^{k}\int_{\Sigma_\tau\cap\{r\leq\frac{9\tau}{10}\}}r^{1-\delta} \left((D^{j_1}S\Gamma^{i_1}\Phi D^{j_2}\Gamma^{i_2}\Phi \Gamma^{i_3}\Phi)^2+(D^{j_1}\Gamma^{i_1}\Phi D^{j_2}\Gamma^{i_2}\Phi S\Gamma^{i_3}\Phi)^2\right)\\
\leq&C\left(\sup_{r\leq\frac{\tau}{4}}\sum_{i+j\leq 7, j\geq 1}r^2\left(D^jS\Gamma^{i}\Phi\right)^2\right)\left(\sup_{r\leq\frac{\tau}{4}}\sum_{i+j\leq 7, j\geq 1} r^{1-\delta}\left(D^j\Gamma^{i}\Phi\right)^2\right)\sum_{i=0}^{k}\int_{\Sigma_\tau\cap\{r\leq\frac{\tau}{4}\}}r^{-2} (\Gamma^{i_3}\Phi)^2\\
&+C\left(\sup_{r\leq\frac{\tau}{4}}\sum_{i+j\leq 7, j\geq 1}r^{1-\delta}(D^j\Gamma^{i}\Phi)^2\right)^2\tau^{1+\delta}\sum_{i=0}^{k}\int_{\Sigma_\tau\cap\{r\leq\frac{\tau}{4}\}}r^{-2}( S\Gamma^{i}\Phi)^2\\
\leq&CA^2\epsilon^2\tau^{-5+2\eta_{S,11}}\sum_{i=0}^{k}\int_{\Sigma_\tau} (D\Gamma^{i}\Phi)^2+CA^2\epsilon^2\tau^{-5+2\eta_{S,11}+\delta}\sum_{i=0}^{k}\int_{\Sigma_\tau}( DS\Gamma^{i}\Phi)^2,
\end{split}
\end{equation*}
by the bootstrap assumptions (\ref{BAP2}) and (\ref{BAP4}), which now clearly decays better than we need by using the bootstrap assumptions (\ref{BA4}), (\ref{BA5}) and (\ref{BA8}).
Therefore,
\begin{equation*}
\begin{split}
&\int_{\Sigma_\tau\cap\{r\leq\frac{\tau}{4}\}}r^{1-\delta} (D^\ell S(N_k))^2\\
\leq&CA\epsilon\tau^{-2+\eta_{S,11}}\sum_{i+j\leq k+\ell}\int_{\Sigma_\tau\cap\{r\leq\frac{9\tau}{10}\}}J^N_\mu(\hat{Y}^j\Gamma^{i}\Phi)+CA\epsilon\tau^{-3+\eta_{S,11}}\sum_{i+j\leq k+\ell}\int_{\Sigma_\tau\cap\{r\leq\frac{9\tau}{10}\}}J^N_\mu(\hat{Y}^jS\Gamma^{i}\Phi) \\
&+CA^2\epsilon^2 \tau^{-4+\eta_{S,12}+\eta_{S,11}}.
\end{split}
\end{equation*}
The Proposition follows from the Bootstrap Assumptions (\ref{BA3}), (\ref{BA5}), (\ref{BA6}) and (\ref{BA7}).
\end{proof}
A similar decay rate can be proved in the region $\{r\leq\frac{9t^*}{10}\}$, if we do not require the estimate for the derivatives:

\begin{proposition}\label{SNI}
\begin{equation*}
\int_{\Sigma_\tau\cap\{r\leq\frac{9t^*}{10}\}}r^{1-\delta} (S(N_{13}))^2 \leq CA^2\epsilon^2 \tau^{-3+\eta_{S,11}}
\end{equation*}
\begin{equation*}
\sum_{k=0}^{12}\int_{\Sigma_\tau\cap\{r\leq\frac{9t^*}{10}\}}r^{1-\delta} (S(N_k))^2 \leq CA^2\epsilon^2 \tau^{-4+\eta_{S,12}+\eta_{S,11}}
\end{equation*}
\end{proposition}
\begin{proof}
Take $k\leq 13$. The proof follows very closely from that of the previous Proposition, by noting that we have similar pointwise decay estimates in the region (without higher derivatives) by the bootstrap assumptions (\ref{BAP3}) and (\ref{BAP7}). As in the previous Proposition, the relevant terms are
$$(DS\Gamma^{i_1}\Phi)(D\Gamma^{i_2}\Phi),$$
$$(DS\Gamma^{i_1}\Phi)(D\Gamma^{i_2}\Phi)(\Gamma^{i_3}\Phi), \quad i_3> 8\quad\mbox{and}$$
$$(D\Gamma^{i_1}\Phi)(D\Gamma^{i_2}\Phi)(S\Gamma^{i_3}\Phi)\quad i_3 >8.$$
We first tackle the quadratic terms:
\begin{equation*}
\begin{split}
&\sum_{i_1=0}^{6}\sum_{i_2=0}^{k}\int_{\Sigma_\tau\cap\{r\leq\frac{9\tau}{10}\}}r^{1-\delta} \left(|{D}S\Gamma^{i_1}\Phi D\Gamma^{i_2}\Phi|^2+|{D}\Gamma^{i_1}\Phi DS\Gamma^{i_2}\Phi|^2\right)\\
\leq &C\left(\sup_{r\leq\frac{9\tau}{10}}\sum_{i=0}^{6}r^{1-\delta}|{D}S\Gamma^{i}\Phi|^2\right)\left(\sum_{i=0}^{k}\int_{\Sigma_\tau\cap\{r\leq\frac{9\tau}{10}\}}| D\Gamma^{i}\Phi|^2\right)\\
&+C\left(\sup_{r\leq\frac{9\tau}{10}}\sum_{i=0}^{6}r^{1-\delta}|{D}\Gamma^{i}\Phi|^2\right)\left(\sum_{i=0}^{k}\int_{\Sigma_\tau\cap\{r\leq\frac{9\tau}{10}\}}| DS\Gamma^{i}\Phi|^2\right)\\
\leq&CA^2\epsilon^2\tau^{-4+\eta_{14}+\eta_{S,11}}+CA\epsilon\tau^{-3+\eta_{S,11}}\left(\sum_{i=0}^{k}\int_{\Sigma_\tau\cap\{r\leq\frac{9\tau}{10}\}}| DS\Gamma^{i}\Phi|^2\right).
\end{split}
\end{equation*}
We then move on to the cubic terms:
\begin{equation*}
\begin{split}
&\sum_{i_1, i_2=0}^{6}\sum_{i_3=0}^{k}\int_{\Sigma_\tau\cap\{r\leq\frac{9\tau}{10}\}}r^{1-\delta} \left((DS\Gamma^{i_1}\Phi D\Gamma^{i_2}\Phi \Gamma^{i_3}\Phi)^2+(D\Gamma^{i_1}\Phi D\Gamma^{i_2}\Phi S\Gamma^{i_3}\Phi)^2\right)\\
\leq&C\left(\sup_{r\leq\frac{9\tau}{10}}\sum_{i=0}^{6}r^2\left(DS\Gamma^{i}\Phi\right)^2\right)\left(\sup_{r\leq\frac{9\tau}{10}}\sum_{i=0}^6 r^{1-\delta}\left(D\Gamma^{i}\Phi\right)^2\right)\sum_{i=0}^{k}\int_{\Sigma_\tau\cap\{r\leq\frac{9\tau}{10}\}}r^{-2} (\Gamma^{i}\Phi)^2\\
&+C\left(\sup_{r\leq\frac{9\tau}{10}}\sum_{i=0}^{6}r^{1-\delta}(D\Gamma^{i}\Phi)^2\right)^2\tau^{1+\delta}\sum_{i=0}^{k}\int_{\Sigma_\tau\cap\{r\leq\frac{9\tau}{10}\}}r^{-2}( S\Gamma^{i}\Phi)^2\\
\leq&CA^2\epsilon^2\tau^{-5+2\eta_{S,,11}}\sum_{i=0}^{k}\int_{\Sigma_\tau} (D\Gamma^{i}\Phi)^2+CA^2\epsilon^2\tau^{-5+2\eta_{S,11}+\delta}\sum_{i=0}^{k}\int_{\Sigma_\tau}( DS\Gamma^{i}\Phi)^2\\
\leq&CA^3\epsilon^3\tau^{-5+2\eta_{S,11}}+CA^2\epsilon^2\tau^{-5+2\eta_{S,11}+\delta}\sum_{i=0}^{k}\int_{\Sigma_\tau}( DS\Gamma^{i}\Phi)^2.\\
\end{split}
\end{equation*}
Therefore,
\begin{equation*}
\begin{split}
&\int_{\Sigma_\tau\cap\{r\leq\frac{\tau}{4}\}}r^{1-\delta} (S(N_k))^2\\
\leq&CA^2\epsilon^2\tau^{-4+\eta_{14}+\eta_{S,11}}+CA\epsilon\tau^{-3+\eta_{S,11}}\sum_{i=0}^{k}\int_{\Sigma_\tau\cap\{r\leq\frac{9\tau}{10}\}}\left( DS\Gamma^{i}\Phi\right)^2+CA^2\epsilon^2\tau^{-5+2\eta_{S,11}+\delta}\sum_{i=0}^{k}\int_{\Sigma_\tau}( DS\Gamma^{i}\Phi)^2.
\end{split}
\end{equation*}
The Proposition follows from the Bootstrap Assumptions (\ref{BA5}), (\ref{BA6}) and (\ref{BA7}).
\end{proof}
We then move on to the region $\{r\geq\frac{9t^*}{10}\}$.
\begin{proposition}\label{SNO}
For $\alpha=0$ or $2$,
\begin{equation*}
\int_{\Sigma_\tau\cap\{r\geq\frac{9\tau}{10}\}} (S(N_{13}))^2 \leq CA^2\epsilon^2 \tau^{-2+\eta_{S,13}}
\end{equation*}
\begin{equation*}
\int_{\Sigma_\tau\cap\{r\geq\frac{9\tau}{10}\}}r^{\alpha} (S(N_{12}))^2 \leq CA^2\epsilon^2 \tau^{-3+\alpha+\eta_{S,12}}
\end{equation*}
\begin{equation*}
\sum_{k=0}^{11}\int_{\Sigma_\tau\cap\{r\geq\frac{9\tau}{10}\}}r^{1-\delta} (S(N_k))^2 \leq CA^2\epsilon^2 \tau^{-4+\alpha+\eta_{S,11}}.
\end{equation*}
\end{proposition}
\begin{proof}
Take $k\leq 13$. Following the reduction before and noticing that $[D,S]\sim D$ and $[\bar{D},S]\sim \bar{D}$, we have to consider quadratic terms
$$(\bar{D}S\Gamma^{i_1}\Phi D\Gamma^{i_2}\Phi), (\bar{D}\Gamma^{i_1}\Phi DS\Gamma^{i_2}\Phi), ({D}S\Gamma^{i_1}\Phi \bar{D}\Gamma^{i_2}\Phi),$$
$$({D}\Gamma^{i_1}\Phi \bar{D}S\Gamma^{i_2}\Phi), r^{-1}({D}S\Gamma^{i_1}\Phi {D}\Gamma^{i_2}\Phi), r^{-1}({D}\Gamma^{i_1}\Phi {D}S\Gamma^{i_2}\Phi),$$
for $i_1\geq i_2$ and the cubic terms
$$ (\bar{D}\Gamma^{i_1}\Phi {D}\Gamma^{i_2}\Phi S\Gamma^{i_3}\Phi), ({D}\Gamma^{i_1}\Phi \bar{D}\Gamma^{i_2}\Phi S\Gamma^{i_3}\Phi), r^{-1}({D}\Gamma^{i_1}\Phi {D}\Gamma^{i_2}\Phi S\Gamma^{i_3}\Phi).$$
For these cubic terms, we can assume $i_1, i_2\leq 6$ for otherwise $i_3\leq 6$ and we can control the last factor in the sup norm and reduce to the quadratic terms above. The cubic terms
$$(\bar{D}S\Gamma^{i_1}\Phi D\Gamma^{i_2}\Phi\Gamma^{i_3}\Phi), (\bar{D}\Gamma^{i_1}\Phi DS\Gamma^{i_2}\Phi\Gamma^{i_3}\Phi), ({D}S\Gamma^{i_1}\Phi \bar{D}\Gamma^{i_2}\Phi\Gamma^{i_3}\Phi),$$
$$({D}\Gamma^{i_1}\Phi \bar{D}S\Gamma^{i_2}\Phi\Gamma^{i_3}\Phi), r^{-1}({D}S\Gamma^{i_1}\Phi {D}\Gamma^{i_2}\Phi\Gamma^{i_3}\Phi), r^{-1}({D}\Gamma^{i_1}\Phi {D}S\Gamma^{i_2}\Phi\Gamma^{i_3}\Phi).$$
are irrelevant here because $i_3\leq 13$ and we can thus control the last factor in the sup norm to reduce to the quadratic terms above. As before, we also have terms that do not have $S$ (from $S\Lambda$ or from the commutators $[D,S], [\bar{D},S]$), but they already appear in $N_k$ and we will use the estimates proved for $N_k$ in Proposition \ref{NO}. We first estimate the quadratic terms. The crucial technical point here is that we do not have an improved pointwise decay estimate for $\bar{D}S\Gamma^i\Phi$ because we have used $S$ in the proof of Proposition \ref{rv} and we are only commuting with $S$ once. Nevertheless, since $k\leq 13$, we can instead put $D\Gamma^i\Phi$ in $L^\infty$.
\begin{equation*}
\begin{split}
&\sum_{i_2=0}^{\lfloor\frac{k}{2}\rfloor}\sum_{i_1=0}^{k}\int_{\Sigma_\tau\cap\{r\geq\frac{9\tau}{10}\}}r^\alpha \left(\mbox{Quadratic Terms}\right)^2\\
\leq&C \left(\sup_{r\geq\frac{9\tau}{10}}\sum_{i_2=0}^{6} r^2|D\Gamma^{i_2}\Phi|^2\right)\sum_{i_1=0}^{k}\int_{\Sigma_\tau\cap\{r\geq\frac{9\tau}{10}\}}r^{\alpha-2} |\bar{D}S\Gamma^{i_1}\Phi |^2\\
&+C\left(\sup_{r\geq\frac{9\tau}{10}}\sum_{i_2=0}^{6} r^2|DS\Gamma^{i_2}\Phi|^2\right)\sum_{i_1=0}^{k}\int_{\Sigma_\tau\cap\{r\geq\frac{9\tau}{10}\}}r^{\alpha-2} |\bar{D}\Gamma^{i_1}\Phi|^2\\
&+C\left(\sup_{r\geq\frac{9\tau}{10}}\sum_{i_2=0}^{6} r^2|\bar{D}\Gamma^{i_2}\Phi|^2\right)\sum_{i_1=0}^{k}\int_{\Sigma_\tau\cap\{r\geq\frac{9\tau}{10}\}}r^{\alpha-2} |{D}S\Gamma^{i_1}\Phi|^2\\
&+C \left(\sup_{r\geq\frac{9\tau}{10}}\sum_{i_1=0}^{k} r^2|D\Gamma^{i_1}\Phi|^2\right)\sum_{i_2=0}^{6}\int_{\Sigma_\tau\cap\{r\geq\frac{9\tau}{10}\}}r^{\alpha-2} |\bar{D}S\Gamma^{i_2}\Phi |^2\\
&+C\tau^{-2}\left(\sup_{r\geq\frac{9\tau}{10}}\sum_{i_2=0}^{6} r^2|{D}\Gamma^{i_2}\Phi|^2\right)\sum_{i_1=0}^{k}\int_{\Sigma_\tau\cap\{r\geq\frac{9\tau}{10}\}}r^{\alpha-2} |{D}S\Gamma^{i_1}\Phi|^2\\
&+C\tau^{-2}\left(\sup_{r\geq\frac{9\tau}{10}}\sum_{i_2=0}^{6} r^2|{D}S\Gamma^{i_2}\Phi|^2\right)\sum_{i_1=0}^{k}\int_{\Sigma_\tau\cap\{r\geq\frac{9\tau}{10}\}}r^{\alpha-2} |{D}\Gamma^{i_1}\Phi|^2\\
\leq&CA\epsilon\sum_{i_1=0}^{k}\int_{\Sigma_\tau\cap\{r\geq\frac{9\tau}{10}\}}r^{\alpha-2} \left(|\bar{D}S\Gamma^{i_1}\Phi |^2+|\bar{D}\Gamma^{i_1}\Phi|^2\right)+CA\epsilon\tau^{-2+\eta_{14}}\sum_{i_1=0}^{k}\int_{\Sigma_\tau\cap\{r\geq\frac{9\tau}{10}\}}r^{\alpha-2} |{D}S\Gamma^{i_1}\Phi |^2\\
&+CA\epsilon\tau^{-4+\alpha+\eta_{S,11}}\left(\sup_{r\geq\frac{9\tau}{10}}\sum_{i_1=0}^{k} r^2|D\Gamma^{i_1}\Phi|^2\right)+CA\epsilon\tau^{-4+\alpha}
\end{split}
\end{equation*}
We then estimate the cubic terms:
\begin{equation*}
\begin{split}
&\sum_{i_1, i_2=0}^{\lfloor\frac{k}{2}\rfloor}\sum_{i_3=0}^{k}\int_{\Sigma_\tau\cap\{r\geq\frac{9\tau}{10}\}}r^\alpha \left(\mbox{Cubic Terms}\right)^2\\
\leq&C\left(\sup_{r\geq\frac{9\tau}{10}}\sum_{i_1=0}^6 \left(r^2\bar{D}\Gamma^{i_1}\Phi\right)^2\right)\left(\sup_{r\geq\frac{9\tau}{10}}\sum_{i_2=0}^6 r^2\left({D}\Gamma^{i_2}\Phi\right)^2\right)\sum_{i_3=0}^{k}\int_{\Sigma_\tau\cap\{r\geq\frac{9\tau}{10}\}}r^{\alpha-4}(S\Gamma^{i_3}\Phi)^2\\
&+C\tau^{-2}\left(\sup_{r\geq\frac{9\tau}{10}}\sum_{i_1=0}^6 r^2\left({D}\Gamma^{i_1}\Phi\right)^2\right)^2\sum_{i_3=0}^{k}\int_{\Sigma_\tau\cap\{r\geq\frac{9\tau}{10}\}}r^{\alpha-4}(S\Gamma^{i_3}\Phi)^2\\
\leq & CA^2\epsilon^2\tau^{-2+\eta_{14}}\sum_{i_3=0}^{k}\int_{\Sigma_\tau\cap\{r\geq\frac{9\tau}{10}\}}r^{\alpha-2}(DS\Gamma^{i_3}\Phi)^2,
\end{split}
\end{equation*}
which is better then the estimates obtained for the quadratic terms. We hence focus on the quadratic terms and spell out explicitly what the estimates amount to for different values of $k$ and $\alpha$. 
\begin{equation*}
\begin{split}
&\sum_{i_2=0}^{6}\sum_{i_1=0}^{13}\int_{\Sigma_\tau\cap\{r\geq\frac{9\tau}{10}\}} \left(\mbox{Quadratic Terms}\right)^2
\leq CA^2\epsilon^2\tau^{-2+\eta_{S,13}}+CA^2\epsilon^2\tau^{-4+\eta_{S,11}+\eta_{16}},
\end{split}
\end{equation*}
and
\begin{equation*}
\begin{split}
&\sum_{i_2=0}^{6}\sum_{i_1=0}^{12}\int_{\Sigma_\tau\cap\{r\geq\frac{9\tau}{10}\}} \left(\mbox{Quadratic Terms}\right)^2
\leq CA^2\epsilon^2\tau^{-3+\eta_{S,12}}+CA^2\epsilon^2\tau^{-4+\eta_{S,11}},
\end{split}
\end{equation*}
and
\begin{equation*}
\begin{split}
&\sum_{i_2=0}^{6}\sum_{i_1=0}^{11}\int_{\Sigma_\tau\cap\{r\geq\frac{9\tau}{10}\}} \left(\mbox{Quadratic Terms}\right)^2
\leq CA^2\epsilon^2\tau^{-4+\eta_{S,11}}+CA^2\epsilon^2\tau^{-4+\eta_{S,11}},
\end{split}
\end{equation*}
and
\begin{equation*}
\begin{split}
&\sum_{i_2=0}^{6}\sum_{i_1=0}^{12}\int_{\Sigma_\tau\cap\{r\geq\frac{9\tau}{10}\}} r^2\left(\mbox{Quadratic Terms}\right)^2
\leq CA^2\epsilon^2\left(\tau^{-1+\eta_{S,12}}+\tau^{-2+\eta_{14}}+\tau^{-2+\eta_{S,11}}\right),
\end{split}
\end{equation*}
and
\begin{equation*}
\begin{split}
&\sum_{i_2=0}^{6}\sum_{i_1=0}^{11}\int_{\Sigma_\tau\cap\{r\geq\frac{9\tau}{10}\}} r^2\left(\mbox{Quadratic Terms}\right)^2
\leq CA^2\epsilon^2\left(\tau^{-2+\eta_{S,11}}+\tau^{-2+\eta_{14}}\right).
\end{split}
\end{equation*}

\end{proof}

With the estimates for the inhomogeneous terms for the equations involving $S$, we can now retrieve the bootstrap assumptions involving $S$. We will follow the order that we proved the estimates without $S$, namely, first proving the pointwise estimates, then the integrated estimates, then the energy estimates and finally the energy estimates involving also $\hat{Y}$. Noticing that $U_{k,j}$ (respectively $N_k$) and $S(U_{k,j})$ (respectively $S(N_k)$) satisfy similar estimates (see Propositions \ref{NI}, \ref{NO}, \ref{U}, \ref{SNI}, \ref{SNO} and \ref{SU}), we would focus on showing that the estimates for $V_k$ are enough to close the bootstrap assumptions. We now prove the pointwise estimates and retrieve the bootstrap assumptions (\ref{BAP6}), (\ref{BAP7}), (\ref{BAPI3}) and (\ref{BAPI4}).

\begin{proposition}
For $r\geq\frac{t^*}{4}$,
\begin{equation}\label{P6}
\sum_{j=0}^8|DS\Gamma^j\Phi|^2\leq \frac{B_S}{2}A\epsilon r^{-2}.
\end{equation}
\begin{equation}\label{P7}
\sum_{j=0}^6|DS\Gamma^j\Phi|^2\leq \frac{B_S}{2}A\epsilon r^{-2}(t^*)^{\eta_{S,11}}(1+|u|)^{-2}.
\end{equation}
For $r\leq\frac{t^*}{4}$,
\begin{equation}\label{PI3}
\sum_{j=0}^{6}|S\Gamma^j\Phi|^2\leq \frac{B_S}{2}A\epsilon (t^*)^{-2+\eta_{15}}.
\end{equation}
\begin{equation}\label{PI4}
\sum_{\ell=1}^{7-j}\sum_{j=0}^{6}|D^\ell S\Gamma^j\Phi|^2\leq \frac{B_S}{2}A\epsilon r^{-2}(t^*)^{-2+\eta_{15}}.
\end{equation}
\end{proposition}
\begin{proof}
The proof of the estimates for $r\geq\frac{t^*}{4}$ (i.e. (\ref{P6}) and (\ref{P7})) is completely analogous to Proposition \ref{pointwise}, with the use of Propositions \ref{U}, \ref{NI}, \ref{NO}  replaced by Propositions \ref{V}, \ref{SU}, \ref{SNI}, \ref{SNO} appropriately. Notice especially that the estimates in Proposition \ref{V} for $V$ is better than that in Proposition \ref{SU} for $SU$ and are thus acceptable.

(\ref{PI3}) follows directly from Proposition \ref{SEinside} and the bootstrap assumptions (\ref{BA3}) and (\ref{BA7}). Here, we need to use also (\ref{BA3}) because we would need to commute $S$ with $\partial_{t^*}$ and would get terms that do not contain $S$. 

(\ref{PI4}) follows directly from Proposition \ref{SEDinside}, the bootstrap assumptions (\ref{BA3}) and (\ref{BA7}), as well as Proposition \ref{U}, \ref{NI}, \ref{V}, \ref{SU} and \ref{SNII} to control the inhomogeneous terms. As before, (\ref{BA3}) and Proposition \ref{U}, \ref{NI} are used to control the terms arising from $[S,\partial_{t^*}]$. Notice here that the decay rate for $\displaystyle\sum_{\ell=1}^{7-j}\sum_{j=0}^{6}|D^\ell S\Gamma^j\Phi|^2$ is not as good as that for $\displaystyle\sum_{\ell=1}^{9-j}\sum_{j=0}^{8}|D^\ell \Gamma^j\Phi|^2$ because in proving the decay rate for $\displaystyle\sum_{\ell=1}^{9-j}\sum_{j=0}^{8}|D^\ell \Gamma^j\Phi|^2$, we have used the quantities associated to $S\Phi$, while we do not have estimates for $S^2\Phi$ at our disposal.
\end{proof}

As before, once we have proved the $L^\infty$ bounds, we will replace the constant $B_S$ by $C$. 

\begin{proposition}
\begin{equation}\label{K7}
\begin{split}
 \sum_{i+j+k\leq 12}A_{S,X,j}^{-1}\iint_{\mathcal R((1.1)^{-1}\tau,\tau)\cap\{r\leq\frac{t^*}{2}\}}K^{X_0}\left(S\partial_{t^*}^{i}\tilde{\Omega}^j\Phi\right)\leq\frac{\epsilon}{2}\tau^{-1+\eta_{S,12}}.
\end{split}
\end{equation}
\begin{equation}\label{K8}
\sum_{i+j\leq 11}A_{S,X,j}^{-1}\iint_{\mathcal R((1.1)^{-1}\tau,\tau)\cap\{r\leq\frac{t^*}{2}\}} K^{X_1}\left(S\partial_{t^*}^{i}\tilde{\Omega}^j\Phi\right)\leq\frac{\epsilon}{2}\tau^{-1+\eta_{S,12}}.
\end{equation}
\begin{equation}\label{K9}
\begin{split}
 \sum_{i+j\leq 11}A_{S,X,j}^{-1}\iint_{\mathcal R((1.1)^{-1}\tau,\tau)\cap\{r\leq\frac{t^*}{2}\}}K^{X_0}\left(S\partial_{t^*}^{i}\tilde{\Omega}^j\Phi\right)\leq\frac{\epsilon}{2}\tau^{-2+\eta_{S,11}}.
\end{split}
\end{equation}
\begin{equation}\label{K10}
 \sum_{i+j\leq 10}A_{S,X,j}^{-1}\iint_{\mathcal R((1.1)^{-1}\tau,\tau)\cap\{r\leq\frac{t^*}{2}\}}K^{X_1}\left(S\partial_{t^*}^{i}\tilde{\Omega}^j\Phi\right)\leq\frac{\epsilon}{2}\tau^{-2+\eta_{S,11}}.
\end{equation}

\end{proposition}
\begin{proof}
 This follows exactly as Proposition \ref{K} except for replacing the use of Propositions \ref{U}, \ref{NI} and \ref{NO} with Propositions \ref{V}, \ref{SU}, \ref{SNI} and \ref{SNO}.
\end{proof}

\begin{proposition}
\begin{equation}\label{5}
  \sum_{i+j=13}A_{S,j}^{-1}\int_{\Sigma_\tau} J^{N}_\mu\left(S\partial_{t^*}^i\tilde{\Omega}^j\Phi\right) n^{\mu}_{\Sigma_\tau}\leq \frac{\epsilon}{4} \tau^{\eta_{S,13}}.
\end{equation}
\begin{equation}\label{8}
  A_{S,j}^{-1}\sum_{j=0}^{12}\int_{\Sigma_\tau} J^{N}_\mu\left(S\partial_{t^*}^i\tilde{\Omega}^j\Phi\right) n^{\mu}_{\Sigma_\tau}\leq \frac{\epsilon}{2} .
\end{equation}
\begin{equation}\label{K6}
\begin{split}
 \sum_{i+j=13}A_{S,X,j}^{-1}\iint_{\mathcal R(\tau_0,\tau)} K^{X_0}\left(S\partial_{t^*}^i\tilde{\Omega}^j\Phi\right)\leq \frac{\epsilon}{2}\tau^{\eta_{S,13}}.
\end{split}
\end{equation}
\begin{equation}\label{K6.5}
\begin{split}
 \sum_{i+j\leq 12}A_{S,X,j}^{-1}\iint_{\mathcal R(\tau_0,\tau)} K^{X_0}\left(S\partial_{t^*}^i\tilde{\Omega}^j\Phi\right)\leq \frac{\epsilon}{2}.
\end{split}
\end{equation}

\end{proposition}
\begin{proof}
 This follows exactly as Proposition \ref{J} except for replacing the use of Propositions \ref{U}, \ref{NI} and \ref{NO} with Propositions \ref{V}, \ref{SU}, \ref{SNI} and \ref{SNO}.
\end{proof}

\begin{proposition}
 \begin{equation}\label{6}
\begin{split}
 \sum_{i+j=12}A_{S,j}^{-1}\left(\int_{\Sigma_\tau} J^{Z+N,w^Z}_\mu\left(S\partial_{t^*}^i\tilde{\Omega}^j\Phi\right) n^{\mu}_{\Sigma_\tau} +C\tau^2\int_{\Sigma_\tau\cap\{r\leq \frac{9\tau}{10}\}} J^{N}_\mu\left(S\partial_{t^*}^i\tilde{\Omega}^j\Phi\right) n^{\mu}_{\Sigma_\tau}\right)\leq \frac{\epsilon}{4} \tau^{1+\eta_{S,12}}.
\end{split}
\end{equation}
\end{proposition}

\begin{proof}
 This follows exactly as Proposition \ref{Z1} except for replacing the use of Propositions \ref{U}, \ref{NI} and \ref{NO} with Propositions \ref{V}, \ref{SU}, \ref{SNI} and \ref{SNO}.
\end{proof}

\begin{proposition}
\begin{equation}\label{7}
\begin{split} 
\sum_{i+j\leq 11}A_{S,j}^{-1}\left(\int_{\Sigma_\tau} J^{Z+N,w^Z}_\mu\left(S\left(\partial_{t^*}^i\tilde{\Omega}^j\Phi\right)\right) n^{\mu}_{\Sigma_\tau} +C\tau^2\int_{\Sigma_\tau\cap\{r\leq \frac{9\tau}{10}\}} J^{N}_\mu\left(S\left(\partial_{t^*}^i\tilde{\Omega}^j\Phi\right)\right) n^{\mu}_{\Sigma_\tau}\right)\leq \frac{\epsilon}{4} \tau^{\eta_{S,11}}
\end{split}
\end{equation}
\end{proposition}

\begin{proof}
 This follows exactly as Proposition \ref{Z2} except for replacing the use of Propositions \ref{U}, \ref{NI} and \ref{NO} with Propositions \ref{V}, \ref{SU}, \ref{SNI} and \ref{SNO}.
\end{proof}

To close the bootstrap argument we need finally to consider energy quantities with both $S$ and $\hat{Y}$.

\begin{proposition}
$$\sum_{i+k=13}A_{S,Y}^{-1}\int_{\Sigma_\tau}J^{N}_\mu\left(\hat{Y}^kS\partial_{t^*}^i\Phi\right) n^{\mu}_{\Sigma_\tau}\leq \frac{\epsilon}{4} \tau^{\eta_{S,13}}.$$
$$\sum_{i+k=12}A_{S,Y}^{-1}\tau^2\int_{\Sigma_\tau\cap\{r\leq r^-_Y\}}J^{N}_\mu\left(\hat{Y}^kS\partial_{t^*}^i\Phi\right) n^{\mu}_{\Sigma_\tau}\leq \frac{\epsilon}{4} \tau^{1+\eta_{S,12}}$$
$$\sum_{i+k\leq 11}A_{S,Y}^{-1}\tau^2\int_{\Sigma_\tau\cap\{r\leq r^-_Y\}}J^{N}_\mu\left(\hat{Y}^kS\partial_{t^*}^i\Phi\right) n^{\mu}_{\Sigma_\tau}\leq \frac{\epsilon}{4} \tau^{\eta_{S,11}}$$
\end{proposition}
\begin{proof}
This follows exactly as Proposition \ref{Yest} except for replacing the use of Propositions \ref{NI} with Propositions \ref{V} and \ref{SNI}.
\end{proof}

\section{Proof of Theorem 1}\label{pfmaintheorem}

Now all the bootstrap assumptions are closed and all the estimates hold. The solution hence exists globally by a standard local existence argument that we omit here. The decay estimates of the the derivatives of $\Phi$ claimed in the Theorem are restatements of (\ref{P3}), (\ref{P4}), (\ref{BAP2}). The decay estimates follow from the use of Proposition \ref{rnoderivatives} and (\ref{Z2}) for $r\geq R$ and Proposition \ref{inside2} and (\ref{K5}) for $r\leq\frac{t^*}{4}$.

\section{Acknowledgment}
The author thanks his advisor Igor Rodnianski for his continual support and encouragement and for many enlightening discussions.

\bibliographystyle{hplain}
\bibliography{Null}

\begin{thebibliography}{10}

\bibitem{AB}
L.~Andersson and P.~Blue.
\newblock Hidden symmetries and decay for the wave equation on the kerr
  spacetime.
\newblock 2009, 0908.2265.

\bibitem{BSt}
P.~Blue and J.~Sterbenz.
\newblock {Uniform decay of local energy and the semi-linear wave equation on
  Schwarzschild space}.
\newblock {\em Comm. Math. Phys.}, 268(2):{481--504}, 2006,
  arXiv:math.AP/0510315.

\bibitem{CG}
D.~Catania and V.~Georgiev.
\newblock {Blow Up for the Semilinear Wave Equation in {S}chwarzschild Metric}.
\newblock 2004, 0407001.

\bibitem{Chr}
D.~Christodoulou.
\newblock Global solutions of nonlinear hyperbolic equations for small initial
  data.
\newblock {\em Comm. Pure Appl. Math.}, 39(2):267--282, 1986.

\bibitem{DRNL}
M.~Dafermos and I.~Rodnianski.
\newblock Small-amplitude nonlinear waves on a black hole background.
\newblock {\em J. Math. Pures Appl. (9)}, 84(9):1147--1172, 2005.

\bibitem{DRK}
M.~Dafermos and I.~Rodnianski.
\newblock {A proof of the uniform boundedness of solutions to the wave equation
  on slowly rotating Kerr backgrounds}.
\newblock 2008, arXiv:gr-qc/0805.4309.

\bibitem{DRL}
M.~Dafermos and I.~Rodnianski.
\newblock {Lectures on black holes and linear waves}.
\newblock 2008, arXiv:gr-qc/0811.0354.

\bibitem{DRNPS}
M.~Dafermos and I.~Rodnianski.
\newblock {A new physical-space approach to decay for the wave equation with
  applications to black hole spacetimes}.
\newblock 2009, 0910.4957.

\bibitem{DRS}
M.~Dafermos and I.~Rodnianski.
\newblock {The red-shift effect and radiation decay on black hole spacetimes}.
\newblock {\em {Comm. Pure Appl. Math.}}, 2009, arXiv:gr-qc/0512119.

\bibitem{FKSY}
F.~Finster, N.~Kamran, J.~Smoller, and S.~T. Yau.
\newblock Decay of solutions of the wave equation in the {K}err geometry.
\newblock {\em Comm. Math. Phys.}, 264(2):465--503, 2006.

\bibitem{FKSY2}
F.~Finster, N.~Kamran, J.~Smoller, and S.~T. Yau.
\newblock Erratum: ``{D}ecay of solutions of the wave equation in the {K}err
  geometry'' [{C}omm. {M}ath. {P}hys. {\bf 264} (2006), no. 2, 465--503; ].
\newblock {\em Comm. Math. Phys.}, 280(2):563--573, 2008.

\bibitem{HW}
J.~B. Hartle and D.~C. Wilkins.
\newblock Analytic properties of the {T}eukolsky equation.
\newblock {\em Comm. Math. Phys.}, 38:47--63, 1974.

\bibitem{MMTT}
{J. Marzuola, J. Metcalfe, D. Tataru, M. Tohaneanu}.
\newblock {Strichartz estimates on Schwarzschild black hole backgrounds}.
\newblock 2008, arXiv:math.AP/0802.3942.

\bibitem{John}
F.~John.
\newblock Blow-up for quasilinear wave equations in three space dimensions.
\newblock {\em Comm. Pure Appl. Math.}, 34(1):29--51, 1981.

\bibitem{JoK}
F.~John and S.~Klainerman.
\newblock Almost global existence to nonlinear wave equations in three space
  dimensions.
\newblock {\em Comm. Pure Appl. Math.}, 37(4):443--455, 1984.

\bibitem{Kl}
S.~Klainerman.
\newblock Uniform decay estimates and the {L}orentz invariance of the classical
  wave equation.
\newblock {\em Comm. Pure Appl. Math.}, 38(3):321--332, 1985.

\bibitem{Knull}
S.~Klainerman.
\newblock The null condition and global existence to nonlinear wave equations.
\newblock In {\em Nonlinear systems of partial differential equations in
  applied mathematics, {P}art 1 ({S}anta {F}e, {N}.{M}., 1984)}, volume~23 of
  {\em Lectures in Appl. Math.}, pages 293--326. Amer. Math. Soc., Providence,
  RI, 1986.

\bibitem{KS}
S.~Klainerman and T.~Sideris.
\newblock On almost global existence for nonrelativistic wave equations in
  {$3$}{D}.
\newblock {\em Comm. Pure Appl. Math.}, 49(3):307--321, 1996.

\bibitem{LS}
J.~Luk.
\newblock {Improved decay for solutions to the linear wave equation on a
  Schwarzschild black hole}.
\newblock 2009, 0906.5588.

\bibitem{LKerr}
Jonathan Luk.
\newblock {A vector field method approach to improved decay for solutions to
  the wave equation on a slowly rotating Kerr black hole}.
\newblock 2010, 1009.0671.

\bibitem{MNS}
J.~Metcalfe, M.~Nakamura, and C.~D. Sogge.
\newblock Global existence of solutions to multiple speed systems of
  quasilinear wave equations in exterior domains.
\newblock {\em Forum Math.}, 17(1):133--168, 2005.

\bibitem{MSo}
J.~Metcalfe and C.~D. Sogge.
\newblock Global existence of null-form wave equations in exterior domains.
\newblock {\em Math. Z.}, 256(3):521--549, 2007.

\bibitem{NiS}
J.-P. Nicolas.
\newblock Nonlinear {K}lein-{G}ordon equation on {S}chwarzschild-like metrics.
\newblock {\em J. Math. Pures Appl. (9)}, 74(1):35--58, 1995.

\bibitem{NiK}
J.-P. Nicolas.
\newblock A nonlinear {K}lein-{G}ordon equation on {K}err metrics.
\newblock {\em J. Math. Pures Appl. (9)}, 81(9):885--914, 2002.

\bibitem{PT}
W.~H. Press and S.~A. Teukolsky.
\newblock Pertubations of a rotating black hole ii.
\newblock {\em The Astrophysical Journal}, 185:649--673, 1973.

\bibitem{ST}
T.~C. Sideris and S.-Y. Tu.
\newblock Global existence for systems of nonlinear wave equations in 3{D} with
  multiple speeds.
\newblock {\em SIAM J. Math. Anal.}, 33(2):477--488 (electronic), 2001.

\bibitem{Sogge}
C.~D. Sogge.
\newblock Global existence for nonlinear wave equations with multiple speeds.
\newblock In {\em Harmonic analysis at {M}ount {H}olyoke ({S}outh {H}adley,
  {MA}, 2001)}, volume 320 of {\em Contemp. Math.}, pages 353--366. Amer. Math.
  Soc., Providence, RI, 2003.

\bibitem{StR}
J.~Sterbenz.
\newblock Angular regularity and {S}trichartz estimates for the wave equation.
\newblock {\em Int. Math. Res. Not.}, (4):187--231, 2005.
\newblock With an appendix by Igor Rodnianski.

\bibitem{Ta}
D.~Tataru.
\newblock Local decay of waves on asymptotically flat stationary space-times.
\newblock 2009, 0910.5290.

\bibitem{TT}
D.~Tataru and M.~Tohaneanu.
\newblock Local energy estimate on {K}err black hole backgrounds.
\newblock 2008, 0810.5766.

\bibitem{To}
M.~Tohaneanu.
\newblock Strichartz estimates on {K}err black hole backgrounds.
\newblock 2009, 0910.1545.

\bibitem{Wh}
B.~F. Whiting.
\newblock Mode stability of the {K}err black hole.
\newblock {\em J. Math. Phys.}, 30(6):1301--1305, 1989.

\end{thebibliography}
\end{document}